\documentclass[draftclsnofoot,onecolumn]{IEEEtran}
\def\3To1BC{$3-$to$-1$}

\def\define{:{=}~}
\def\naturals{\mathbb{N}}
\def\reals{\mathbb{R}}

\def\UM{\mathscr{U}\!\mathcal{M}-}
\def\PCC{\mathscr{P}\mathcal{C}\mathcal{C}-}

\def\underlinem{\underline{m}}
\def\underlineM{\underline{M}}
\def\underlineY{\underline{Y}}
\def\underliney{\underline{y}}
\def\underlined{\underline{d}}
\def\underlineR{\underline{R}}
\def\fieldpi{\mathcal{F}_{\pi}}
\def\bin{i}
\def\Var{\mbox{Var}}
\def\underlinem{\underline{m}}
\def\tildeQRV{\tilde{\TimeSharingRV}}
\def\tildeQRVSet{\tilde{\TimeSharingRVSet}}
\def\underlineSetY{\underline{\OutputAlphabet}}
\def\underlineCardinalityMessageSet{\underline{\mathcal{M}}}
\def\MessageSetM{\mathcal{M}}
\def\underlineMessageSetM{\underline{\MessageSetM}}
\def\hatm{\hat{m}}
\def\cl{\mbox{cl}}
\def\cocl{\mbox{cocl}}

\def\SetOfDistributions{\mathbb{D}}
\def\AuxiliaryAlphabet{\mathcal{U}}
\def\TimeSharingRVSet{\mathcal{Q}}
\def\TimeSharingRV{Q}
\def\tildeq{\tilde{q}}
\def\PublicRVSet{\mathcal{W}}
\def\PrivateRVSet{\mathcal{V}}
\def\SemiPrivateRVSet{\mathcal{U}}
\def\underlineSemiPrivateRV{\underline{\SemiPrivateRV}}
\def\underlinePrivateRV{\underline{\PrivateRV}}
\def\underlineSemiPrivateRVSet{\underline{\SemiPrivateRVSet}}
\def\underlinePrivateRVSet{\underline{\PrivateRVSet}}
\def\InputRV{X}
\def\PublicRV{W}
\def\PrivateRV{V}
\def\SemiPrivateRV{U}
\def\underlinePrivateRV{\underline{\PrivateRV}}
\def\Expectation{\mathbb{E}}
\def\fieldq{\mathcal{F}_{q}}

\def\GPV{\tilde{V}}
\def\GPX{\tilde{X}}
\def\GPS{\tilde{S}}
\def\GPY{\tilde{Y}}
\def\GPSetV{\tilde{\mathcal{V}}}

\def\ulineInputRV{\underline{X}}
\def\ulineOutputRV{\underline{Y}}
\newcommand{\msout}[1]{\text{\sout{\ensuremath{#1}}}}
\def\fieldpii{\mathcal{F}_{\pi_{i}}}
\def\fieldpik{\mathcal{F}_{\pi_{k}}}
\def\fieldpij{\mathcal{F}_{\pi_{j}}}
\newif\ifProofForORDBC

\usepackage{amssymb}
\usepackage{amsmath}
\usepackage{mathrsfs}
\usepackage{ulem}
\usepackage{epsf,epsfig}
\usepackage{cite}

\newcommand{\BinaryField}{\mathbb{F}_{2}}
\newcommand{\InputAlphabet}{\mathcal{X}}
\newcommand{\OutputAlphabet}{\mathcal{Y}}
\newcommand{\StateAlphabet}{\mathcal{S}}

\newcommand{\comment}[1]{}
\begin{document}
\sloppy
\newtheorem{remark}{\it Remark}
\newtheorem{thm}{Theorem}
\newtheorem{corollary}{Corollary}
\newtheorem{definition}{Definition}
\newtheorem{lemma}{Lemma}
\newtheorem{example}{Example}
\newtheorem{prop}{Proposition}

\title{Achievable rate region for three user discrete broadcast channel based on coset codes}

\author {Arun~Padakandla and S. Sandeep~Pradhan,
~\IEEEmembership{Member,~IEEE}%
\thanks{Arun Padakandla and S. Sandeep Pradhan are with the Department of Electrical
and Computer Engineering, University of Michigan, Ann Arbor
48109-2122, USA. This work was supported by NSF grant CCF-1116021.}}
\maketitle
\vspace{-0.5in}
\begin{abstract}
We consider the problem of developing coding techniques and deriving
achievable rate regions for discrete memoryless broadcast channels
with $3$ receivers ($3-$DBC). We begin by identifying a novel vector
additive $3-$DBC for which we characterize an upper bound on the
the largest  achievable rate region based on unstructured codes, henceforth referred to as
$\UM$region. We propose a coding technique based on coset codes that
yield an achievable rate triple not contained within $\UM$region. We generalize the proposed
coding technique using a new ensemble of codes - \textit{partitioned
  coset codes} (PCC) - containing both empirical and algebraic
properties, and evaluate it's performance to derive an achievable rate region for the general
$3-$DBC. The new elements in this derivation are binning and joint typicality encoding and decoding of statistically correlated PCCs. We validate the
utility of this technique by identifying \textit{non-additive} instances of
$3-$DBC for which the proposed coding techniques based on PCC yield
strictly larger rates.

\end{abstract}
\section{Introduction}
\label{Sec:Introduction}
The problem of characterizing the capacity region of a general
broadcast channel
(BC) was proposed by Cover \cite{197201TIT_Cov} in 1972, and he
introduced a novel coding technique to derive achievable rate regions
for particular degraded BCs. In a seminal work aimed at deriving an
achievable rate region for the general degraded BC, 
Bergmans \cite{197303TIT_Ber} generalized Cover's technique into what
is currently referred to as superposition coding. Gallager
\cite{197403PPI_Gal} and Bergmans \cite{197403TIT_Ber} concurrently
and independently proved optimality of 
superposition coding for the class of degraded BCs. This in particular
yielded capacity region for 
the scalar additive Gaussian BC. However, the case of general discrete
BC (DBC) remained open. This  led to the discovery of another
ingenious coding technique by Gelfand
\cite{197707PIT_Gel}. In 1979, Marton \cite{197905TIT_Mar} generalized
Gelfand's technique \cite{197707PIT_Gel} into what is currently
referred to as \textit{binning}. In conjunction with superposition,
she derived the largest known achievable rate region
\cite{197905TIT_Mar} for the general two user DBC ($2-$DBC). A
generalization \cite[p.391 Problem 10(c)]{CK-IT2011} of superposition
and binning to incorporate a common message yields \textit{Marton's
  rate region}, the current known largest achievable rate region for
the general $2-$DBC and it’s capacity is yet unknown.\footnote{It is
  of interest to note that though superposition and binning were known
  in particular settings \cite{197201TIT_Cov}, \cite{197707PIT_Gel},
  it’s generalization led to fundamentally new ideas. 
}

Though the capacity region has been found for many interesting classes
of BCs  \cite{197201TIT_Cov,197303TIT_Ber,197403PPI_Gal,197403TIT_Ber,197707PIT_Gel,197905TIT_Mar,CK-IT2011,197507TIT_Cov,
  197511TIT_AhlKor,197701TIT_KorMar,197903TIT_Elg,  
1977MMISIT_Mar,197804PPI_Pin,1980MMPPI_GelPin,198001TIT_Elg,198101TIT_GamMeu,197503TIT_Meu,197901TIT_HajPur,201006ISIT_GenNaiShaWan,200910TIT_NaiElg},
the question of whether the techniques of superposition and binning,
in conjunction, is optimal for the general DBC has remained open. 
Gohari and Anantharam
\cite{201202TIT_GohAna} have proved computability of Marton's rate
region. This enabled them identify a class of binary $2-$DBCs for
which Marton's rate region when computed is strictly smaller than the
tightest known outer bound \cite{197805TIT_Sat,200701TIT_NaiGam},
which is due to Nair 
and El Gamal. On the other hand, Weingarten, Steinberg and Shamai
\cite{200609TIT_WeiSteSha} have proved Marton's binning (also referred
to, in the Gaussian setting, as Costa's dirty paper coding
\cite{198305TIT_Cos}) to be optimal for Gaussian MIMO BC with
quadratic cost constraints and  arbitrary number of receivers, and thereby
characterized the capacity region. $3-$DBC with degraded message sets
has been studied in \cite{200910TIT_NaiElg}.

In this article, we begin by characterizing an achievable rate
region, referred to as $\UM$region, for the
  general $3-$DBC incorporating all current known coding techniques,
  i.e., message-splitting, superposition and binning of
  unstructured codes.  We identify a novel additive $3-$DBC (example
\ref{Ex:3-BCExample}) for which we propose a
  technique based on linear codes that yields an achievable rate triple not contained within
$\UM$region. We remark that even within the larger class of BCs that include
continuous valued alphabets, any number of receivers and multiple
antennae, we have, thus far, been unaware of any BC for which the $\UM$region can be strictly improved upon. 
One of the key elements of our work is an analytical proof of
sub-optimality of $\UM$region for this $3-$DBC. 

Motivated by the above findings,  we propose a general coding technique
based on a new ensemble of codes endowed with algebraic
structure- \textit{partitioned coset codes} \cite{201108ISIT_PadPra}
(PCC). We analyze the proposed coding technique and derive 
an achievable rate region\footnote{In general this region neither
  subsumes nor is subsumed by the $\UM$region.}- referred to as
$\PCC$region- for the general $3-$DBC expressed
in terms 
of single-letter information quantities. This region is a continuous function of the channel transition probability matrix. One of the key elements of
this analysis is an interplay of joint typical encoding and 
decoding of statistically correlated algebraic codebooks resulting in new proof
techniques. We identify a non-additive $3-$DBC  (example \ref{Ex:A3To1-OR-BC}) for which we
analytically prove the existence of rate triples that belong to  $\PCC$region
but lie outside the $\UM$region.  Finally, we indicate a way to combine the two coding
techniques that enables one to derive an achievable rate region that
includes the $\UM$region. 

Why do codes endowed with algebraic structure outperform traditional
independent unstructured codes for a BC? The central aspect of a
coding technique designed for a BC is interference
management. Marton's coding incorporates two techniques -
superposition and binning - for tackling interference. Superposition
enables each user decode a \textit{univariate} component of the other
user's signal and thus subtract it off. Binning enables the encoder
counter the component of each user's interfering signal not decoded by
the other, by precoding for the same. Except for particular cases, the
most popular being dirty paper coding, precoding results in a rate
loss, and is therefore less efficient than decoding the interfering
signal at the decoder. The presence of a rate loss motivates each
decoder to decode as large a part of interference as
possible.\footnote{For the Gaussian case, there is no rate loss. Thus
  the encoder can precode all the interference. Indeed, the optimal
  strategy does not require any user to decode a part of signal not
  intended for  it. Thus constraining interference patterns is
  superfluous. This explains why lattices are not necessary to achieve capacity
of Gaussian vector BC.} However decoding a large part of the
interference constrains the individual rates. In a three user BC, each
user's reception is plagued by interference caused by signals intended
for the other two users. The interference is in general a bivariate
function of signals intended for the other users. If the signals of
the two users are endowed with a structure that can help compress the
range of this bivariate function when applied to all possible
signals, then the receivers can decode a larger part of the
interfering signal. This minimizes the component of the interference
precoded, and therefore the rate loss.
This is where codebooks endowed with algebraic
structure outperform unstructured independent codebooks. Indeed,
linear codes constrain the interference pattern to an affine subspace
if the interference is the sum of user 2 and 3's signals. 

As evidenced by the non-additive example (example
\ref{Ex:A3To1-OR-BC}), linear codes provide gain even when the
bivariate function is not a field addition. Furthermore,
we have considered a natural
generalization of linear codes to sets with looser algebraic structure such as groups. Our
investigation of group codes  to improve achievable rate
regions for information theoretic problems has been pursued in concurrent research
threads \cite{201108ISIT_SahPra}. Containing the sum
of transmitted codewords using linear codes is just the first step,
and we envision an achievable rate region involving a union over all
relevant algebraic objects. 

Related Works: The use of structured codes for improving information
theoretic rate regions began with the ingenious technique of K\"orner
and Marton \cite{197903TIT_KorMar}, proposed for the source coding
problem of reconstructing modulo$-2$ sum of distributed binary
sources.  Ahlswede and Han \cite[Section VI]{198305TIT_AhlHan}
proposed a universal coding technique that brings together coding
techniques based on unstructured and structured codes.
More recently, there is a wider interest
\cite{200710TIT_NazGas,200808TIT_CadJaf,200912arXiv_CadJaf} in developing coding
techniques for particular problem instances that perform better than
unstructured codes. In
\cite{200906TIT_PhiZam} nested linear codes 
are employed to communicate over a particular binary doubly dirty multiple access
channel (MAC). The use of structured codes for interference channels (referred
to as interference alignment) toward
improved achievable rate region has been addressed in several works
\cite{201009TIT_BreParTse,200809Allerton_SriJafVisJafSha,201408TIT_HonCai,
201407TIT_KriJaf,201207ISIT_PadSahPra}.

It was shown in \cite{198701TIT_HanKob}, in the setting of distributed
source coding that for any non-trivial bivariate function,
there exists at least
one source distribution for which linear codes outperform random
codes.  However, linear codes were known to be
suboptimal for arbitrary point-to-point (PTP)
communication \cite{197102TAMS_Ahl}, and therefore, the basic
building block in the coding scheme for any multi-terminal
communication problem could not be filled by linear codes. 
The ensemble of nested coset codes was proposed in
\cite{201103TIT_KriPra} as the basic building block of algebraic codes
for distributed lossy compression of general sources subject to arbitrary
distortion criterion. 

This article is organized as follows. We begin with definitions in
section \ref{Sec:BroadcastChannelDefinitionsMartonRateRegion}. 
In section \ref{SubSec:NaturalExtensionOfMartonTo3BC}, we
present the $\UM$ achievable region for $3-$DBC. 
Section \ref{Sec:3To1BCAndTheNeedForStructuredCodes} contains our first main
finding - identification of a vector additive $3-$DBC for which the
$\UM$technique is proved to be strictly sub-optimal. In section
\ref{Sec:AchievableRateRegionsFor3To1BCUsingNestedCosetCodes} we present our second main finding - characterization of $\PCC$region for $3-$DBC - in three pedagogical steps. In section
\ref{Sec:EnlargingMarton'sRateRegionUsingNestedCosetCodes}, we
indicate how to glue together $\UM$technique and the technique based on PCC for general $3-$DBC. 
We conclude in section \ref{Sec:ConcludingRemarks} by
pointing to fundamental connections between several layers of coding
in a three user communication problem and common information of a
triple of random variables. 

\section{Broadcast channel: definitions and Marton's rate region}
\label{Sec:BroadcastChannelDefinitionsMartonRateRegion}


\subsection{Notation}
\label{SubSec:Notation}
We employ notation that has now been widely accepted in the information theory literature
supplemented with the following. The empty sum has value $0$, i.e,
$\sum_{a \in \phi}=0$. For a set $A \subseteq \reals^{k}$, $\cocl
\left( A \right)$ denotes closure of convex hull of $A$. Throughout
this article, $\log$ and $\exp$ functions are taken with respect to
the base $2$. Let
$h_{b}(x) \define -x\log_{2}x - (1-x)\log_{2}(1-x)$ denote binary
entropy function. Let $a * b \define a(1-b)+(1-a)b$ denote binary
convolution. For $K \in \naturals$, we let $[K]\define \left\{
1,2\cdots,K \right\}$. We let $\fieldq$ denote the finite field of
cardinality $q$. While $+$ denotes addition in $\reals$, we let
$\oplus$ 
denote addition in a finite field. The particular finite field,
which is uniquely determined (up-to an isomorphism) by it's cardinality, is clear from
context. When ambiguous, or to enhance clarity, we specify addition in $\fieldq$ using
$\oplus_{q}$. For elements $a,b$, in a finite field, $a \ominus b
\define a \oplus (-b)$, where $(-b)$ is the additive inverse of $b$.In this article, we will need to define multiple objects, mostly triples, of the same type. In order to reduce clutter, we use an \underline{underline} to
denote aggregates of objects of similar type. For example, (i) if
$\OutputAlphabet_{1},\OutputAlphabet_{2}, \OutputAlphabet_{3}$ denote (finite) sets, we
let $\underlineSetY$ either denote the Cartesian product
$\OutputAlphabet_{1} \times \OutputAlphabet_{2} \times \OutputAlphabet_{3}$ or
abbreviate the collection $(\OutputAlphabet_{1},\OutputAlphabet_{2},\OutputAlphabet_{3})$
of sets, the particular reference being clear from context, (ii) if $y_{k} \in \OutputAlphabet_{k}:k=1,2,3$, we let $\underliney \in
\underlineSetY$ abbreviate $(y_{1},y_{2},y_{3}) \in \OutputAlphabet$ (iii) if
$d_{k}:\OutputAlphabet_{k}^{n} \rightarrow \MessageSetM_{k}:k=1,2,3$ denote (decoding)
maps, then we let $\underlined(\underliney^{n})$ denote $(d_{1}(y_{1}^{n}),
d_{2}(y_{2}^{n}), d_{3}(y_{3}^{n}))$.
\subsection{Definitions: Broadcast channel, code, achievability and capacity}
\label{SubSec:DefinitionsBroadcastChannelCodeAchievabilityCapacity}
A $3-$DBC consists of a finite input alphabet
set $\InputAlphabet$ and three finite output alphabet sets $\OutputAlphabet_{1},
\OutputAlphabet_{2}, \OutputAlphabet_{3}$. The discrete time channel is (i) time
invariant, i.e., the probability mass function (PMF) of $\underlineY_{t}=(Y_{1t},Y_{2t},Y_{3t})$, the output at time
$t$, conditioned on $X_{t}$, the input at time $t$, is invariant with $t$, (ii)
memoryless, i.e., conditioned on present input $X_{t}$, the present output $\underlineY_{t}$ is
independent of past inputs $X_{1},\cdots, X_{t-1}$, past outputs $\underline{Y}_{1},\underline{Y}_{2},\cdots,\underline{Y}_{t-1}$, and (iii) used without feedback, i.e.,
the encoder has no information of the symbols received by the decoder. Let
$W_{\underlineY|X}(\underliney|x)=W_{Y_{1}Y_{2}Y_{3}|X}(y_{1},y_{2},y_{3}|x)$ denote
probability of observing $\underliney \in \underlineSetY$ at the respective outputs
conditioned on $x \in \InputAlphabet$ being input. Input is constrained with respect
to a cost function $\kappa : \InputAlphabet \rightarrow [0, \infty)$. The cost function
is assumed additive, i.e., cost of transmitting the vector $x^{n} \in \InputAlphabet^{n}$
is $\bar{\kappa}^{n}(x^{n}) \define \sum_{i=1}^{n}\kappa(x_{i})$. We refer to this $3-$DBC as
$(\InputAlphabet,\underlineSetY,W_{\underlineY|X},\kappa)$.
In this article, we restrict attention to communicating private
messages to the three users. The focus of this article therefore is
the (private message) capacity region of a $3-$DBC, and in particular
corresponding achievable rate regions. The following definitions make
the relevant notions precise. 

\begin{definition}
 \label{Defn:3-BCCode}
 A $3-$DBC code $(n,\underlineMessageSetM,e,\underlined)$ consist of (i) finite index
sets $\MessageSetM_{1}, \MessageSetM_{2},\MessageSetM_{3}$ of messages, (ii) encoder map
$e:\underlineMessageSetM \rightarrow
\InputAlphabet^{n}$, and (iii) three decoder maps $d_{k} : \OutputAlphabet_{k}^{n}
\rightarrow \MessageSetM_{k}:k=1,2,3$.
\end{definition}

\begin{definition}
 \label{Defn:BCCodeErrorProbability}
 The error probability of a $3-$DBC code
$(n,\underlineCardinalityMessageSet,e,\underline{d})$ conditioned on message triple
$(m_{1},m_{2},m_{3}) \in \underlineMessageSetM$ is
\begin{equation}
 \label{Eqn:ErrorProbabilityOf3-BCCode}
 \xi(e,\underlined|\underlinem) \define 1-\sum_{\underliney^{n} : \underlined
(\underliney^{n}) = \underlinem}
W_{\underlineY|X}(\underliney^{n}|e(\underlinem)).\nonumber
\end{equation}
The average error probability of
a $3-$DBC code $(n,\underlineCardinalityMessageSet,e,\underlined)$ is
$\bar{\xi}(e,\underlined)\define \sum_{\underlinem
 \in \underlineMessageSetM}
\frac{1}{|\mathcal{M}_{1}||\mathcal{M}_{2}||\mathcal{M}_{3}|}
\xi(e,\underlined|\underlinem)$. Cost of transmitting message $\underlinem \in
\underlineMessageSetM$ per symbol is $\tau(e|\underlinem) \define
\frac{1}{n}\bar{\kappa}^{n}(e(\underlinem))$ and
average cost of $3-$DBC code $(n,\underlineCardinalityMessageSet,e,\underlined)$ is $\tau(e)
\define \frac{1}{|\mathcal{M}_{1}||\mathcal{M}_{2}||\mathcal{M}_{3}|}\sum_{\underlinem
 \in \underlineMessageSetM}\tau(e|\underlinem)$.
\end{definition}

\begin{definition}
\label{Defn:3-BCAchievabilityAndCapacity}
A rate-cost quadruple $(R_{1},R_{2},R_{3},\tau)\in [0,\infty)^{4}$ is
achievable if for every $\eta > 0$, there exists $N(\eta)\in \naturals$ such that for all
$n > N(\eta)$, there exists a $3-$DBC code $(n, \underlineCardinalityMessageSet^{(n)},
e^{(n)},\underlined^{(n)})$ such that (i)
$\frac{\log_{2}|\mathcal{M}_{k}^{(n)}|}{n} \geq R_{k}-\eta:k=1,2,3$, (ii)
$\bar{\xi}(e^{(n)},\underlined^{(n)})
\leq \eta$, and (iii) average cost $\tau(e^{(n)}) \leq \tau+\eta$. The capacity region
$\mathbb{C}(W_{\underlineY|X},\kappa,\tau)$ ($\mathbb{C}(\tau)$ for
short)  is defined as $\cl{\left\{ \underlineR \in \reals^{3}:
(\underlineR,\tau)\mbox{ is achievable} \right\}}$.
\end{definition}

In some cases, we consider projections of the capacity
region. For any $3-$DBC, if receivers $2$ and $3$ can
simultaneously achieve their respective capacities, then 
$\mathbb{C}_1(\tau)$ is defined as the maximum rate achieved
by receiver 1. Otherwise $\mathbb{C}_1(\tau)=0$. The currently known largest achievable rate region, $\UM$ region,  for
$3-$DBC is obtained via message-splitting, superposition and binning
of unstructured codes. 

\subsection{Marton's rate region}
\label{SubSec:MartonRateRegionFor2BC}
Marton's coding for $2-$DBC incorporates two fundamental techniques -
superposition and precoding - accomplished using a two layer coding
scheme. First layer, which is public, contains a codebook over
$\PublicRVSet$. Second layer is private and contains two codebooks one
each on $\PrivateRVSet_{1}$ and $\PrivateRVSet_{2}$. Precoding is
accomplished by setting aside a \textit{bin} of codewords for each
private message, thus enabling the encoder to choose a compatible pair
of codewords in the indexed bins. User $j$th message is split into two
parts - public and private. The public parts together index a codeword
in $\PublicRVSet-$codebook and the private part of user $j$th message
index a codeword in $\PrivateRVSet_{j}-$codebook. Both users decode
from the public codebook and their respective private codebooks. 
Definition \ref{Defn:MartonTestChannels} and theorem \ref{Thm:MartonRateRegionFor2-BC}
provide a characterization of rate pairs achievable using Marton's
coding technique for $2-$DBC.
We omit restating the definitions analogous to definitions \ref{Defn:3-BCCode},
\ref{Defn:BCCodeErrorProbability}, \ref{Defn:3-BCAchievabilityAndCapacity} for a $2-$DBC.

\begin{definition}
 \label{Defn:MartonTestChannels}
Let $\SetOfDistributions_{M}(W_{\underlineY|X},\kappa,\tau)$ denote the collection of distributions
$p_{\TimeSharingRV\PublicRV \PrivateRV_{1} \PrivateRV_{2}XY_{1}Y_{2}}$ defined on $\TimeSharingRVSet \times \PublicRVSet \times
\PrivateRVSet_{1} \times \PrivateRVSet_{2} \times \InputAlphabet \times
\OutputAlphabet_{1} \times \OutputAlphabet_{2}$, where (i) $\TimeSharingRVSet$, $\PublicRVSet$,
$\PrivateRVSet_{1}$ and $\PrivateRVSet_{2}$ are finite sets of cardinality at most
$|\InputAlphabet|+4$, $|\InputAlphabet|+4$, $|\InputAlphabet|+1$ and $|\InputAlphabet|+1$ respectively, (ii)
$p_{\underlineY|X \underlinePrivateRV \PublicRV \TimeSharingRV}=p_{\underlineY|X} = W_{\underlineY|X}$,
(iii) $\Expectation\left\{\kappa(X)\right\} \leq \tau$. For $p_{\TimeSharingRV\PublicRV
\underlinePrivateRV X \underlineY} \in \SetOfDistributions_{M}(W_{\underlineY|X},\kappa,\tau)$,
let $\alpha_{M}(p_{\TimeSharingRV\PublicRV
\underlinePrivateRV X \underlineY})$ denote the set of $(R_{1},R_{2}) \in \reals^{2}$ that satisfy
\begin{flalign}
 \label{Eqn:MartonAchievableRateRegionForParticularTestChannel}
0 \leq R_{k} &\leq I(\PublicRV
\PrivateRV_{k};Y_{k}|\TimeSharingRV):k=1,2,\nonumber\\ 
R_{1}+R_{2} &\leq\! \min\left\{ I(W;Y_{1}|\TimeSharingRV),I(W;Y_{2}|\TimeSharingRV)  \right\}+I(\PrivateRV_{1};Y_{1}|\TimeSharingRV W )+
I(\PrivateRV_{2};Y_{2}|\PublicRV,\TimeSharingRV)\!-\!I(\PrivateRV_{1};\PrivateRV_{2}|\PublicRV,\TimeSharingRV) \nonumber
\end{flalign}
and
\begin{equation}
 \label{Eqn:MartonRateRegionAsAUnionOfTestChannelRateRegions}
 \alpha_{M}(W_{\underlineY|X},\kappa,\tau) = \cocl\left(\underset{\substack{p_{\TimeSharingRV\PublicRV
\underlinePrivateRV X \underlineY} \\\in \SetOfDistributions_{M}(W_{\underlineY|X},\kappa,\tau)}
}{\bigcup}\alpha_{M}(p_{\TimeSharingRV\PublicRV
\underlinePrivateRV X \underlineY})\right)\nonumber
\end{equation}
\end{definition}

\begin{thm}
 \label{Thm:MartonRateRegionFor2-BC}
For $2-$DBC $(\InputAlphabet,\underlineSetY,W_{\underlineY|X},\kappa)$,
$\alpha(W_{\underlineY|X},\kappa,\tau)$ is achievable, i.e., $\alpha(W_{\underlineY|X},\kappa,\tau) \subseteq \mathbb{C}(W_{\underlineY|X},\kappa,\tau)$.
\end{thm}
\begin{remark}
 \label{Rem:CardinalityBoundsDerivedByGohariAndAnantharam}
 The bounds on cardinality of $\PublicRVSet,\PrivateRVSet_{1}$ and $\PrivateRVSet_{2}$
were derived by Gohari and Anantharam in \cite{201202TIT_GohAna}.
\end{remark}
We refer the reader to \cite{197905TIT_Mar} for a proof of achievability. El Gamal and Meulen
\cite{198101TIT_GamMeu} provide a simplified proof using the method of second moment.\subsection{$\mathscr{U}\!\mathcal{M}-$region : Current known largest achievable rate region for $3-$DBC}
\label{SubSec:NaturalExtensionOfMartonTo3BC}
The $\mathscr{U}\!\mathcal{M}-$technique is a 3 layer coding
technique. For simplicity, we describe the coding technique without
referring to the time sharing random variable and employ the same in
characterizing $\mathscr{U}\!\mathcal{M}-$region. User $j$th message
$M_{j}$ is split into four parts - two semi-private parts, and one,
private and public parts each. We let message (i) $M^{W}_{j} \in
\MessageSetM^{W}_{j}$ of rate $K_{j}$ denote it's public part (ii)
$M^{U}_{i\msout{j}} \in
\MessageSetM^{U}_{i\msout{j}},M^{U}_{\msout{j}k} \in
\MessageSetM^{U}_{\msout{j}k}$ of rates $L_{ij},K_{jk}$ respectively,
denote it's semi-private parts, where $(i,j,k)$ is an appropriate
triple in $\left\{ (1,2,3),(2,3,1),(3,1,2) \right\}$, and (iii)
$M^{V}_{j} \in \MessageSetM^{V}_{j}$ of rate $T_{j}$ denote it's
private part. 
The first layer is public with a single codebook
$(w^{n}(\underline{m}^{W}):\underline{m}^{W} \in
\underlineMessageSetM^{W})$ of rate $K_{1}+K_{2}+K_{3}$ over
$\PublicRVSet$. $\underline{M}^{W} \define
(M_{1}^{W},M_{2}^{W},M_{3}^{W})$ indexes a codeword in
$\PublicRVSet-$codebook and each user decodes from 
$\PublicRVSet-$codebook. 

Each codeword in $\PublicRVSet-$codebook is linked to a triple of codebooks - one each on $\SemiPrivateRVSet_{ij}:(i,j) \in \left\{ (1,2),(2,3),(3,1) \right\}$- in the second layer. The second layer is semi-private. Each of the three semi-private codebooks is composed of \textit{bins}, wherein each bin comprises a collection of codewords. For each pair $(i,j) \in \left\{ (1,2),(2,3),(3,1) \right\}$ the following hold. $M^{U}_{i\msout{j}}$ and $M^{U}_{\msout{i}j}$ together index a bin in $\SemiPrivateRVSet_{ij}-$codebook. Each bin in $\SemiPrivateRVSet_{ij}-$codebook is of rate $S_{ij}$. Let $(u_{ij}^{n}(\underlinem^{W},m_{\msout{i}j}^{U},m_{i\msout{j}}^{U},s_{ij}):s_{ij} \in [\exp\{ nS_{ij} \}])$ denote the bin corresponding to semi-private messages $\underline{m^{U}_{ij}}\define (m_{\msout{i}j}^{U},m_{i\msout{j}}^{U}) $ in the $\SemiPrivateRVSet_{ij}-$codebook linked to public message $\underlinem^{W}$. Users $i,j$ decode from $\SemiPrivateRVSet_{ij}-$codebook and it maybe verified that $\SemiPrivateRVSet_{ij}-
$codebook is of rate $K_{ij}+L_{ij}+S_{ij}$.

Let $(i,j)$ and $(j,k)$ be distinct pairs in $\left\{ (1,2),(2,3),(3,1) \right\}$. Every pair of codewords in $\SemiPrivateRVSet_{ij}-$ and $\SemiPrivateRVSet_{jk}-$codebooks is linked to a codebook on $\PrivateRVSet_{j}$. The codebooks over $\PrivateRVSet_{j}:j=1,2,3$ comprise the third layer which is private. $M_{j}^{V}$ indexes a bin in $\PrivateRVSet_{j}-$codebook, each of which is of rate $S_{j}$, and thus $\PrivateRVSet_{j}-$codebook is of rate $T_{j}+S_{j}$. Let $(v_{j}^{n}(\underlinem^{W},\underline{m^{U}_{ij}},s_{ij},\underline{m^{U}_{jk}},s_{jk},m_{j}^{V},s_{j}):s_{j} \in [\exp\{ nS_{j} \}])$ denote bin corresponding to private message $m_{j}^{V}$ in the $\PrivateRVSet_{j}-$codebook linked to codeword pair $(u_{ij}^{n}(\underlinem^{W},\underline{m^{U}_{ij}},s_{ij}),u_{jk}^{n}(\underlinem^{W},\underline{m^{U}_{jk}},s_{jk}))$. User $j$ decodes from the private codebook over $\PrivateRVSet_{j}$.
How does the encoder map messages to a codeword?
Let $p_{W\underlineSemiPrivateRV \underlinePrivateRV \InputRV}$ be a distribution on
$\PublicRVSet \times \underlineSemiPrivateRVSet \times \underlinePrivateRVSet \times \InputAlphabet$ such that $\Expectation\left\{ \kappa(\InputRV) \right\} \leq \tau$. The encoder looks for $(s_{12},s_{23},s_{31},s_{1},s_{2},s_{3})$ such that the septuple
\begin{eqnarray}
 \label{Eqn:JointlyTypicalSeptupleOfCodewords}
\left(\substack{w^{n}(\underlineM^{W}),u_{ij}^{n}(\underlineM^{W},\underline{M^{U}_{ij}},s_{ij}):(i,j)=(1,2),(2,3),(3,1),\\v_{j}^{n}(\underlineM^{W},\underline{M^{U}_{ij}},s_{ij},\underline{M^{U}_{jk}},s_{jk},M_{j}^{V},s_{j}):(i,j,k)=(1,2,3),(2,3,1),(3,1,2)}\right)\nonumber
\end{eqnarray}
of codewords is jointly typical with respect to $p_{W\underlineSemiPrivateRV \underlinePrivateRV}$. If such a septuple is found, this is mapped to a codeword on $\InputAlphabet^{n}$ which is input to the channel. If it does not find any such septuple, an error is declared. 

Decoder $j$ looks for all quadruples $(\underline{\hatm}^{W},\underline{\hatm_{ij}}^{U},\underline{\hatm_{jk}}^{U},\hatm_{j}^{V})$ such that
\begin{equation}
 \label{Eqn:DecodingRuleOfDecoderj}
\left(   w^{n}(\underline{\hatm}^{W}),u_{ij}^{n}(\underline{\hatm}^{W},\underline{\hatm^{U}_{ij}},s_{ij}),u_{jk}^{n}(\underline{\hatm}^{W},\underline{\hatm^{U}_{jk}},s_{jk}),v_{j}^{n}(\underlinem^{W},\underline{m^{U}_{ij}},s_{ij},\underline{m^{U}_{jk}},s_{jk},m_{j}^{V},s_{j}),Y_{j}^{n} \right) \nonumber
\end{equation}
is jointly typical with respect to $p_{W\underlineSemiPrivateRV
  \underlinePrivateRV \InputRV\underlineY}\define
p_{W\underlineSemiPrivateRV \underlinePrivateRV
  \InputRV}W_{\underline{Y}|X}$ for some $(s_{ij},s_{jk},s_{j})$,
where (i) $(i,j,k)$ is the appropriate triple in 
$\{ (1,2,3), (2,3,1), (3,1,2) \}$ and (ii) $Y_{j}^{n}$ is the received vector. If there is a unique such quadruple, it declares $\hatm_{j} \define (\hatm_{j}^{W},\hatm_{i\msout{j}}^{U},\hatm_{\msout{j}k}^{U},\hatm^{V}_{j})$ as user $j$th message. Otherwise, i.e., none or more than one such quadruple is found, it declares an error.

We incorporate the time sharing random variable, average the
error probability over the ensemble of codebooks, and 
provide upper bounds on the same using the second moment method
\cite{198101TIT_GamMeu}.  Let $\TimeSharingRV$, taking values over the finite
alphabet $\TimeSharingRVSet$, denote the time sharing random
variable. Let $p_{\TimeSharingRV}$ be a PMF on $\TimeSharingRVSet$ and
$q^{n} \in \TimeSharingRVSet^{n}$ denote a sequence picked according
to $p^{n}_{Q}$. $q^{n}$ is revealed to the encoder and all
decoders. The codewords in $\PublicRVSet-$codebook are identically and
independently distributed according to
$p^{n}_{W|Q}(\cdot|q^{n})$. Conditioned on entire public codebook
$(W^{n}(\underline{m}^{W})=w^{n}(\underline{m}^{W}):\underline{m}^{W}
\in \underlineMessageSetM^{W})$ and the time sharing sequence $q^{n}$,
each of the codewords
$U_{ij}^{n}(\underlinem^{W},\underline{m^{U}_{ij}},s_{ij}):(\underline{m^{U}_{ij}},s_{ij})
\in \underline{\MessageSetM_{ij}}^{U} \times [\exp\{ nS_{ij} \}]$ are
independent and identically distributed according to
$p_{U_{ij}|WQ}^{n}(\cdot|w^{n}(\underline{m}^{W}),q^{n})$. Conditioned
on a realization of the entire collection of public and semi-private
codebooks, the private codewords 
$(V_{j}^{n}(\underlinem^{W},\underline{m^{U}_{ij}},s_{ij},\underline{m^{U}_{jk}},s_{jk},m_{j}^{V},s_{j}):s_{j} \in [\exp\{ nS_{j} \}])$ are independent and identically distributed according to
\begin{eqnarray}
 \label{Eqn:DistributionOfVCodebookConditionedOnPublicAndSemiPrivate}
p_{V_{j}|U_{ij}U_{jk}WQ}^{n}\left(\cdot|w^{n}(\underline{m}^{W}),u_{ij}^{n}(\underlinem^{W},\underline{m^{U}_{ij}},s_{ij}),u_{jk}^{n}(\underlinem^{W},\underline{m^{U}_{jk}},s_{jk}),q^{n}\right).\nonumber
\end{eqnarray}
The probability of the error event at the encoder decays exponentially with $n$ if for each triple $(i,j,k) \in \{ (1,2,3),(2,3,1),(3,1,2)\} $
\begin{eqnarray}
 \label{Eqn:3BCSourceCodingBoundNonnegativity}
S_{i}\!\!\!\!&>&\!\!\!\! 0\\
\label{Eqn:3BCSourceCodingPairwiseBound}
S_{ij}+S_{jk}\!\!\!\! &>& \!\!\!\!I(U_{ij};U_{jk}|WQ) \\
\label{Eqn:3BCSourceCodingTripleBound}
S_{ij}+S_{jk}+S_{ki} \!\!\!\!&>&\!\!\!\! I(U_{ij};U_{jk};U_{ki}|WQ)\footnotemark\\
\label{Eqn:3BCSourceCodingQuadrapleBound}
S_{i}+S_{ij}+S_{jk}+S_{ki}\!\!\!\!&>&\!\!\!\! I(U_{ij};U_{jk} ;U_{ki}|WQ)
 +I(V_{i};U_{jk}|U_{ij},U_{ki},WQ)\\
S_{i}+S_{j}+S_{ij}+S_{jk}+S_{ki} \!\!\!\!&>& \!\!\!\!I(V_{i};U_{jk}|U_{ij},U_{ki},WQ)+I(V_{j};U_{ki}|U_{ij},U_{jk},WQ)\nonumber\\
\label{Eqn:3BCSourceCodingPentaBound}
&&\!\!\!\!+I(U_{ij};U_{jk};U_{ki}|WQ)+I(V_{i};V_{j}|U_{jk},U_{ij},U_{ki},WQ)\\
S_{1}+S_{2}+S_{3}+S_{12}+S_{23}+S_{31} \!\!\!\!&>& \!\!\!\!I(V_{1};U_{23}|U_{12},U_{31},WQ)+I(V_{2};U_{31}|U_{12},U_{23},WQ)+I(V_{1};V_{2};V_{3}|\TimeSharingRV \PublicRV \underlineSemiPrivateRV)\nonumber\\
\label{Eqn:3BCSourceCodingSextupleBound}
&&\!\!\!\!+I(U_{12};U_{23};U_{31}|WQ)+I(V_{3};U_{12}|U_{23},U_{31},WQ).
\end{eqnarray}
\footnotetext{For three random variables, $A,B,C$, we have $I(A;B;C)=I(A;B)+I(AB;C)$.}
The probability of decoder error event decays exponentially if for each triple $(i,j,k) \in \{ (1,2,3),(2,3,1),(3,1,2)\} $
\begin{flalign}
\label{Eqn:3BCChannelCodingSingleBound}
I(V_i;Y_i|QWU_{ij}U_{ki}) &> T_i+S_i  \\
\label{Eqn:3BCChannelCodingDoubleBound}
I(U_{ij}V_i;Y_i|QWU_{ki})+I(U_{ij};U_{ki}|QW)  &>
K_{ij}+L_{ij}+S_{ij}+T_i+S_i \\
\label{Eqn:3BCChannelCodingDoubleSecondBound}
I(U_{ki}V_i;Y_i|QWU_{ij})+I(U_{ij};U_{ki}|QW)  &>
K_{ki}+L_{ki}+S_{ki}+T_i+S_i \\
\label{Eqn:3BCChannelCodingTripleBound}
I(U_{ij}U_{ki}V_i;Y_i|QW)+I(U_{ij};U_{ki}|QW) &>
K_{ij}+L_{ij}+S_{ij}+K_{ki}+L_{ki}+S_{ki}+T_i+S_i \\ 
\label{Eqn:3BCChannelCodingSextupleBound}
I(WU_{ij}U_{ki}V_i;Y_i|Q)+I(U_{ij};U_{ki}|QW) &>
K_i+K_j+K_k+K_{ij}+L_{ij}+S_{ij}+K_{ki}+L_{ki}+S_{ki}+T_i+S_i
\end{flalign}

For each PMF $p_{\TimeSharingRV W\underlineSemiPrivateRV
  \underlinePrivateRV \InputRV}W_{\underlineY|\InputRV}$ defined on
$\TimeSharingRVSet \times \PublicRVSet \times
\underlineSemiPrivateRVSet \times \underlinePrivateRVSet \times
\InputAlphabet \times \underline{\OutputAlphabet}$, let
$\alpha_{\mathscr{U}}(p_{QW\underlineSemiPrivateRV \underlinePrivateRV
  \InputRV\underlineY})$ denote the set of all triples
$(R_{1},R_{2},R_{3}) \in [0,\infty)^{4}$ such that (i) there exists
  non-negative real numbers $K_{ij},L_{ij},S_{ij},K_{j},T_{j},S_{j}$
  that satisfies
  (\ref{Eqn:3BCSourceCodingBoundNonnegativity})-(\ref{Eqn:3BCChannelCodingSextupleBound})
  for each pair $(i,j) \in \left\{ (1,2),(2,3),(3,1)  \right\}$ and
  (ii) $R_{j}=T_{j}+K_{jk}+L_{ij}+K_{j}$ for each triple $(i,j,k) \in
  \left\{ (1,2,3),(2,3,1),(3,1,2)\right\} $. The $\UM$region is 
\begin{equation}
 \label{Eqn:3BCMartonRateRegionAsAUnionOfTestChannelRateRegions}
 \alpha_{\mathscr{U}}(W_{\underlineY|X},\kappa,\tau) = \cocl\left(\underset{\substack{p_{\TimeSharingRV\PublicRV \underlineSemiPrivateRV
\underlinePrivateRV X \underlineY} \\\in \SetOfDistributions_{\mathscr{U}}(W_{\underlineY|X},\kappa,\tau)}
}{\bigcup}\alpha_{\mathscr{U}}(p_{\TimeSharingRV\PublicRV\underlineSemiPrivateRV
\underlinePrivateRV X \underlineY})\right),\nonumber
\end{equation}
where $\SetOfDistributions_{\mathscr{U}}(W_{\underlineY|X},\kappa,\tau)$ denote the collection of distributions
$p_{\TimeSharingRV\PublicRV \underlineSemiPrivateRV \underlinePrivateRV \InputRV \underlineY}$ defined on $\TimeSharingRVSet \times \PublicRVSet \times
\underlineSemiPrivateRVSet \times \underlinePrivateRVSet \times \InputAlphabet \times
\underlineSetY$, where (i) $\TimeSharingRVSet,\PublicRVSet,
\underlineSemiPrivateRVSet, \underlinePrivateRVSet$ are finite sets, (ii)
$p_{\underlineY|X \underlinePrivateRV \underlineSemiPrivateRV \PublicRV\TimeSharingRV}=p_{\underlineY|X} = W_{\underlineY|X}$,
(iii) $\Expectation\left\{\kappa(X)\right\} \leq \tau$.

\begin{thm}
\label{Thm:MartonRateRegionFor3-BC}
For $3-$DBC $(\InputAlphabet,\underlineSetY,W_{\underlineY|X},\kappa)$,
$\UM$region $\alpha_{\mathscr{U}}(W_{\underlineY|X},\kappa,\tau)$ is achievable, i.e.,
$\alpha_{\mathscr{U}}(W_{\underlineY|X},\kappa,\tau) \subseteq
\mathbb{C}(W_{\underlineY|X},\kappa,\tau)$. 
\end{thm}
\section{Strict sub-optimality of $\UM$technique}
\label{Sec:3To1BCAndTheNeedForStructuredCodes}

In this section, we present our first main finding - strict
sub-optimality of $\UM$technique. In particular, we identify a vector
additive $3-$DBC (example \ref{Ex:3-BCExample}) and propose a linear
coding technique for the same. In section
\ref{Sec:StrictSubOptimalityOfMartonCodingTechnique}, we prove strict
sub-optimality of $\UM$technique for this vector additive $3-$DBC. 

\begin{example}
 \label{Ex:3-BCExample}
Consider the $3-$DBC depicted in figure \ref{Fig:TheThreeUserBroadcastChannel}. Let the input alphabet $\InputAlphabet = \InputAlphabet_{1} \times \InputAlphabet_{2} \times \InputAlphabet_{3}$ be a triple Cartesian product of the binary field $\InputAlphabet_{1}=\InputAlphabet_{2} = \InputAlphabet_{3} = \BinaryField$ and the output alphabets $\OutputAlphabet_{1}=\OutputAlphabet_{2}=\OutputAlphabet_{3}=\BinaryField$ be binary fields. If $X=X_{1}X_{2}X_{3}$ denote the three binary digits input to the channel, then the outputs are $Y_{1}=X_{1}\oplus X_{2}\oplus X_{3} \oplus N_{1}$, $Y_{2} =X_{2} \oplus N_{2}$ and $Y_{3} = X_{3} \oplus N_{3}$, where (i) $N_{1}, N_{2}, N_{3}$ are independent binary random variables with $P(N_{j}=1)=\delta_{j} \in (0,\frac{1}{2})$ and (ii) $(N_{1},N_{2},N_{3})$ is independent of the input $X$. The binary digit $X_{1}$ is constrained to an average Hamming weight of $\tau \in (0,\frac{1}{2})$. In other words, $\kappa (x_{1}x_{2}x_{3}) = 1_{\left\{x_{1}=1\right\}}$ and the average cost 
of input is constrained to $\tau\in (0,\frac{1}{2})$. For the sake of clarity, we provide a formal description of this channel in terms of section \ref{SubSec:DefinitionsBroadcastChannelCodeAchievabilityCapacity}. This $3-$DBC maybe referred to as $(\InputAlphabet,\underlineSetY,W_{\underlineY|X},\kappa)$ where $\InputAlphabet \define \left\{ 0,1 \right\}\times \left\{ 0,1 \right\} \times \left\{ 0,1 \right\}, \OutputAlphabet_{1}=\OutputAlphabet_{2}=\OutputAlphabet_{3}=\left\{ 0,1 \right\}, W_{\underlineY|X}(y_{1},y_{2},y_{3}|x_{1}x_{2}x_{3})=BSC_{\delta_{1}}(y_{1}|x_{1}\oplus x_{2}\oplus x_{3})BSC_{\delta_{2}}(y_{2}|x_{2})BSC_{\delta_{3}}(y_{3}|x_{3})$, where $\delta_{j} \in (0,\frac{1}{2}):j=1,2,3$, $BSC_{\eta}(1|0)=BSC_{\eta}(0|1)=1-BSC_{\eta}(0|0)=1-BSC_{\eta}(1|1)=\eta$ for any $\eta \in (0,\frac{1}{2})$ and the cost function $\kappa (x_{1}x_{2}x_{3}) = 1_{\left\{x_{1}=1\right\}}$.
\end{example}
\begin{figure}
\centering
\includegraphics[height=1.8in,width=1.9in]{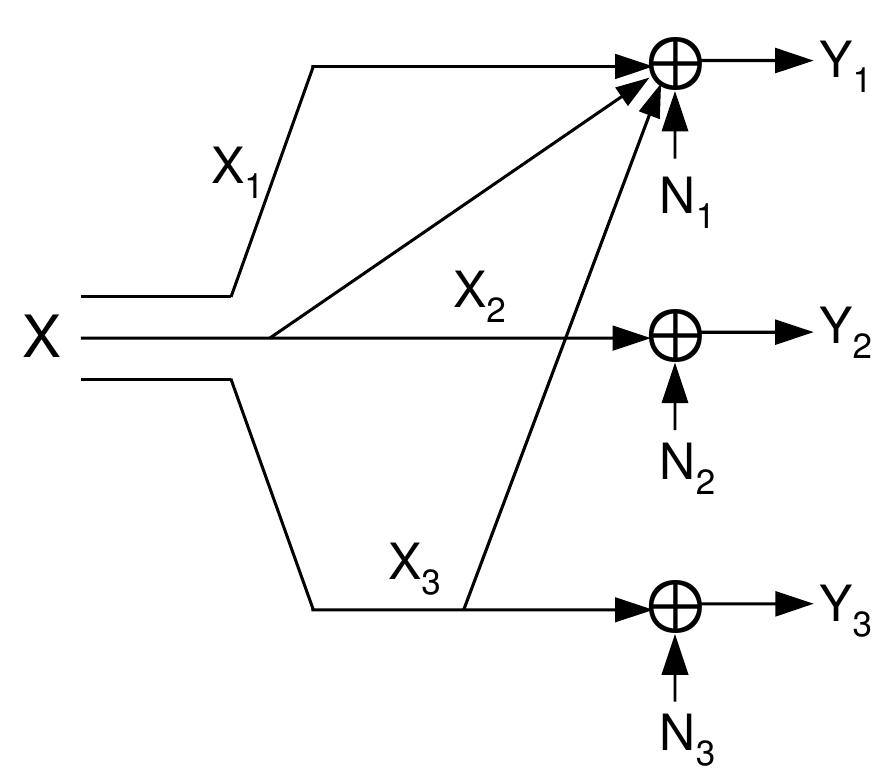}
\caption{A $3-$DBC with  octonary input and  binary outputs described in example \ref{Ex:3-BCExample}.}
\label{Fig:TheThreeUserBroadcastChannel}
\end{figure}
We begin with some observations for the above channel. Users $2$ and $3$ see \textit{interference free point-to-point} (PTP) links from the input. It is therefore possible to communicate to them simultaneously at their PTP capacities using any PTP channel codes achieving their respective capacities. For the purpose of this discussion, let us assume $\delta \define \delta_{2}=\delta_{3}$. This enables us to employ the same capacity achieving code of rate $1-h_{b}(\delta)$ for both users $2$ and $3$. What about user $1$? Three observations are in order. Firstly, if users $2$ and $3$ are being fed at their respective PTP capacities, then information can be pumped to user $1$ only through the first binary digit, henceforth referred to as $X_{1}$. In this case, we recognize that the sum of user $2$ and $3$'s transmissions interferes at receiver $1$. Thirdly, the first binary digit $X_{1}$ is costed, and therefore cannot cancel the interference caused by users $2$ and $3$ at the transmitters.

Since average Hamming weight of $X_{1}$ is restricted to $\tau$, $X_{1}\oplus N_{1}$ is restricted to an average Hamming weight of $\tau * \delta_{1}$. If the rates of users $2$ and $3$ are sufficiently small, receiver $1$ can attempt to decode codewords transmitted to users $2$ and $3$, cancel the interference and decode the desired codeword. This will require $2-2h_{b}(\delta) \leq 1- h_{b}(\delta_{1} * \tau)$ or equivalently $\frac{1+h_{b}(\delta_{1} * \tau)}{2} \leq h_{b}(\delta)$. What if this were not the case?

In the case $\frac{1+h_{b}(\delta_{1} * \tau)}{2} > h_{b}(\delta)$, we
are left with two choices. The first choice is to enable decoder $1$
to decode as large a part of the interference as possible and precode
for the rest of the uncertainty.\footnote{Since $X_{1}$ is costed,
  precoding results in a rate loss, i.e., in terms of rate achieved,
  the technique of precoding is in general inferior to the technique
  of decoding interference. This motivates a preference for decoding
  the interference as against to precoding. } The second choice is to
attempt decoding the sum of user $2$ and $3$'s codewords, instead of
the pair. In the sequel, we pursue the second choice using linear
codes. In section
\ref{Sec:StrictSubOptimalityOfMartonCodingTechnique}, we prove
$\UM$technique is forced to take the first choice which results  
in it's sub-optimality.

Since linear codes achieve the capacity of binary symmetric channels, there exists a single linear code, or a coset thereof, of rate $1-h_{b}(\delta)$ that achieves capacity of both user $2$ and $3$ channels. Let us employ this linear code for communicating to users $2$ and $3$. The code being linear or affine, the collection of sums of all possible pairs of codewords is restricted to a coset of rate $1-h_{b}(\delta)$. This suggests that decoder $1$ decode the sum of user $2$ and $3$ codewords. Indeed, if $1-h_{b}(\delta) \leq 1-h_{b}(\tau * \delta_{1})$, or equivalently $\tau*\delta_{1} \leq \delta$, then user $1$ can first decode the interference, peel it off, and then go on to decode the desired signal. Under this case, a rate $h_{b}(\tau*\delta_{1})-h_{b}(\delta_{1})$ is achievable for user $1$ even while communicating independent information at rate $1-h_{b}(\delta)$ for both users $2$ and $3$. We have therefore proposed a coding technique based on linear codes that achieves the rate triple $(h_{b}(\tau*\delta_{1})-h_{b}(\delta_{1}),1-h_{b}(\delta),1-h_{b}(\delta))$ if $\tau*\delta_{1} \leq \delta=\delta_{2}=\delta_{3}$. 

Let us now consider the general case with respect to $\delta_{2}, \delta_{3}$. Without loss of generality we may assume $\delta_{2} \leq \delta_{3}$. We employ a capacity achieving linear code to communicate to user $2$. This code is sub sampled (uniformly and randomly) to yield a capacity achieving code for user $3$. This construction ensures the sum of all pairs of user $2$ and $3$ codewords to lie within user $2$'s linear code, or a coset thereof, of rate $1-h_{b}(\delta_{2})$. If $1-h_{b}(\delta_{2}) \leq 1-h_{b}(\tau*\delta_{1})$, or equivalently $\tau*\delta_{1} \leq \delta_{2}$, then decoder $1$ can decode the sum of user $2$ and $3$'s codewords, i.e., the interfering signal, peel it off and decode the desired message at rate $h_{b}(\tau*\delta_{1})-h_{b}(\delta_{1})$. The above arguments are summarized in the following lemma.

\begin{lemma}
 \label{Lem:RateTripleAchievableUsingLinearCodes}
Consider the vector additive $3-$DBC in example
\ref{Ex:3-BCExample}. If $\tau*\delta_{1} \leq
\min\left\{\delta_{2},\delta_{3}\right\}$, then 
$(h_{b}(\tau*\delta_{1})-h_{b}(\delta_{1}),1-h_{b}(\delta_{2}),1-h_{b}(\delta_{3}))
\in \mathbb{C}(\tau)$. Moreover
$\mathbb{C}_1(\tau)=h_{b}(\tau*\delta_{1})-h_{b}(\delta_{1})$. 
\end{lemma}
In the above discussion, we have argued
$(h_{b}(\tau*\delta_{1})-h_{b}(\delta_{1}),1-h_{b}(\delta_{2}),1-h_{b}(\delta_{3}))
\in \mathbb{C}(\tau)$ for the vector additive $3-$DBC in example
\ref{Ex:3-BCExample}. It can be easily argued that $\mathbb{C}_1(\tau)
\leq h_{b}(\tau*\delta_{1})-h_{b}(\delta_{1})$,
and in conjunction with the former statement, the proof of lemma
\ref{Lem:RateTripleAchievableUsingLinearCodes} is complete. 

We now state the conditions under which
$(h_{b}(\tau*\delta_{1})-h_{b}(\delta_{1}),1-h_{b}(\delta_{2}),1-h_{b}(\delta_{3}))
\notin \alpha_{\mathscr{U}}(\tau)$. In particular, we show below in
Theorem
\ref{Thm:ConditionsUnderWhichMartonCodingTechniqueDoesNotAchieveRateTriple}
that if
$1+h_{b}(\delta_{1} * \tau) > h_{b}(\delta_{2})+h_{b}(\delta_{3})$,
then
$(h_{b}(\tau*\delta_{1})-h_{b}(\delta_{1}),1-h_{b}(\delta_{2}),1-h_{b}(\delta_{3}))
\notin \alpha_{\mathscr{U}}(\tau)$. We therefore conclude that if
$\tau,\delta_{1},\delta_{2},\delta_{3}$ are such that
$1+h_{b}(\delta_{1} * \tau) > h_{b}(\delta_{2})+h_{b}(\delta_{3})$ and
$\min\left\{ \delta_{2},\delta_{3}\right\}\geq \delta_{1} * \tau$,
then $\UM$technique is strictly suboptimal for the $3-$DBC presented
in example \ref{Ex:3-BCExample}. 
We prove the theorem in
section \ref{Sec:StrictSubOptimalityOfMartonCodingTechnique}. 

\begin{thm}
 \label{Thm:ConditionsUnderWhichMartonCodingTechniqueDoesNotAchieveRateTriple}
Consider the $3-$DBC in example \ref{Ex:3-BCExample}. If $h_b(\delta_{2})+h_{b}(\delta_{3}) < 1+h_b(\delta_{1}*\tau)
$, then $(h_{b}(\tau*\delta_{1})-h_{b}(\delta_{1}),1-h_{b}(\delta_{2}),1-h_{b}(\delta_{3})) \notin \alpha_{\mathscr{U}}(\tau)$.
\end{thm}

\begin{corollary}
\label{Cor:ConditionsUnderWhichMartonCodingTechniqueIsStrictlySubOptimal}
Consider the $3-$DBC in example \ref{Ex:3-BCExample} with $\delta=\delta_{2}=\delta_{3}$. If $h_{b}(\tau*\delta_{1}) \leq h_b(\delta) < \frac{1+h_b(\delta_{1}*\tau)}{2}$, then $(h_{b}(\tau*\delta_{1})-h_{b}(\delta_{1}),1-h_{b}(\delta),1-h_{b}(\delta)) \notin \alpha_{\mathscr{U}}(\tau)$ but $(h_{b}(\tau*\delta_{1})-h_{b}(\delta_{1}),1-h_{b}(\delta),1-h_{b}(\delta)) \in \mathbb{C}(\tau)$ and thus $\alpha_{\mathscr{U}}(\tau) \neq \mathbb{C}(\tau)$. In particular, if $\delta_{1}=0.01$ and $\delta \in (0.1325,0.21)$, then $\alpha_{\mathscr{U}}(\frac{1}{8}) \neq \mathbb{C}(\frac{1}{8})$.
\end{corollary}

\section{Achievable rate regions for $3-$DBC using partitioned coset codes}
\label{Sec:AchievableRateRegionsFor3To1BCUsingNestedCosetCodes}

In this section we present our second main finding - a new coding technique based on PCC for communicating over an arbitrary $3-$DBC - that enables us to derive $\PCC$region, a new achievable rate region for $3-$DBC. We present this in three pedagogical steps. Step I, presented in section \ref{SubSec:DecodingSumOfCodewordsUsingNestedCosetCodes}, describes all the new elements of our framework in a simple setting. In particular, we employ PCC to manage interference seen by one receiver, and derive a corresponding achievable rate region. For this step, we also provide a complete proof of achievability. Step II (section \ref{SubSec:AnEnlargedAchievableRateRegion}) builds on step I by incorporating private codebooks. Finally in step III (section \ref{Sec:AchievableRateRegionFor3BCUsingNestedCosetCodes}), we employ PCC to manage interference seen by all receivers, and thereby derive $\PCC$region.

\subsection{Step I: Using PCC to manage interference seen by a single receiver}
\label{SubSec:DecodingSumOfCodewordsUsingNestedCosetCodes}

\subsubsection{Description of the coding technique}
\label{SubSubSec:DescriptionOfSingleLayerCodingTechnique}
The essential aspect of the linear coding strategy proposed for example \ref{Ex:3-BCExample} is that users $2$ and $3$ employ a code that is closed under addition, the linear code being the simplest such example. Since linear codes only achieve symmetric capacity, we are forced to bin codewords from a larger linear code in order to find codewords that are typical with respect to a nonuniform distribution. This is akin to binning for channels with state information, wherein $\exp\left\{ nI(U;S) \right\}$ codewords, each picked according to $\prod_{t=1}^{n}p_{U}$, are chosen for each message in order to find a codeword in $T_{\delta}(U|s^{n})$ jointly typical with state sequence $s^{n}$.

We now generalize the coding technique proposed for example \ref{Ex:3-BCExample}. Consider auxiliary alphabet sets $\PrivateRVSet_{1},\SemiPrivateRVSet_{2},\SemiPrivateRVSet_{3}$ where $\SemiPrivateRVSet_{2}=\SemiPrivateRVSet_{3} =\fieldpi$ be the finite field of cardinality $\pi$ and let $p_{\PrivateRV_{1}\SemiPrivateRV_{2}\SemiPrivateRV_{3} X \underlineY}$ be a PMF on $\PrivateRVSet_{1} \times \SemiPrivateRVSet_{2} \times \SemiPrivateRVSet_{3} \times \InputAlphabet \times \underlineSetY$. For $j=2,3$, let $\lambda_{j} \subseteq \SemiPrivateRVSet_{j}^{n}$ be coset of a linear code $\overline{\lambda_{j}}\subseteq\fieldpi^{n}$ of rate $S_{j}\log\pi$. The linear codes are contained in one another, i.e., if $S_{j_{1}} \leq S_{j_{2}}$, then $\overline{\lambda_{j_{1}}} \subseteq \overline{\lambda_{j_{2}}}$. Codewords of $\lambda_{j}$ are partitioned independently and uniformly into $\exp \left\{ nT_{j} \right\}$ bins. A codebook $\mathcal{C}_{1}$ of rate $K_{1}+R_{1}$ is built over $\PrivateRVSet_{1}$. The codewords of $\mathcal{C}_{1}$ are independently and uniformly partitioned into $\exp \left\{ nR_{1} \right\}$ bins. Messages of users $1,2,3$ at rates $L_{1},T_{2}\log\pi,T_{3}\log\pi$ are used to index bins in $\mathcal{C}_{1},\lambda_{2},\lambda_{3}$ respectively. The encoder looks for a jointly typical triple, with respect to $p_{\PrivateRV_{1}\SemiPrivateRV_{2}\SemiPrivateRV_{3}}$, of codewords in the indexed triple of bins. Following a second moment method similar to that employed in \cite[Appendix A]{201403arXiv_PadPra}, it can be proved that the encoder finds at least one jointly typical triple if
\begin{eqnarray}
 \label{Eqn:SingleLayerSourceCodingBound1}
&K_{1} > 0,~~(S_{j}-T_{j})\log\pi > \log \pi - H(U_{j}),~~
(S_{j}-T_{j})\log\pi+K_{1} > \log \pi - H(U_{j})+I(U_{j};V_{1}),: j=2,3\\
\label{Eqn:SingleLayerSourceCodingBound2}
&\underset{j=2}{\overset{3}{\sum}}(S_{j}-T_{j})\log\pi > 2\log \pi - H(U_{2})-H(U_{3})+I(U_{2};U_{3})\\
\label{Eqn:SingleLayerSourceCodingBound3}
&K_{1}+\max\{S_{2},S_{3}\}\log\pi> \log\pi-H(U_{2}\oplus U_{3})+I(V_{1};U_{2}\oplus U_{3}), ~~\max\{S_{2},S_{3}\}\log\pi \geq \log\pi -H(U_{2}\oplus U_{3})\\
\label{Eqn:SingleLayerSourceCodingBound4}
&\underset{j=2}{\overset{3}{\sum}}(S_{j}-T_{j})\log\pi+K_{1} > 2\log \pi - \underset{j=2}{\overset{3}{\sum}}H(U_{j})+I(U_{2};U_{3};V_{1}).
\end{eqnarray}
Having chosen one such jointly typical triple, say $V_{1}^{n},U_{2}^{n},U_{3}^{n}$, it generates a vector $X^{n}$ according to \[p^{n}_{X|\PrivateRV_{1}\SemiPrivateRV_{2}\SemiPrivateRV_{3}}(\cdot | \PrivateRV_{1}^{n},\SemiPrivateRV_{2}^{n},\SemiPrivateRV_{3}^{n}) = \prod_{t=1}^{n}p_{X|\PrivateRV_{1}\SemiPrivateRV_{2}\SemiPrivateRV_{3}}(\cdot | \PrivateRV_{1t},\SemiPrivateRV_{2t},\SemiPrivateRV_{3t})\]and feeds the same as input on the channel.

Decoders $2$ and $3$ perform a standard PTP decoding. For example, decoder $2$ receives $Y_{2}^{n}$ and looks for all codewords in $\Lambda_{2}$ that are jointly typical with $Y_{2}^{n}$. If it finds all such codewords in a unique bin it declares the corresponding bin index as the decoded message. It can be proved by following the technique similar to \cite[Proof of Theorem 1]{201108ISIT_PadPra} that if
\begin{equation}
 \label{Eqn:BoundsOnUser2And3Rates}
S_{j}\log\pi< \log \pi - H(\SemiPrivateRV_{j}|Y_{j})\mbox{ for }j=2,3 
\end{equation}
then probability of decoding error at decoders $2$ and $3$ can be made arbitrarily small for sufficiently large $n$.
Having received $Y_{1}^{n}$, decoder $1$ looks for all codewords $v_{1}^{n} \in \mathcal{C}_{1}$ for which there exists a codeword $u_{2 \oplus 3}^{n} \in \Lambda_{2} \oplus \Lambda_{3}$ such that $(v_{1}^{n},u_{2 \oplus 3}^{n},Y_{1}^{n})$ is jointly typical with respect to $p_{\PrivateRV_{1},\SemiPrivateRV_{2}\oplus \SemiPrivateRV_{3},Y_{1}}$. Here \[\Lambda_{2}\oplus \Lambda_{3} \define \left\{ U_{2}^{n}\oplus U_{3}^{n}: U_{j}^{n} \in \Lambda_{j}^{n}:j=2,3 \right\}.\]If all such codewords in $\mathcal{C}_{1}$ belong to a unique bin, the corresponding bin index is declared as the decoded message. Again following the technique similar to \cite[Proof of Theorem 1]{201108ISIT_PadPra}, it can be proved, that if, for $j=2,3$
\begin{eqnarray}
 \label{Eqn:BoundsOnUser1Rates}
K_{1}\!+\!R_{1} \!< \!H(V_{1})\! - \!H(\PrivateRV_{1}|\SemiPrivateRV_{2}\oplus \SemiPrivateRV_{3},Y_{1}) ,~
K_{1}\!+\!R_{1}\!+\!S_{j}\!\log\pi < \log \pi +H(V_{1})\!-\! H(\PrivateRV_{1},\SemiPrivateRV_{2}\oplus \SemiPrivateRV_{3}|Y_{1}),
\end{eqnarray}
then probability of decoding error at decoder $1$ falls exponentially
with $n$. In the sequel, we provide a formal proof of
achievability. 

\subsubsection{Proof of achievability}

\begin{definition}
 \label{Defn:TestChannelsForSingleLayerWithOneUserDecodingInterference}
For $a=2,3$, let
$\SetOfDistributions_{1a}^{f}(W_{\underlineY|X},\kappa,\tau)$ denote
the collection of PMF's
$p_{Q\SemiPrivateRV_{2}\SemiPrivateRV_{3}\PrivateRV_{1} X
  \underlineY}$ defined on $\mathcal{Q} \times
\SemiPrivateRVSet_{2}\times \SemiPrivateRVSet_{3}\times
\PrivateRVSet_{1} \times \InputAlphabet \times \underlineSetY$, where
(i) $\SemiPrivateRVSet_{2}=\SemiPrivateRVSet_{3}=\fieldpi$ is the
finite field of cardinality $\pi$, $\PrivateRVSet_{1}$ is a finite
set, (ii)
$p_{\underlineY|X\PrivateRV_{1}\underlineSemiPrivateRV}=p_{\underlineY|X}=W_{\underlineY|X}$,
and (iii) $\Expectation\left\{ \kappa(X) \right\} \leq \tau$ and (iv)
$H(U_{a}|Y_{a}Q) < H(U_{2}\oplus U_{3}|Q)$. For
$p_{Q\underlineSemiPrivateRV\PrivateRV_{1} X \underlineY} \in
\SetOfDistributions_{1a}^{f}(W_{\underlineY|X},\kappa,\tau)$, let
$\beta_{1a}(p_{Q\underlineSemiPrivateRV\PrivateRV_{1} X \underlineY})$
be defined as the set of triples $(R_{1},R_{2},R_{3})$ that satisfy 
\begin{align}
\label{Eqn:OneLayerRateRegionForSpecificTestChannel}
0 < R_{1} &< I(V_{1};U_{2} \oplus  U_{3},Y_{1}|Q), ~~~~0 < R_{j} < I(U_{j};Y_{j}|Q):j=2,3,\nonumber\\
R_{1}+R_{a} &< I(U_{a};Y_{a}|Q)-I(U_{a};V_{1}|Q)+I(V_{1},U_{2}\oplus
  U_{3};Y_{1}|Q)+I(V_{1};U_{2}\oplus U_{3}|Q)\nonumber\\ 
R_{2}+R_{3} &< I(U_{2};Y_{2}|Q)+I(U_{3};Y_{3}|Q)-I(U_{2};U_{3}|Q)
 \nonumber\\
R_{1}+R_{j} &< H(V_{1},U_{j}|Q)-H(V_{1},U_{2}\oplus
U_{3}|Y_{1}Q)+\min\{0,H(U_2 \oplus U_3,Y_1|Q)-H(U_j|Y_jQ)\}:j=2,3\nonumber\\
\sum_{k=1}^{3}R_{k} &< H(U_{2},U_{3},V_{1}|Q)-H(V_{1},U_{2}\oplus U_{3}|Y_{1}Q)-\max \left\{H(U_{2}|Y_{2}Q),H(U_{3}|Y_{3}Q)  \right\}\nonumber\\
\sum_{k=1}^{3}R_{k} &< H( U_{2},U_{3},V_{1}|Q)\!-\!H(V_{1}|QU_{2}\!\oplus\! U_{3},Y_{1})-\sum_{k=2}^{3}H(U_{k}|QY_{k})\nonumber\\
R_{1}+\sum_{k=1}^{3}R_{k} &< H(V_{1}|Q)+H(U_{2}U_{3}V_{1}|Q)\!-\!2H(V_{1},U_{2}\!\oplus\! U_{3}|QY_{1}) \nonumber\\
R_{j}+\sum_{k=1}^{3}R_{k} &< H(V_{1},U_{j}|Q)+H(U_{2},U_{3}|Q)-2H(U_{j}|QY_{j})-H(V_{1},U_{2}\oplus U_{3}|QY_{1}):j=2,3\nonumber
\end{align}
and
\begin{eqnarray}
 \label{Eqn:OneLayerAchievableRateRegionWithOneUserDecodingInterference}
 \beta_{1}(W_{\underlineY|X},\kappa,\tau) = \cocl\left(\bigcup_{a=2}^{3}\underset{\substack{p_{\underlineSemiPrivateRV\PrivateRV_{1} X \underlineY} \\\in \SetOfDistributions_{1a}^{f}(W_{\underlineY|X},\kappa,\tau)}
}{\bigcup}\beta_{1a}(p_{Q\underlineSemiPrivateRV\PrivateRV_{1} X \underlineY})\right).\nonumber
\end{eqnarray}
\end{definition}
\begin{thm}
 \label{Thm:OneLayerAchievableRateRegionWithOneUserDecodingInterference}
For a $3-$DBC $(\mathcal{X},\underlineSetY,W_{\underlineY|X},\kappa)$, $\beta_{1}(W_{\underlineY|X},\kappa,\tau)$ is achievable, i.e., $\beta_{1}(W_{\underlineY|X},\kappa,\tau) \subseteq \mathbb{C}(W_{\underlineY|X},\kappa,\tau)$.
\end{thm}

\begin{proof}
Given $p_{QV_{1}\underline{U}X\underline{Y}} \in \mathbb{D}_{1a}^{f}(W_{\underlineY|X},\kappa,\tau)$, for some $a=2,3$, $\underline{R} \in \beta_{1}(p_{QV_{1}\underlineSemiPrivateRV X \underlineY}), \tilde{\eta} >0$, our task is to identify a $3-$DBC code $(n,\underlineMessageSetM,e,\underlined)$ of rate $\frac{\log \mathcal{M}_{j}}{n} \geq R_{j}-\tilde{\eta}:j =1,2,3$, average error probability $\overline{\xi}(e,\underlined) \leq \tilde{\eta}$, and average cost $\tau(e) \leq \tau+\tilde{\eta}$.
Taking a cue from the above coding technique, we begin with an
alternate characterization of
$\beta_{1a}(p_{QV_{1}\underlineSemiPrivateRV X \underlineY})$ in terms
of the parameters of the code. 

\begin{definition}
 \label{Defn:3To1DBCOneLayerAchievabilityAlternateCharacterizationOfRateRegion}
Consider $p_{QV_{1}\underlineSemiPrivateRV X \underlineY} \in \mathbb{D}_{1a}^{f}(W_{\underlineY|X},\kappa,\tau)$ and let $\pi \define |\mathcal{U}_{2}|=|\mathcal{U}_{3}|$. For $a=2,3$, let $\tilde{\beta_{1a}}(p_{QV_{1}\underlineSemiPrivateRV X \underlineY})$ be defined as the set of rate triples $\underlineR \define (R_{1},R_{2},R_{3}) \in [0,\infty)^{3}$ for which $\underset{\delta > 0}{\cup}\mathcal{S}_{a}(\underlineR,p_{QV_{1}\underlineSemiPrivateRV X \underlineY},\delta)$ is non-empty, where, for any $\delta > 0$, $\mathcal{S}_{a}(\underlineR,p_{QV_{1}\underlineSemiPrivateRV X \underlineY},\delta)$ is defined as the set of vectors $(K_{1},R_{1},S_{2},T_{2},S_{3},T_{3}) \in [0,\infty)^{6}$ that satisfy $R_{j}=T_{j}\log \pi$,
 \begin{eqnarray}
  \label{Eqn:3To1BCAchievabilityProofPreFourierMotzkinSourceCodingBound1}
 &K_{1}>\delta,~~(S_{j}-T_{j})\log \pi > \log \pi - H(U_{j}|Q)+\delta,\\ &K_{1}+(S_{j}-T_{j})\log \pi > \log \pi - H(U_{j}|Q,V_{1})+\delta,~~
 \label{Eqn:3To1BCAchievabilityProofPreFourierMotzkinSourceCodingBound2}
 \underset{l=2}{\overset{3}{\sum}}(S_{l}-T_{l})\log \pi > 2\log \pi -H(\underline{U}|Q)+\delta,  \\
 \label{Eqn:3To1BCAchievabilityProofPreFourierMotzkinSourceCodingBound3}
 &K_{1}+\underset{l=2}{\overset{3}{\sum}}(S_{l}-T_{l})\log \pi  > 2\log \pi - H(\underline{U}|Q,V_{1})+\delta,~~S_{a}\log\pi > \log\pi-H(U_{2}\oplus U_{3}|Q)+\delta,
\end{eqnarray}
 \begin{eqnarray}
 \label{Eqn:3To1BCAchievabilityProofPreFourierMotzkinSourceCodingBound4}
&K_{1}+S_{a}\log\pi \overset{(a)}{>} \log\pi-H(U_{2}\oplus U_{3}|Q,V_{1})+\delta,~~
 \label{Eqn:3To1BCAchievabilityProofPreFourierMotzkinChannelCodingBound1}
  K_{1}+R_{1} < I(V_{1};Y_{1},U_{2}\oplus U_{3}|Q)-\delta,\\
 \label{Eqn:3To1BCAchievabilityProofPreFourierMotzkinChannelCodingBound2}
 &K_{1}+R_{1} + \max \left\{ S_{2},S_{3} \right\}\log \pi< \log \pi +H(V_{1}|Q)-H(V_{1},U_{2}\oplus U_{3}|Q,Y_{1})-\delta\nonumber\\
 \label{Eqn:3To1BCAchievabilityProofPreFourierMotzkinChannelCodingBound3}
 &S_{j}\log\pi < \log \pi - H(U_{j}|Q,Y_{j})-\delta,\nonumber
 \end{eqnarray}
for $j=2,3$.
\end{definition}
\begin{lemma}
 \label{Lem:AltCharRateRegnIsSame}
$\tilde{\beta_{1a}}(p_{QV_{1}\underlineSemiPrivateRV X \underlineY})=\beta_{1a}(p_{QV_{1}\underlineSemiPrivateRV X \underlineY})$ for every $p_{QV_{1}\underlineSemiPrivateRV X \underlineY} \in \mathbb{D}_{1a}^{f}(W_{\underlineY|X},\kappa,\tau)$ and $a=2,3$.
\end{lemma}
\begin{proof}
The proof follows by substituting $R_{j}=T_{j}\log \pi$ for $j=2,3$ in
the bounds characterizing
$\mathcal{S}_{a}(\underlineR,p_{QV_{1}\underlineSemiPrivateRV X
  \underlineY},\delta)$ and eliminating $K_{1},S_{j}:j=2,3$ via the
technique proposed in \cite{chaharsooghi2011new}. The
presence of strict inequalities in the bounds characterizing
${\beta_{1a}}(p_{QV_{1}\underlineSemiPrivateRV X \underlineY})$ and
$\mathcal{S}_{a}(\underlineR,p_{QV_{1}\underlineSemiPrivateRV X
  \underlineY},\delta)$ enables one to prove $\underset{\delta >
  0}{\cup}\mathcal{S}_{a}(\underlineR,p_{QV_{1}\underlineSemiPrivateRV
  X \underlineY},\delta)$ is non-empty for every $\underline{R} \in
{\beta_{1a}}(p_{QV_{1}\underlineSemiPrivateRV X \underlineY})$. 
\end{proof}
For the given rate triple $\underlineR \in \beta_{1a}(p_{QV_{1}\underlineSemiPrivateRV X \underlineY})$, we have $\delta_{1}>0$ and $(K_{1},R_{1},S_{2},T_{2},S_{3},T_{3}) \in \mathcal{S}_{a}(\underlineR,p_{QV_{1}\underlineSemiPrivateRV X \underlineY},\delta_{1})$. Set $\eta \define \min \left\{ \tilde{\eta},\delta_{1} \right\}$. Consider a codebook $\mathcal{C}_{1} = ( v_{1}^{n}(m_{1},b_{1}):m_{1} \in \mathcal{M}_{1}, b_{1} \in \mathcal{B}_{1} )$ built over $\mathcal{V}_{1}$ consisting of $|\mathcal{M}_{1}|$ bins, each consisting of $|\mathcal{B}_{1}|$ codewords. We let $\mathcal{M}_{1}=[\lfloor \exp \left\{ n(R_{1}-\frac{\eta}{2}) \right\} \rfloor]$ and $\mathcal{B}_{1} = [\lceil\exp \left\{ n(K_{1}+\frac{\eta}{8}) \right\}\rceil]$. $\mathcal{C}_{1}$ is employed to encode user $1$'s message. Codebooks employed to encode user $2$ and $3$'s messages are partitioned coset codes which are described in the sequel. Henceforth, we let $\pi \define |\SemiPrivateRVSet_{2}|=|\SemiPrivateRVSet_{3}|$ and therefore $\fieldpi = \SemiPrivateRVSet_{2}=\SemiPrivateRVSet_{3}$. Consider a linear code $\overline{\lambda} \subseteq \fieldpi^{n}$ with generator matrix $g \in \fieldpi^{s \times n}$ and let ${\lambda} \subseteq \fieldpi^{n}$ denote the coset of $\overline{\lambda}$ with respect to shift $b^{n} \in \fieldpi^{n}$. Clearly, the codewords of ${\lambda}$ are given by $u(a^{s}) \define a^{s}g \oplus b^{n}: a^{s} \in \fieldpi^{s}$. Consider a partition of ${\lambda}$ into $\pi^{t}$ bins. Each codeword $u(a^{s})$ is assigned a bin index $\bin(a^{s}) \in \fieldpi^{t}$. For every $m^{t} \in \fieldpi^{t}$, $c(m^{t}) \define \left\{ a^{s}:\bin(a^{s}) = m^{t} \right\}$ denotes the set of indices whose codewords are assigned to bin $m^{t}$. The coset code $\lambda$ with it's partitions is called a \textit{partitioned coset code} (PCC) and is referred to as the PCC $(n,s,t,g,b^{n},\bin)$.
 
For $j=2,3$, user $j$ is provided the PCC $(n,s_{j},t_{j},g_{j},b_{j}^{n},i_{j})$, where $s_{j}= \lfloor nS_{j} \rfloor, t_{j} \define \lceil n(T_{j}-\frac{\eta}{4\log \pi}) \rceil $. Let $u_{j}^{n}(a_{j}^{s_{j}}) \define a_{j}^{s_{j}}g_{j}\oplus b_{j}^{n}$ denote a generic codeword in $\lambda_{j}$ and $c_{j}(m_{j}^{t_{j}}) \define \left\{ a_{j}^{s_{j}}:i_{j}(a_{j}^{s_{j}})=m_{j}^{t_{j}} \right\}$ denote the indices of codewords in bin corresponding to message $m_{j}^{t_{j}}$. These codes are such that if $s_{j_{1}} \leq s_{j_{2}}$, then $g_{j_{2}}^{t} = \left[ g_{j_{1}}^{t} ~~ g_{j_{2}/j_{1}}^{t} \right]$. In other words, the linear code corresponding to the larger coset code contains the linear code corresponding to the smaller coset code. Without loss of generality, we henceforth assume $s_{2} \leq s_{3}$ and therefore $g_{3}^{t} = \left[ g_{2}^{t} ~~ g_{3/2}^{t} \right]$. It is now appropriate to derive some relationships between the code parameters that would be of use at a later time. There exists $N_{1}(\eta) \in \naturals$ such that for all $n \geq N_{1}(\eta)$
\begin{eqnarray}
\label{Eqn:3To1DBCOneLayerAchievabilityLowerAndUpperBoundsOnSj}
& nS_{j}-1 \leq s_{j} \leq nS_{j}\mbox{ and therefore } S_{j}-\frac{\eta}{8\log\pi}\leq S_{j}-\frac{1}{n} \leq \frac{s_{j}}{n} \leq S_{j}, \\
 \label{Eqn:3To1DBCOneLayerAchievabilityLowerAndUpperBoundsOnTj}
 &n\left(T_{j}-\frac{\eta}{4\log \pi}\right) \leq t_{j} \leq n\left(T_{j}-\frac{\eta}{4\log \pi}\right)+1\mbox{ and therefore }T_{j}-\frac{\eta}{4\log \pi} \leq \frac{t_{j}}{n} \leq T_{j}-\frac{\eta}{8\log \pi}+\frac{1}{n},\\
 \label{Eqn:3To1DBCOneLayerAchievabilityLowerAndUpperBoundsOnR1AndK1}
&R_{1}-\eta \leq \frac{\log |\mathcal{M}_{1}|}{n} \leq R_{1}-\frac{\eta}{2}\mbox{ and }K_{1}+\frac{\eta}{8} \leq \frac{\log |\mathcal{B}_{1}|}{n} \leq K_{1}+\frac{\eta}{4}.
\end{eqnarray}
We now describe the encoding and decoding rules. A vector $q^{n} \in T_{\eta_{2}}(Q)$ is chosen to be the time-sharing vector, where $\eta_{2}$ will be specified in due course. Without loss of generality, we assume the message sets are $\mathcal{M}_{j}\define \fieldpi^{t_{j}}$ for $j=2,3$ and as stated before $\mathcal{M}_{1} \define [\lfloor \exp \left\{ n(R_{1}-\frac{\eta}{2}) \right\} \rfloor]$. Let $(M_{1},M_{2}^{t_{2}},M_{3}^{t_{3}}) \in \underline{\mathcal{M}}$ denote the uniformly distributed triple of message random variables to be communicated to the respective users. The encoder looks for a triplet $(b_{1},a_{2}^{s_{2}},a_{3}^{s_{3}}) \in \mathcal{B}_{1} \times c_{2}(M_{2}^{t_{2}}) \times c_{3}(M_{3}^{t_{3}})$ such that $(v_{1}^{n}(M_{1},b_{1}),u_{2}^{n}(a_{2}^{s_{2}}),u_{3}^{n}(a_{3}^{s_{3}})) \in 
T_{2\eta_{2}}(V_{
1},U_{2},U_{3}|q^{n})$.\footnote{Here, the typicality is with respect to $p_{QV_{1}\underline{U}X\underlineY}$.} If it finds at least one such triple, one of them is chosen according to a predefined rule. Otherwise, i.e, if it finds no triple of codewords in the indexed triple of bins that is jointly typical, it chooses a fixed triple of codewords in $\mathcal{C}_{1}\times \lambda_{2} \times \lambda_{3}$. In either case, let $(v_{1}^{n}(M_{1},B_{1}),u_{2}^{n}(A_{2}^{s_{2}}),u_{3}(A_{3}^{s_{3}}))$ denote the chosen triple of codewords. In the former case, the encoder maps the triple to a vector in $T_{4\eta_{2}}(X|v_{1}^{n}(M_{1},B_{1}),u_{2}^{n}(A_{2}^{s_{2}}),u_{3}(A_{3}^{s_{3}}))$ and feeds the same as input on the channel. In the latter case, it picks a fixed vector in $\mathcal{X}^{n}$ and feeds the same as input on the channel. In either case, let $x^{n}(M_{1},M_{2}^{t_{2}},M_{3}^{t_{3}})$ denote the vector input on the channel.

The operations of decoders $2$ and $3$ are identical and we describe the same through the generic index $j$. Having received vector $Y_{j}^{n}$, it looks for all messages $\hat{m}_{j}^{t_{j}} \in \mathcal{M}_{j}$ such that for some $a_{j}^{s_{j}} \in c_{j}(\hat{m}_{j}^{t_{j}})$, $u_{j}(a_{j}^{s_{j}}) \in T_{8\eta_{2}}(U_{j}|q^{n},Y_{j}^{n})$. If it finds exactly one such message, this is declared as the decoded message. Otherwise, an error is declared.
Decoder $1$ is provided with the codebook $\lambda_{2} \oplus \lambda_{3} \define \left\{ u_{2}^{n}(a_{2}^{s_{2}}) \oplus u_{3}^{n}(a_{3}^{s_{3}}): a_{j}^{s_{j}} \in \fieldpi^{s_{j}}:j=2,3 \right\}$. Note that $\lambda_{2}\oplus \lambda_{3} = \left\{ u_{\oplus}(a_{3}^{s_{3}}) \define a_{3}^{s_{3}}g_{3}\oplus b_{2}^{n}\oplus b_{3}^{n}: a_{3}^{s_{3}} \in \fieldpi^{s_{3}}  \right\}$.
Having received $Y_{1}^{n}$, decoder $1$ looks for all messages $\hat{m}_{1} \in \mathcal{M}_{1}$ such that $(v_{1}^{n}(\hat{m}_{1},b_{1}),u_{\oplus}(a_{3}^{s_{3}})) \in T_{8\eta_{2}}(V_{1},U_{2}\oplus U_{3}|q^{n},Y_{1}^{n})$ for some $(b_{1},a_{3}^{s_{3}}) \in \mathcal{B}_{1} \times \fieldpi^{s_{3}}$. If it finds exactly one such $\hat{m}_{1} \in \mathcal{M}_{1}$, this is declared as the decoded message. Otherwise, an error is declared.

The above encoding and decoding rules map a triplet $\mathcal{C}_{1},\lambda_{2},\lambda_{3}$ of codebooks into a $3-$DBC code\footnote{This map also relies on a `predefined' rule to choose among many jointly typical triples within an indexed pair of bins and furthermore, a rule to decide among many input sequences that is conditionally typical with this chosen triple of codewords.}. Moreover, (\ref{Eqn:3To1DBCOneLayerAchievabilityLowerAndUpperBoundsOnTj}) and (\ref{Eqn:3To1DBCOneLayerAchievabilityLowerAndUpperBoundsOnR1AndK1}) imply that the rates of the corresponding $3-$DBC code satisfy $\frac{\log \mathcal{M}_{1}}{n} \geq R_{1}-\eta$, $\frac{t_{j}\log\pi}{n} \geq R_{j}-\frac{\tilde{\eta}}{4}$ for $j=2,3$. Since every triple $\mathcal{C}_{1},\lambda_{2},\lambda_{3}$ of codebooks, and a choice for the predefined rules map to a corresponding $3-$DBC code, we have characterized an ensemble of $3-$DBC codes, one for each $n \in \naturals$. We now induce a distribution over this ensemble of $3-$DBC codes.

Consider a random triple $\mathcal{C}_{1},\Lambda_{2},\Lambda_{3}$ of codebooks, where $\mathcal{C}_{1}= \left( V_{1}^{n}(m_{1},b_{1}):(m_{1},b_{1}) \in \mathcal{M}_{1}\times \mathcal{B}_{1} \right)$ and $\Lambda_{j}$ is the random PCC $(n,s_{j},t_{j},G_{j},B_{j}^{n},I_{j})$. Note that the joint distribution of $V_{1}^{n}(m_{1},b_{1}):(m_{1},b_{1}) \in \mathcal{M}_{1}\times \mathcal{B}_{1}, G_{2},G_{3/2},B_{2}^{n},B_{3}
^{n},I_{2}(a_{2}^{s_{2}}):a_{2}^{s_{2}} \in \fieldpi^{s_{2}},I_{3}(a_{3}^{s_{3}}): a_{3}^{s_{3}} \in \fieldpi^{s_{3}}$ uniquely characterizes the distribution of $\mathcal{C}_{1},\Lambda_{2},\Lambda_{3}$. We let $V_{1}^{n}(m_{1},b_{1}):(m_{1},b_{1}) \in \mathcal{M}_{1}\times \mathcal{B}_{1}, G_{2},G_{3/2},B_{2}^{n},B_{3}^{n},I_{2}(a_{2}^{s_{2}}):a_{2}^{s_{2}} \in \fieldpi^{s_{2}},I_{3}(a_{3}^{s_{3}}): a_{3}^{s_{3}} \in \fieldpi^{s_{3}}$ be mutually independent. For every $(m_{1},b_{1}) \in \mathcal{M}_{1}\times \mathcal{B}_{1}$, $v_{1}^{n} \in \mathcal{V}_{1}^{n}$, let $P(V_{1}^{n}(m_{1})=v_{1}^{n}) = \prod_{t=1}^{n}p_{V_{1}|Q}(v_{1t}|q_{t})$. The rest of the random objects $G_{2},G_{3/2},B_{2}^{n},B_{3}^{n},I_{2}(a_{2}^{s_{2}}
):a_{2}^{s_{2}} \in \fieldpi^{s_{2}},I_{
3}(a_{3}^{s_{3}}): a_{3}^{s_{3}} \in \fieldpi^{s_{3}}$ are uniformly distributed over their respective range spaces. We have therefore specified the distribution of the random triple $\mathcal{C}_{1},\Lambda_{2},\Lambda_{3}$ of codebooks. For $j=2,3$, we let $U_{j}^{n}(a_{j}^{s_{j}})=a_{j}^{s_{j}}G_{j}\oplus B_{j}^{n}$ denote a generic random codeword in the random codebook $\Lambda_{j}$. Likewise, we let $U_{\oplus}^{n}(a_{3}^{s_{3}})=a_{3}^{s_{3}}G_{3}\oplus B_{2}^{n}\oplus B_{3}^{n}$ denote a generic codeword in $\Lambda_{2}\oplus \Lambda_{3}$. Let $(V_{1}^{n}(M_{1},B_{1}),U_{2}^{n}(A_{2}^{s_{2}}),U_{3}^{n}(A_{3}^{s_{3}}))$ denote the triple of codewords chosen by the encoder and $X^{n}(M_{1},M_{2}^{t_{2}},M_{3}^{t_{3}})$ denote the vector input on the channel.

While the above specifies the distribution of the random triple of $\mathcal{C}_{1} , \Lambda_{2} , \Lambda_{3}$ of codebooks, the predefined rules that map it to a $3-$DBC code is yet unspecified. In other words, the distribution of $(V_{1}^{n}(M_{1},B_{1}),U_{2}^{n}(A_{2}^{s_{2}}),U_{3}^{n}(A_{3}^{s_{3}}))$ and $X^{n}(M_{1},M_{2}^{t_{2}},M_{3}^{t_{3}})$ need to be specified. All the $3-$DBC codes that a particular triplet of codebooks $\mathcal{C}_{1},\lambda_{2},\lambda_{3}$ map to, are uniformly distributed. Alternatively, the encoder picks a triple in \[\left\{ (V_{1}^{n}(M_{1},b_{1}),U_{2}(a_{2}^{s_{2}}),U_{3}(a_{3}^{s_{3}})) \in T_{2\eta_{2}}(V_{1},\underline{U}|q^{n}) : (b_{1},a_{2}^{s_{2}},a_{3}^{s_{3}}) \in \mathcal{B}_{1} \times C_{2}(M_{2}^{t_{2}}) \times C_{3}(M_{3}^{t_{3}})  \right\}\] uniformly at random and independent of other choices. Denoting this random triple as $(V_{1}^{n}(M_{1},B_{1}),U_{2}^{n}(A_{2}^{s_{2}}),$ $U_{3}^{n}(A_{3}^{s_{3}}))$, the encoder picks an input sequence in $T_{2\eta_{2}}(X_{}|(V_{1}^{n}(M_{1},B_{1}),U_{2}^{n}(A_{2}^{s_{2}}),U_{3}^{n}(A_{3}^{s_{3}})))$ uniformly at random and independent of other choices. We have therefore specified the distribution induced on the corresponding ensemble of $3-$DBC codes. In the sequel, we characterize error events associated with this random $3-$DBC code.

If
\begin{eqnarray}
\epsilon_{1} &\define& \underset{\substack{(b_{1},a_{2}^{s_{2}},a_{3}^{s_{3}}) \\ \mathcal{B}_{1} \times C_{2}(M_{2}^{t_{2}}) \times C_{3}(M_{3}^{t_{3}}) }}{\bigcap} \left\{ (V_{1}(M_{1},b_{1}),U_{2}(a_{2}^{s_{2}}),U_{3}(a_{3}^{s_{3}})) \notin T_{2\eta_{2}}(V_{1},U_{2},U_{3}|q^{n})  \right\} \nonumber\\
\epsilon_{31} &\define& \underset{ \substack{(b_{1},a_{3}^{s_{3}}) \\\in \mathcal{B}_{1}\times\fieldpi^{s_{3}}}}{\bigcap} \left\{  (V_{1}(M_{1},b_{1}),U_{\oplus}^{n}(a_{3}^{s_{3}}),Y_{1}^{n}) \notin T_{8\eta_{2}}(V_{1},U_{2}\oplus U_{3},Y_{1}|q^{n}) \right\}, \nonumber\\
\epsilon_{3j} &\define& \underset{\substack{a_{j}^{s_{j}} \in C_{j}(M_{j}^{t_{j}}) }}{\bigcap}\left\{ (U_{j}(a_{j}^{s_{j}}),Y_{j}^{n}) \notin T_{8\eta_{2}}(U_{j},Y_{j}|q^{n}) \right\} \nonumber\\
\epsilon_{41} &\define& \underset{\substack{(b_{1},a_{3}^{s_{3}}) \\\in \mathcal{B}_{1}\times\fieldpi^{s_{3}} }}{\bigcup} \underset{\substack{\hat{m}_{1} \neq M_{1}}}{\bigcup} \left\{  (V_{1}(\hat{m}_{1},b_{1}),U_{\oplus}^{n}(a_{3}^{s_{3}}),Y_{1}^{n}) \in T_{8\eta_{2}}(V_{1},Y_{1}|q^{n}) \right\}, \nonumber\\
\epsilon_{4j} &\define& \underset{\substack{a_{j}^{s_{j}} \in C_{j}(\hat{m}_{j}^{t_{j}}) \\\hat{m}_{j}^{t_{j}}\neq M_{j}^{t_{j}}}}{\bigcup}\left\{ (U_{j}(a_{j}^{s_{j}}),Y_{j}^{n}) \in T_{8\eta_{2}}(U_{j},Y_{j}|q^{n}) \right\}, \nonumber
\end{eqnarray}
then $\epsilon \define \underset{j=1}{\overset{3}{\cup}} \left( \epsilon_{1}\cup \epsilon_{3j}\cup \epsilon_{4j} \right)$ contains the error event. Our next task is to derive an upper bound on $P(\epsilon)$.

Let
\begin{eqnarray}
 \label{Eqn:3DBCProofListIndicator}
 \phi(m_{1},m_{2}^{t_{2}},m_{3}^{t_{3}}) &\define &\!\!\!\!\sum_{\substack{(b_{1},a^{s_{2}},a^{s_{3}}) \in \\ \mathcal{B}_{1} \times \fieldpi^{s_{2}} \times \fieldpi^{s_{3}}}} 1_{\left\{ (V_{1}^{n}(m_{1},b_{1}),U_{2}(a^{s_{2}}),U_{3}(a^{s_{3}})) \in T_{2\eta_{2}}(V_{1},U_{2},U_{3}|q^{n}),I(a^{s_{j}})=m_{j}^{t_{j}}:j=2,3  \right\}}, \nonumber\\
 \epsilon_{l} &\define &\!\!\!\!\left\{ \phi(M_{1},M_{2}^{t_{2}},M_{3}^{t_{3}}) < \mathcal{L}(n) \right\}, \mbox{ where } \mathcal{L}(n) \define \frac{1}{2} \Expectation\left\{ \phi(M_{1},M_{2}^{t_{2}},M_{3}^{t_{3}})\right\}.\nonumber
\end{eqnarray}
Clearly $P(\epsilon) \leq P(\epsilon_{l})+P(\epsilon_{l}^{c}\cap \epsilon)$, and it therefore suffices to derive upper bounds on each of these terms.\newline

\textit{Upper bound on $P(\epsilon_{l})$:-} Substituting for $\mathcal{L}(n)$, we have
\begin{eqnarray}
 \label{Eqn:3DBCProofSecondMomentMethod}
 P(\epsilon_{l}) &\leq& P(\left\{|\phi(M_{1},M_{2}^{t_{2}},M_{3}^{t_{3}})-\Expectation \left\{ \phi(M_{1},M_{2}^{t_{2}},M_{3}^{t_{3}}) \right\}|\geq \frac{\Expectation \left\{ \phi(M_{1},M_{2}^{t_{2}},M_{3}^{t_{3}}) \right\}}{2}\right\})\nonumber\\
 \label{Eqn:3To1DBCOneLayerEncoderErrorEventUseOfCheybyshevIneq}
 &\leq& \frac{4\Var\left\{ \phi(M_{1},M_{2}^{t_{2}},M_{3}^{t_{3}}) \right\}}{\left(\Expectation\left\{ \phi(M_{1},M_{2}^{t_{2}},M_{3}^{t_{3}}) \right\}\right)^{2}}
\end{eqnarray}
from the Cheybyshev inequality. In appendix \ref{AppSec:AnalysisOfEncoderErrorEventFor3To1DBC}, we evaluate the variance and expectation of $\phi(M_{1},M_{2}^{t_{2}},M_{3}^{t_{3}})$ and derive an upper bound on $P(\epsilon_{l})$. In particular, we prove for $n \geq \max\{ N_{1}(\eta), N_{2}(\eta_{2})\}$,
\begin{equation}
\label{Eqn:3To1DBCOneLayerAchievabilityFinalupperBoundOnEncoderErrorEvent}
 P(\epsilon_{1}) \leq (28+8\pi)\exp \left\{ -n\left( \delta_{1}-\frac{\eta}{8}-48\eta_{2}\right) \right\}.
\end{equation}

Now consider $\epsilon_{l}^{c}\cap \epsilon_{1}$. Note that $P(\epsilon_{1}) = P(\phi(M_{1},M_{2}^{t_{2}},M_{3}^{t_{3}})=0)$, and hence $\epsilon_{l}^{c}\cap \epsilon_{1} = \phi$, the empty set, if $\mathcal{L}(n)> 1$. At the end of appendix \ref{AppSec:AnalysisOfEncoderErrorEventFor3To1DBC}, we prove $\mathcal{L}(n)> 1$ for sufficiently large $n$. We are left to derive an upper bound on $P(\epsilon_{l}^{c} \cap \underset{j=1}{\overset{3}{\cup}} \left( \epsilon_{3j}\cup \epsilon_{4j} \right))$.

Since $\mathcal{L}(n)> 1$, $\epsilon_{l}^{c} \subseteq \epsilon_{1}^{c}$, it suffices to derive an upper bound on the terms $P(\epsilon_{1}^{c}\cap \left( \epsilon_{31}\cup \epsilon_{32}\cup\epsilon_{33}\right))$, $P(\epsilon_{l}^{c}\cap \left( \epsilon_{31}\cup \epsilon_{32}\cup\epsilon_{33}\right)^{c}\cap\epsilon_{4j}):j=1,2,3$.

\textit{Upper bound on $P(\epsilon_{1}^{c}\cap \left( \epsilon_{31}\cup \epsilon_{32}\cup\epsilon_{33}\right))$:-} Consider $P(\epsilon_{1}^{c} \cap \epsilon_{2})$, where
\begin{eqnarray}
 \label{Eqn:3To1DBCOneLayerAchievabilityInputIsTypical}
 \epsilon_{2} \define \left\{ (V_{1}(M_{1},B_{1}),U_{2}(A_{2}^{s_{2}}),U_{3}(A_{3}^{s_{3}}),X^{n}) \notin T_{4\eta_{2}}(V_{1},\underline{U},X|q^{n})  \right\}.\nonumber
\end{eqnarray}
By the encoding rule $P(\epsilon_{1}^{c}\cap \epsilon_{2}) = 0$. Since the encoding rule also ensures $\epsilon_{1}^{c}\cap (\epsilon_{31}\cup \epsilon_{32}\cup \epsilon_{33}) \subseteq \epsilon_{1}^{c} \cap \epsilon_{3}$, where
\begin{eqnarray}
 \label{Eqn:3To1DBCOneLayerLegitimateCodewordsNotBeingJointlyTypical}
 \epsilon_{3} \define \left\{ (V_{1}^{n}(M_{1},B_{1}),U_{2}^{n}(A_{2}^{s_{2}}),U_{3}^{n}(A_{3}^{s_{3}}),X^{n}(M_{1},M_{2}^{t_{2}},M_{3}^{t_{3}}),\underline{Y}^{n}) \notin T_{8\eta_{2}}(V_{1},\underline{U},X,\underline{Y})  \right\}, \nonumber
\end{eqnarray}
it suffices to derive an upper bound on $P((\epsilon_{1}\cup\epsilon_{2})^{c}\cap \epsilon_{3})$. This follows from conditional frequency typicality and $p_{\underlineY|X\PrivateRV_{1}\underlineSemiPrivateRV\TimeSharingRV}=p_{\underlineY|X}=W_{\underlineY|X}$. We conclude the existence of $N_{3}(\eta_{2})$ such that for all $n\geq N_{4}(\eta_{2})$, $P((\epsilon_{1}\cup\epsilon_{2})^{c}\cap \epsilon_{3}) \leq \frac{\eta}{32}$.

\textit{Upper bound on $P((\epsilon_{l}\cup \epsilon_{2}\cup \epsilon_{3})^{c}\cap\epsilon_{41})$} : We refer the reader to appendix \ref{AppSec:AnalysisOfDecoder1ErrorEventFor3To1DBC} for the derivation of an upper bound on $P((\epsilon_{1}\cup \epsilon_{2}\cup \epsilon_{3})^{c}\cap\epsilon_{41})$. Therein, we prove existence of $N_{4}(\eta_{2}) \in \naturals$ such that for all $n \geq \max \left\{ N_{1}(\eta) , N_{4}(\eta_{2})\right\}$, we have
\begin{equation}
 \label{Eqn:3DBCOneLayerAchievabilityUpperBoundOnDecoder1ErrorEventCopiedFromAppendix}
 P((\epsilon_{l}\cup \epsilon_{2}\cup \epsilon_{3})^{c}\cap \epsilon_{41}) 
  \leq  4 \exp \left\{ -n\left(\delta_{1}+\frac{\eta}{4}-56\eta_{2} \right) \right\}.
\end{equation}

\textit{Upper bound on $P((\epsilon_{l}\cup \epsilon_{2}\cup \epsilon_{3})^{c}\cap\epsilon_{4j})$} : For $j=2,3$, decoder $j$ performs a simple PTP decoding and therefore the reader might expect the analysis here to be quite standard. The partitioned coset code structure of user $j$'s codebook that involves correlated codewords and bins lends some technical complexities. We flesh out the details in appendix \ref{AppSec:AnalysisOfDecoder2And3ErrorEventsFor3To1DBC}. In particular, we prove (\ref{Eqn:3DBCProofFinalInequalityInDecoder2ErrorEvent}) existence of $N_{5}(\eta_{2}) \in \naturals$ such that for all $n \geq \max \{ N_{1}(\eta), N_{5}(\eta_{2})\}$
\begin{equation}
 \label{Eqn:3To1DBCOneLayerAchievabilityFinalUpperBoundOnDecoder2And3ErrorEventsCopiedFromAppendix}
 P((\epsilon_{1}\cup\epsilon_{2}\cup\epsilon_{3})^{c}\cap \epsilon_{4j}) \leq 2\exp \left\{ -n \left(\delta_{1}-32\eta_{2} \right) \right\}.
\end{equation}
Let us now compile the upper bounds derived in (\ref{Eqn:3To1DBCOneLayerAchievabilityFinalupperBoundOnEncoderErrorEvent}), (\ref{Eqn:3DBCOneLayerAchievabilityUpperBoundOnDecoder1ErrorEventCopiedFromAppendix}) and (\ref{Eqn:3To1DBCOneLayerAchievabilityFinalUpperBoundOnDecoder2And3ErrorEventsCopiedFromAppendix}). For $n \geq \max \{ N_{1}(\eta), N_{2}(\eta_{2}) N_{3}(\eta_{2}),N_{4}(\eta_{2}),$ $N_{5}(\eta_{2})\}$, we have
\begin{eqnarray}
 \label{Eqn:3To1BCOneLayerAchievabilityCompilingAllUpeerBounds}
P(\epsilon_{1}\cup\epsilon_{2}\cup\epsilon_{3}\cup\epsilon_{41}\cup\epsilon_{42}) \leq \frac{\eta}{32}+(34+8\pi)\exp \left\{ -n\left(\delta_{1}-\frac{\eta}{8}-56\eta_{2}\right) \right\}.
\end{eqnarray}
Recall that $\eta $ is chosen to be $\min \left\{ \tilde{\eta}, \delta_{1} \right\}$. By choosing $\eta_{2}=\frac{\eta}{56 \times 8}$, we have $\delta_{1}-\frac{\eta}{8}-\frac{\eta}{8} > \frac{3\eta}{4}$ and we can drive the probability of error below $\tilde{\eta}$ by choosing $n$ sufficiently large.

The only element left to argue is the random code satisfies the cost constraint. Since $P(\epsilon_{1}\cup \epsilon_{2})$ is lesser than $\frac{\tilde{\eta}}{2}$ for sufficiently large $n$, the encoder inputs a vector on the channel that is typical with respect $p_{X}$ with probability $1-\frac{\tilde{\eta}}{2}$. Since $\Expectation\left\{\kappa(X)\right\} \leq \tau$, a standard argument proves that the expected cost of the input vector can be made arbitrarily close to $\tau$ by choosing $n$ sufficiently large and $\eta_{2}$ sufficiently small. We leave the details to the reader.
\end{proof}

For example \ref{Ex:3-BCExample}, if $\tau * \delta_{1} \leq \min \left\{ \delta_{2},\delta_{3} \right\}$, then $(h_{b}(\tau*\delta_{1})-h_{b}(\delta_{1}),1-h_{b}(\delta_{2}),1-h_{b}(\delta_{3})) \in \beta_{1}(W_{\underlineY|X},\kappa,\tau)$. Indeed, it can be verified that if $\tau * \delta_{1} \leq \min \left\{ \delta_{2},\delta_{3} \right\}$, then $(h_{b}(\tau*\delta_{1})-h_{b}(\delta_{1}),1-h_{b}(\delta_{2}),1-h_{b}(\delta_{3})) \in \beta_{1}(p_{\underlineSemiPrivateRV\PrivateRV_{1} X \underlineY})$, where $p_{\underlineSemiPrivateRV \PrivateRV_{1}X}=p_{V_{1}}p_{U_{21}}p_{U_{31}}1_{\left\{X_{1}=V_{1}\right\}}1_{\left\{X_{2}=U_{21}\right\}}1_{\left\{X_{3}=U_{31}\right\}}$, $p_{U_{21}}(1)=p_{U_{31}}(1)=\frac{1}{2}$ and $p_{V_{1}}(1)=\tau$.

\subsubsection{Non-additive example} 
We now present a non-additive
example for which we analytically prove strict sub-optimality of
$\UM$technique. 
\begin{example}
\label{Ex:A3To1-OR-BC}
Consider the $3-$DBC $(\InputAlphabet,\underlineSetY,W_{\underlineY|X},\underline{\kappa})$ depicted in figure \ref{Fig:ORBC}, where $\InputAlphabet \define \left\{ 0,1 \right\}\times \left\{ 0,1 \right\} \times \left\{ 0,1 \right\}, \OutputAlphabet_{1}=\OutputAlphabet_{2}=\OutputAlphabet_{3}=\left\{ 0,1 \right\}, W_{\underlineY|X}(y_{1},y_{2},y_{3}|x_{1}x_{2}x_{3})=BSC_{\delta_{1}}(y_{1}|x_{1}\oplus (x_{2}\vee x_{3}))BSC_{\delta_{2}}(y_{2}|x_{2})BSC_{\delta_{3}}(y_{3}|x_{3})$, where $\delta_{j} \in (0,\frac{1}{2}):j=1,2,3$, $BSC_{\eta}(1|0)=BSC_{\eta}(0|1)=1-BSC_{\eta}(0|0)=1-BSC_{\eta}(1|1)=\eta$ for any $\eta \in (0,\frac{1}{2})$ and the cost function $\underline{\kappa} = (\kappa_{1},\kappa_{2},\kappa_{3})$, where $\kappa_{j} (x_{1}x_{2}x_{3}) = 1_{\left\{x_{j}=1\right\}}$.
\end{example}
\begin{figure}
\centering
 \includegraphics[width=1.65in]{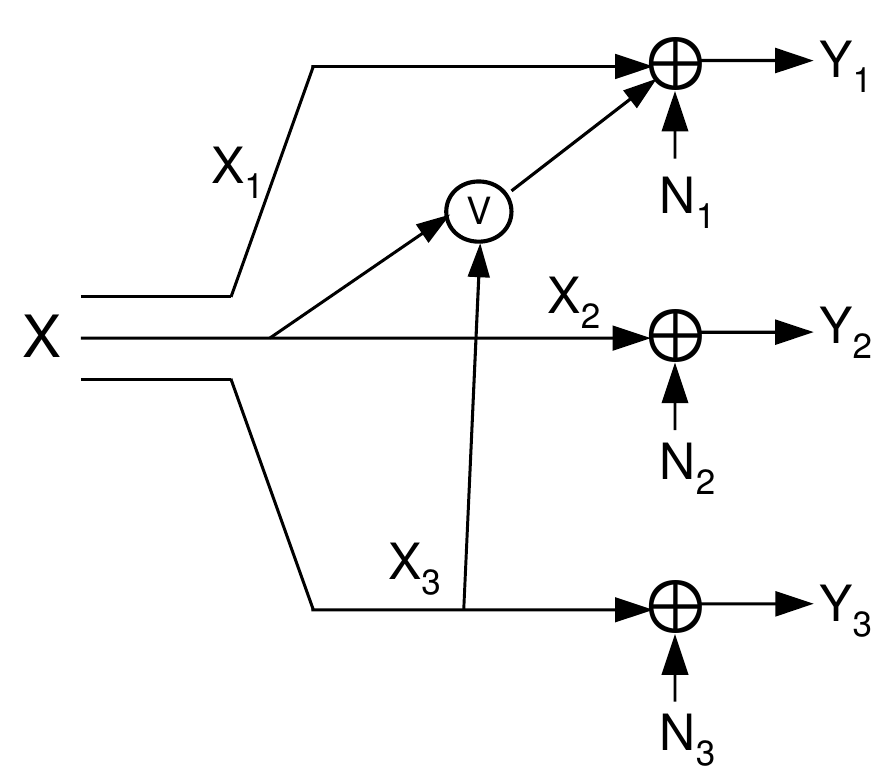}
\caption{The $3-$BC described in example \ref{Ex:A3To1-OR-BC}.}
\label{Fig:ORBC}
\end{figure}

We begin by stating the conditions for sub-optimality of $\UM$technique.
\begin{lemma}
 \label{Lem:ConditionsForSub-OptimalityOfMartonTechnique}
Consider example \ref{Ex:A3To1-OR-BC} with $\delta \define \delta_{2}=\delta_{3} \in (0,\frac{1}{2})$ and $\tau \define \tau_{2}=\tau_{3} \in (0,\frac{1}{2})$. Let $\beta \define \delta_{1}*(2\tau -\tau^{2})$. The rate triple $(h_{b}(\tau_{1} * \delta_{1})-h_{b}(\delta_{1}),h_{b}(\tau * \delta)-h_{b}(\delta),h_{b}(\tau * \delta)-h_{b}(\delta)) \notin \alpha_{\mathscr{U}}(\underline{\tau})$ if
\begin{equation}
 \label{Eqn:3To1-OR-ICConditionForStrictSubOptimalityOfMarton}
h_{b}(\tau_{1} * \delta_{1})-h_{b}(\delta_{1})+2(h_{b}(\tau * \delta)-h_{b}(\delta))> h_{b}(\tau_{1}(1-\beta)+(1-\tau_{1})\beta)-h_{b}(\delta_{1}).
\end{equation}
\end{lemma}
\begin{proof}
Please refer to appendix \ref{AppSec:UpperBoundMartonCodingTechniuque}
 \end{proof}
We now derive conditions under which $(h_{b}(\tau_{1} * \delta_{1})-h_{b}(\delta_{1}),h_{b}(\tau * \delta)-h_{b}(\delta),h_{b}(\tau * \delta)-h_{b}(\delta)) \in \beta_{1}(W_{\underline{Y}|X},\underline{\kappa},\underline{\tau})$.
\begin{lemma}
 \label{Lem:CndsRteTrplAchvblUsingCosets}
Consider example \ref{Ex:A3To1-OR-BC} with $\delta \define
\delta_{2}=\delta_{3} \in (0,\frac{1}{2})$ and $\tau \define
\tau_{2}=\tau_{3} \in (0,\frac{1}{2})$. Let $\beta \define
\delta_{1}*(2\tau -\tau^{2})$. 
The rate triple $(h_{b}(\tau_{1} *
\delta_{1})-h_{b}(\delta_{1}),h_{b}(\tau *
\delta)-h_{b}(\delta),h_{b}(\tau * \delta)-h_{b}(\delta)) \in
\beta_{1}(W_{\underline{Y}|X},\underline{\kappa},\underline{\tau})$
i.e., achievable using coset codes, if, 
\begin{eqnarray}
 \label{Eqn:3To1-OR-BCConditionForAchievabilityUsingCosetCodes}
h_{b}(\tau * \delta)-h_{b}(\delta)\leq  \theta ,
\end{eqnarray}
where $\theta
=h_{b}(\tau)-h_{b}((1-\tau)^{2})-(2\tau-\tau^{2})h_{b}(\frac{\tau^{2}}{2\tau-\tau^{2}})-h_{b}(\tau_{1}
* \delta_{1})+h_{b}(\tau_{1}*\beta)$. Moreover $\mathbb{C}_1(\underline{\tau})=h_b(\tau_1*\delta_1)-h_b(\delta_1)$.
\end{lemma}
\begin{proof}
 The proof only involves identifying the appropriate test channel
 $p_{\underline{U}V_{1}} \in
 \mathbb{D}_{1}(W_{\underline{Y}|X},\underline{\kappa},\underline{\tau})$. Let
 $\TimeSharingRVSet = \phi$ be empty,
 $\SemiPrivateRVSet_{21}=\SemiPrivateRVSet_{31}=\mathcal{F}_{3}$.
 Let $p_{X_{1}}(1)=1-p_{X_{1}}(0)=\tau_{1}$. Let
 $p_{U_{j1}X_{j}}(0,0)=1-p_{U_{j1}X_{j}}(1,1)=1-\tau$ and therefore
 $P(U_{j1}=2)=P(X_{j} \neq U_{j})=0$ for $j=2,3$. It is easily
 verified that $p_{\underline{U}V_{1}\ulineInputRV\ulineOutputRV} \in
 \mathbb{D}_{1}(W_{\underline{Y}|X},\underline{\kappa},\underline{\tau})$,
 i.e, in particular respects the cost constraints. 

The choice of this test channel, particularly the ternary field, is motivated by $H(X_{2}\vee X_{3}|U_{21}\oplus_{3} U_{31})=0$. The decoder $1$ can reconstruct the interfering pattern after having decoded the ternary sum of the codewords. It maybe verified that for this test channel $p_{QU_{21}U_{31}\ulineInputRV\ulineOutputRV}$, $\beta_{1}(p_{QU_{21}U_{31}\ulineInputRV\ulineOutputRV})$ is defined as the set of rate triples $(R_{1},R_{2},R_{3}) \in [0,\infty)^{3}$ that satisfy
\begin{eqnarray}
 \label{Eqn:3To1ORICDescriptionOfTheRateRegionForTernaryTestChannelChosenForCosetCodes}
&R_{1} < \min\left\{0,\theta \right\}+h_{b}(\tau_{1}*\delta_{1})-h_{b}(\delta_{1}), ~~~
R_{j} < h_{b}(\tau * \delta)-h_{b}(\delta):j=2,3\nonumber\\
&R_{1}+R_{j} < h_{b}(\tau_{1}*\delta_{1})-h_{b}(\delta_{1})+\theta,
\end{eqnarray}
where $\theta$ is as defined in the statement of the lemma. Clearly, $(h_{b}(\tau_{1} * \delta_{1})-h_{b}(\delta_{1}),h_{b}(\tau * \delta)-h_{b}(\delta),h_{b}(\tau * \delta)-h_{b}(\delta)) \in \cocl(\beta_{1}(p_{\underlineSemiPrivateRV\PrivateRV_{1} X \underlineY}))$ if (\ref{Eqn:3To1-OR-BCConditionForAchievabilityUsingCosetCodes}) is satisfied.
Using standard information-theoretic arguments, one can easily
establish that $\mathbb{C}_1(\underline{\tau}) \leq h_b(\tau_1 *
\delta_1)-h_b(\delta_1)$. This completes the proof.
\end{proof}

Conditions (\ref{Eqn:3To1-OR-ICConditionForStrictSubOptimalityOfMarton}) and (\ref{Eqn:3To1-OR-BCConditionForAchievabilityUsingCosetCodes}) are \textit{not} mutually exclusive. It maybe verified that the choice $\tau_{1}=\frac{1}{90}$, $\tau=0.15$, $\delta_{1}=0.01$ and $\delta=0.067$ satisfies both conditions. We therefore conclude the existence of non-additive $3-$DBC's for which PCC yield strictly larger achievable rate regions.
We extract the key elements of lemmas
\ref{Lem:ConditionsForSub-OptimalityOfMartonTechnique} and
\ref{Lem:CndsRteTrplAchvblUsingCosets} in the following theorem. 

\begin{thm}
 \label{Thm:3To1ORICConcludingStatement}
For a vector $3-$DBC studied in example \ref{Ex:A3To1-OR-BC} that
satisfies
(\ref{Eqn:3To1-OR-ICConditionForStrictSubOptimalityOfMarton}) and
(\ref{Eqn:3To1-OR-BCConditionForAchievabilityUsingCosetCodes}),
linear coding technique achieves
$\mathbb{C}_1(\underline{\tau})=h_b(\tau_1 *\delta_1)-h_b(\delta_1)$,
and  $\UM$technique cannot achieve this performance. 
In particular, for the choice $\tau_{1}=\frac{1}{90}$, $\tau=0.15$,
$\delta_{1}=0.01$ and $\delta=0.067$, these conditions are satisfied. 
\end{thm}

\subsection{Step II: Incorporating private codebooks}
 \label{SubSec:AnEnlargedAchievableRateRegion}
We revisit the coding technique proposed in section \ref{SubSec:DecodingSumOfCodewordsUsingNestedCosetCodes}. Observe that (i) user $1$ decodes a sum of the entire codewords transmitted to users $2$ and $3$ and (ii) users $2$ and $3$ decode only their respective codewords. This technique may be enhanced in the following way. User $1$ can decode the sum of \textit{one component} of user $2$ and $3$ signals each. In other words, we may include private codebooks for users $2$ and $3$.

Specifically, in addition to auxiliary alphabet sets $\PrivateRVSet_{1},\SemiPrivateRVSet_{2},\SemiPrivateRVSet_{3}$ introduced in section \ref{SubSec:DecodingSumOfCodewordsUsingNestedCosetCodes}, let $\PrivateRVSet_{2},\PrivateRVSet_{3}$ denote arbitrary finite sets and $p_{U_{2}U_{3}V_{1}V_{2}V_{3}}$ denote a PMF on $\SemiPrivateRVSet_{2} \times \SemiPrivateRVSet_{3} \times \PrivateRVSet_{1}\times \PrivateRVSet_2 \times \PrivateRVSet_3$. For $j=2,3$, consider a random codebook $\mathcal{C}_{j} \subseteq \mathcal{V}_{j}^{n}$ of rate $K_{j}+L_{j}$ whose codewords are independently chosen according to $p^{n}_{V_{j}}$. Codewords of $\mathcal{C}_{j}$ are independently and uniformly partitioned into $\exp \left\{ nL_{j}\right\}$ bins. The distribution induced on $\mathcal{C}_{1},\Lambda_{2},\Lambda_{3}$ is identical to that in section \ref{SubSec:DecodingSumOfCodewordsUsingNestedCosetCodes}. Moreover, the triplet $\mathcal{C}_{2},\mathcal{C}_{3}, (\mathcal{C}_{1},\Lambda_{2},\Lambda_{3})$ are mutually independent.\footnote{Here $(\mathcal{C}_{1},\Lambda_{2},\Lambda_{3})$ is treated as a single random object.} Having specified the distribution of codewords of $\mathcal{C}_{j}:j=2,3$, we have thus specified the distribution of quintuple of random codebooks. Messages of users' $2$ and $3$ are split into two parts each. One part of user $2$'s ($3$'s) message, of rate $T_{2}\log\pi$ ($T_{3}\log\pi$), index a bin in $\Lambda_{2}$ ($\Lambda_{3}$), and the other part, of rate $L_{2}$ ($L_{3}$), index a bin in $\mathcal{C}_{2}$ ($\mathcal{C}_{3}$). User $1$'s message indexes a bin in $\mathcal{C}_{1}$. The encoder looks for a quintuple of jointly typical codewords with respect to $p_{\underline{U}\underlinePrivateRV}$, in the quintuple of indexed bins. Following a second moment method similar to that employed in appendix \ref{AppSec:AnalysisOfEncoderErrorEventFor3To1DBC}, it can be proved that the encoder finds at least one jointly typical triple if
\begin{eqnarray}
 \label{Eqn:ManyToOneSourceCodingBounds}
(S_{A}-T_{A})\log\pi+K_{B} &>& |A|\log_{2} \pi + \sum_{b \in B} H(V_{b}) - H(U_{A},V_{B})\footnotemark \\
\label{Eqn:SourceCodingBoundInvolvingThatSillySumTerm}
\max \{ S_{2},S_{3}\}\log\pi +K_{B} &>& \log\pi +  \sum_{b \in B} H(V_{b}) - \min_{\theta \in \fieldpi\setminus\{0\}}H(U_{2}\oplus \theta U_{3},V_{B})
\end{eqnarray}\footnotetext{We remind the reader that the empty sum has value $0$, i.e, $\sum_{a \in \phi}=0$}
for all $A \subseteq \left\{2,3\right\}, B \subseteq \left\{ 1,2,3 \right\}$, where $S_{A} = \sum _{j \in A}S_{j}$, $K_{B} = \sum_{b \in B}K_{b}$, $U_{A} = (U_{j}:j \in A)$ and $V_{B}=(V_{b}:b \in B)$.\footnote{Recall that $\fieldpi=\mathcal{U}_{2}=\mathcal{U}_{3}$.} Having chosen one such jointly typical quintuple, say $(\SemiPrivateRV_{2}^{n},\SemiPrivateRV_{3}^{n},\underlinePrivateRV^{n})$, the encoder generates a vector $X^{n}$ according to $p^{n}_{X|\underlinePrivateRV\SemiPrivateRV_{2}\SemiPrivateRV_{3}}(\cdot | \underlinePrivateRV^{n},\SemiPrivateRV_{2}^{n},\SemiPrivateRV_{3}^{n})$ and inputs the same on the channel.

The operations of decoders $2$ and $3$ are identical and we describe one of them. Decoder $3$ receives $Y_{3}^{n}$ and looks for all pairs of codewords in the Cartesian product $\Lambda_{3} \times \mathcal{C}_{3}$ that are jointly typical with $Y_{3}^{n}$ with respect to $p_{U_{3}V_{3}Y_{3}}$. If all such pairs belong to a unique pair of bins, the corresponding pair of bin indices is declared as the decoded message of user $3$. Else an error is declared. It can be proved that if
\begin{eqnarray}
 \label{Eqn:UpperboundOnPublicCodebook}
S_{j}\log\pi < \log_{2} \pi-H(U_{j}|V_{j},Y_{j}),&&
K_{j}+L_{j} < H(V_{j})-H(V_{j}|Y_{j},U_{j})\\
\label{Eqn:UpperboundOnPublicAndPrivateCodebook}
S_{j}\log\pi+K_{j}+L_{j} &<& \log_{2} \pi+H(V_{j})-H(V_{j},U_{j}|Y_{j})
\end{eqnarray}
for $j=2,3$, then probability of users $2$ or $3$ decoding into an incorrect message falls exponentially with $n$.

Operation of decoder $1$ is identical to that described in section \ref{SubSec:DecodingSumOfCodewordsUsingNestedCosetCodes}. If (\ref{Eqn:BoundsOnUser1Rates}) holds, then probability of error at decoder 1 falls exponentially with $n$. Substituting $R_{1}=K_{1}, R_{2}=T_{2}\log\pi+L_{2}, R_{3}=T_{3}\log\pi+L_{3}$ and eliminating $S_{2}\log\pi,S_{3}\log\pi,K_{1},K_{2},K_{3}$ in (\ref{Eqn:BoundsOnUser1Rates}), (\ref{Eqn:ManyToOneSourceCodingBounds})-(\ref{Eqn:UpperboundOnPublicAndPrivateCodebook}) yields an achievable rate region. We provide a mathematical characterization of this achievable rate region.
\begin{definition}
 \label{Defn:TestChannelsForTwoLayersWithOneUserDecodingInterference}
Let $\SetOfDistributions_{2}^{f}(W_{\underlineY|X},\kappa,\tau)$ denote the collection of PMFs $p_{Q\SemiPrivateRV_{2}\SemiPrivateRV_{3}\PrivateRV_{1}\PrivateRV_{2}\PrivateRV_{3} X \underlineY}$ defined on $\mathcal{Q}\times \SemiPrivateRVSet_{2}\times \SemiPrivateRVSet_{3}\times \PrivateRVSet_{1} \times \PrivateRVSet_{2} \times\PrivateRVSet_{3} \times \InputAlphabet \times \underlineSetY$, where (i) $\SemiPrivateRVSet_{2}=\SemiPrivateRVSet_{3}=\fieldpi$ is the finite field of cardinality $\pi$, $\mathcal{Q},\PrivateRVSet_{1},\PrivateRVSet_{2}, \PrivateRVSet_{3}$ are finite sets, (ii) $p_{\underlineY|X\underlinePrivateRV\underlineSemiPrivateRV Q}=p_{\underlineY|X}=W_{\underlineY|X}$, and (iii) $\Expectation\left\{ \kappa(X) \right\} \leq \tau$. For $p_{Q\underlineSemiPrivateRV\underlinePrivateRV X \underlineY} \in \SetOfDistributions_{2}^{f}(W_{\underlineY|X},\kappa,\tau)$, let $\beta_{2}(p_{Q\underlineSemiPrivateRV\underlinePrivateRV X \underlineY})$ be defined as the set of triples $(R_{1},R_{2},R_{3}) \in [0,\infty)^{3}$ for which there exists 
nonnegative numbers $S_{2},T_{2},S_{3},T_{3},K_{j},L_{j}:j=1,2,3$ such that $R_{1}=K_{1}, R_{2}=T_{2}\log\pi+L_{2}, R_{3}=T_{3}\log\pi+L_{3}$,
\begin{eqnarray}
& (S_{A}-T_{A})\log\pi+K_{B} > |A|\log_{2} \pi + \sum_{b \in B} H(V_{b}|Q) - H(U_{A},V_{B}|Q),\footnotemark \nonumber\\
&\max \{ S_{2},S_{3}\}\log\pi +K_{B} > \log\pi +  \sum_{b \in B} H(V_{b}|Q) - \min_{\theta \in \fieldpi\setminus\{0\}}H(U_{2}\oplus \theta U_{3},V_{B}|Q), \nonumber\\
&K_{1}\!+\!R_{1} \!< \!I(V_{1};U_{2}\oplus U_{3},Y_{1}|Q),~~~
K_{1}\!+\!R_{1}\!+\!S_{j}\!\log\pi < \log \pi +H(V_{1}|Q)\!-\! H(\PrivateRV_{1},\SemiPrivateRV_{2}\oplus \SemiPrivateRV_{3}|Q,Y_{1}):j=2,3,\nonumber\\
&S_{j}\log\pi < \log_{2} \pi-H(U_{j}|Q,V_{j},Y_{j}):j=2,3,~~
K_{j}+L_{j} < H(V_{j}|Q)-H(V_{j}|Q,Y_{j},U_{j}):j=2,3\nonumber\\
&S_{j}\log\pi+K_{j}+L_{j} < \log_{2} \pi+H(V_{j}|Q)-H(V_{j},U_{j}|Q,Y_{j}):j=2,3\nonumber
\end{eqnarray}
for all $A \subseteq \left\{2,3\right\}, B \subseteq \left\{ 1,2,3 \right\}$, where $S_{A} = \sum _{j \in A}S_{j}$, $K_{B} = \sum_{b \in B}K_{b}$, $U_{A} = (U_{j}:j \in A)$ and $V_{B}=(V_{b}:b \in B)$. Let
\begin{eqnarray}
 \label{Eqn:TwoLayerAchievableRateRegionWithOneUserDecodingInterference}
 \beta_{2}(W_{\underlineY|X},\kappa,\tau) = \cocl\left(\underset{\substack{p_{Q\underlineSemiPrivateRV\underlinePrivateRV X \underlineY} \\\in \SetOfDistributions_{2}^{f}(W_{\underlineY|X},\kappa,\tau)}
}{\bigcup}\beta_{2}(p_{Q\underlineSemiPrivateRV\underlinePrivateRV X \underlineY})\right).\nonumber
\end{eqnarray}
\end{definition}
\begin{thm}
 \label{Thm:AchievableRateRegionsFor3To1BCUsingNestedCosetCodes}
For a $3-$DBC $(\mathcal{X},\underlineSetY,W_{\underlineY|X},\kappa)$, $\beta_{2}(W_{\underlineY|X},\kappa,\tau)$ is achievable, i.e., $\beta_{2}(W_{\underlineY|X},\kappa,\tau) \subseteq \mathbb{C}(W_{\underlineY|X},\kappa,\tau)$.
\end{thm}
The proof is similar to that of theorem \ref{Thm:OneLayerAchievableRateRegionWithOneUserDecodingInterference}. The only differences being (i) the encoder looks for a quintuple of codewords instead of a triple, and (ii) decoders $2$ and $3$ decode from a pair of codebooks. These can be handled using the techniques developed in proof of \ref{Thm:OneLayerAchievableRateRegionWithOneUserDecodingInterference}. The reader in need of an elaboration is referred to \cite[Thm. 5]{201403arXiv_PadPra}.

\subsection{Step III: $\PCC$region : Using PCC to manage interference over a $3-$DBC}
\label{Sec:AchievableRateRegionFor3BCUsingNestedCosetCodes}
Here we employ PCC to manage interference seen by each receiver. In the sequel, we propose a simple extension of the technique presented in section \ref{SubSec:AnEnlargedAchievableRateRegion} to enable each user decode a bivariate interference component.
Throughout the following discussion $i,j,k$ denote distinct indices in $\left\{ 1,2,3 \right\}$. Let $\SemiPrivateRVSet_{ji}=\fieldpii, \SemiPrivateRVSet_{jk}=\fieldpik$ be finite fields and $\PrivateRVSet_{j}$ be an arbitrary finite set. User $j$ splits it's message $M_{j}$ into three parts $(M_{ji}^{U},M_{jk}^{U},M_{j}^{V})$ of rates $T_{ji}\log\pi_{i}, T_{jk}\log\pi_{k},L_{j}$ respectively. User $j$'s message indexes three codebooks - $\mathcal{C}_{j}, \Lambda_{ji},\Lambda_{jk}$ - whose structure is described in the following. Consider a random codebook $\mathcal{C}_{j} \subseteq \mathcal{V}_{j}^{n}$ of rate $K_{j}+L_{j}$ whose codewords are independently chosen according to $p^{n}_{V_{j}}$. Codewords of $\mathcal{C}_{j}$ are independently and uniformly partitioned into $\exp \left\{ nL_{j}\right\}$ bins. Consider random PCC $(n,nS_{ji},nT_{ji},G_{ji},B_{ji}^{n},I_{ji})$ and $(n,nS_{jk},nT_{jk},G_{jk},B_{jk}^{n},I_{jk})$ denoted $\Lambda_{ji}$ and $\Lambda_{jk}$ respectively. Observe that PCC $\Lambda_{ji}$ and $\Lambda_{ki}$ are built over the same finite field $\fieldpii$. The corresponding linear codes are nested, i.e., if $S_{ji} \leq S_{ki}$, then $G_{ki}^{t} = \left[ G_{ji}^{t} ~ G_{ki/ji}^{t}\right]$ where $G_{ki/ji} \in \fieldpi^{n(S_{ji}-S_{ki}) \times n}$, and vice versa. We have thus specified the structure of $9$ random codebooks. We now specify the distribution of these random codebooks.

The random PCCs are independent of $\mathcal{C}_{j}:j=1,2,3$. $\mathcal{C}_{1},\mathcal{C}_{2},\mathcal{C}_{3}$ are mutually independent. We now specify the distribution of the PCCs. The triplet $(\Lambda_{12},\Lambda_{32}),(\Lambda_{21},\Lambda_{31}),(\Lambda_{23},\Lambda_{13})$ are mutually independent. All of the bias vectors are mutually independent and uniformly distributed. The collection of generator matrices is independent of the collection of bias vectors. We only need to specify the distribution of the generator matrices. The rows of the larger of the two generator matrices $G_{ji}$ and $G_{ki}$ are uniformly and independently distributed. This specifies the distribution of the $9$ random codebooks.

$M_{ji}^{U}$,$M_{jk}^{U}$ and $M_{j}^{V}$ index bins in $\Lambda_{ji}$, $\Lambda_{jk}$ and $\mathcal{C}_{j}$ respectively. The encoder looks for a collection of $9$ codewords from the indexed bins that are jointly typical with respect to a PMF $p_{\underlineSemiPrivateRV \underlinePrivateRV}$ defined on $\underlineSemiPrivateRVSet \times \underlinePrivateRVSet$.\footnote{$\underlineSemiPrivateRV$ abbreviates $\SemiPrivateRV_{12}\SemiPrivateRV_{13}\SemiPrivateRV_{21}\SemiPrivateRV_{23}\SemiPrivateRV_{31}\SemiPrivateRV_{32}$.} We now state the bounds that ensure the probability of encoder not finding a jointly typical collection of codewords from the indexed bins. We introduce some notation to aid reduce clutter. Throughout the following, in every instance $i,j,k$ will denote distinct indices in $\{1,2,3\}$. For every $A \subseteq \{ 12,13,21,23,31,32 \}, B \subseteq \{1,2,3\}, C\subseteq \{1,2,3\}$, let $S_{A} = \sum _{jk \in A}S_{jk}, M_{B}\define \sum_{j \in B} \max\{ S_{ij}+T_{ij},S_{kj}+T_{kj}\}, K_{C} = \sum_{c \in C}K_{c}$. For every $B \subseteq \{1,2,3\}$, let $A({B}) = \cup_{j \in B}\{ ji,jk\}$.
Following a second moment method similar to that employed in appendix \ref{AppSec:AnalysisOfEncoderErrorEventFor3To1DBC}, it can be proved that the encoder finds at least one jointly typical collection if (\ref{Eqn:ManyToManySourceCodingBounds}) is satisfied for all $A \subseteq \left\{12,13,21,23,31,32\right\}, B \subseteq \left\{ 1,2,3 \right\}, C \subseteq \left\{ 1,2,3 \right\}$, that satisfy $A \cap A({B}) = \phi$, where $U_{A} = (U_{jk}:jk \in A)$ and $V_{C}=(V_{c}:c \in C)$. Having chosen one such jointly typical collection, say $(\underlineSemiPrivateRV^{n},\underlinePrivateRV^{n})$, the encoder generates a vector $X^{n}$ according to $p^{n}_{X|\underlineSemiPrivateRV\underlinePrivateRV}(\cdot | \underlineSemiPrivateRV^{n},\underlinePrivateRV^{n})$ and feeds the same as input on the channel.

Decoder $j$ receives $Y_{j}^{n}$ and looks for all triples $(u_{ji}^{n},u_{jk}^{n},v_{j}^{n})$ of codewords in $\lambda_{ji} \times \lambda_{jk} \times \mathcal{C}_{j}$ such that there exists a $u^{n}_{ij\oplus kj} \in (\lambda_{ij} \oplus \lambda_{kj})$ such that $(u_{ij\oplus kj}^{n},u_{ji}^{n},u_{jk}^{n},v_{j}^{n},Y_{j}^{n})$ are jointly typical with respect to $p_{U_{ij}\oplus U_{kj},U_{ji},U_{jk},V_{j},Y_{j}}$. If it finds all such triples in a unique triple of bins, the corresponding triple of bin indices is declared as decoded message of user $j$. Else an error is declared. The probability of error at decoder $j$ can be made arbitrarily small for sufficiently large block length if (\ref{Eqn:ManyToManyBCChannelCodingBounds}) holds for every $\mathcal{A}_{j} \subseteq \left\{ ji,jk\right\}$ with distinct indices $i,j,k$ in $\left\{ 1,2,3 \right\}$, where $S_{\mathcal{A}_{j}} \define \sum_{a \in \mathcal{A}_{j}}S_{a}, T_{\mathcal{A}_{j}} \define \sum_{a \in \mathcal{A}_{j}}T_{a}, U_{\mathcal{A}_{j}} = (U_{a}:a \in \mathcal{A}_{j})$.
.
Recognize that user $j$'s rate
$R_{j}=T_{ji}\log\pi_{i}+T_{jk}\log\pi_{k}+L_{j}$. We are now equipped
to state $\PCC$region for a general $3-$DBC. 

\begin{definition}
 \label{Defn:CollectionOfTestChannelsForCommunicatingOverGeneral3BCUsingNCC}
Let $\mathbb{D}^{f}(W_{\underlineY|X},\kappa,\tau)$ denote the collection of probability mass functions $p_{Q\underlineSemiPrivateRV\underlinePrivateRV X \underlineY}$ defined on $\mathcal{Q} \times \underlineSemiPrivateRVSet \times \underlinePrivateRVSet \times \mathcal{X} \times \underlineSetY$, where (i) $\mathcal{Q}, \mathcal{V}_{1},\mathcal{V}_{2}, \mathcal{V}_{3}$ are arbitrary finite sets,  $\underlinePrivateRVSet \define \PrivateRVSet_{1} \times \PrivateRVSet_{2} \times \PrivateRVSet_{3}$, (ii) $\SemiPrivateRVSet_{ij}=\fieldpij$\footnote{Recall $\fieldpij$ is the finite field of cardinality $\pi_{j}$.} for each $1\leq  i,j\leq3$, and $\underlineSemiPrivateRVSet \define \SemiPrivateRVSet_{12} \times \SemiPrivateRVSet_{13} \times \SemiPrivateRVSet_{21} \times \SemiPrivateRVSet_{23} \times \SemiPrivateRVSet_{31} \times \SemiPrivateRVSet_{32}$,
(iii) $\underlinePrivateRV\define (\PrivateRV_{1},\PrivateRV_{2},\PrivateRV_{3})$ and $\underlineSemiPrivateRV\define (\SemiPrivateRV_{12},\SemiPrivateRV_{13},\SemiPrivateRV_{21},\SemiPrivateRV_{23},\SemiPrivateRV_{31},\SemiPrivateRV_{32})$, such that (i) $p_{\underlineY| X \underlinePrivateRV \underlineSemiPrivateRV}=p_{\underlineY |X}=W_{\underlineY |X}$, (ii) $\Expectation\left\{ \kappa(X) \right\} \leq \tau$.

For $p_{\underlineSemiPrivateRV\underlinePrivateRV X \underlineY} \in \mathbb{D}^{f}(W_{\underlineY|X},\kappa,\tau)$, let $\beta(p_{\underlineSemiPrivateRV\underlinePrivateRV X \underlineY})$ be defined as the set of rate triples $(R_{1},R_{2},R_{3}) \in [0,\infty)^{3}$ for which there exists nonnegative numbers $S_{ij},T_{ij}:ij \in \left\{12,13,21,23,31,32 \right\}, K_{j},L_{j}:j\in \left\{ 1,2,3\right\}$ such that $R_{1}=T_{12}\log\pi_{2}+T_{13}\log\pi_{3}+L_{1}, R_{2}=T_{21}\log\pi_{1}+T_{23}\log\pi_{3}+L_{2}, R_{3}=T_{31}\log\pi_{1}+T_{32}\log\pi_{2}+L_{3}$ and
\begin{equation}\begin{aligned}\label{Eqn:ManyToManySourceCodingBounds}
 \lefteqn{S_{A}+M_{B}+K_{C} >\Theta(A,B,C)\mbox{ where,}}\\
 &\displaystyle\!\!\!\!\Theta(A,B,C) \define \max_{(\theta_{j}:j \in B) \in \underset{j \in B}{\prod} \fieldpij} \{\sum_{a \in A}\log |\mathcal{U}_{a}| + \sum_{j \in B}\log \pi_{j}+\sum_{c \in C} H(V_{c}|Q) - H(U_{A},U_{ji}\oplus \theta_{j}U_{jk}:j \in B,V_{C}|Q)\}\!\!\!\!\!
\end{aligned} \nonumber\end{equation}
for all $A \subseteq \left\{12,13,21,23,31,32\right\}, B \subseteq \left\{ 1,2,3 \right\}, C \subseteq \left\{ 1,2,3 \right\}$, that satisfy $A \cap A(B) = \phi$, where $A({B}) = \cup_{j \in B}\{ ji,jk\}$, $U_{A} = (U_{jk}:jk \in A)$, $V_{C}=(V_{c}:c \in C)$, $S_{A} = \sum _{jk \in A}S_{jk}, M_{B}\define \sum_{j \in B} \max\{ S_{ij}+T_{ij},S_{kj}+T_{kj}\}, K_{C} = \sum_{c \in C}K_{c}$, and
\begin{equation}
\begin{aligned}
\label{Eqn:ManyToManyBCChannelCodingBounds}
\lefteqn{S_{\mathcal{A}_{j}}+T_{\mathcal{A}_{j}} \leq \sum_{a \in \mathcal{A}_{j}}\log |\mathcal{U}_{a}| - H(U_{\mathcal{A}_{j}}|Q,U_{\mathcal{A}_{j}^{c}},U_{ij}\oplus U_{kj},V_{j},Y_{j})}
\\
\lefteqn{S_{\mathcal{A}_{j}}+T_{\mathcal{A}_{j}}+S_{ij}+T_{ij} \leq \sum_{a \in \mathcal{A}_{j}}\log |\mathcal{U}_{a}| + \log \pi_{j} - H(U_{\mathcal{A}_{j}},U_{ij}\oplus
U_{kj}|Q,U_{\mathcal{A}_{j}^{c}},V_{j},Y_{j})} \\
&S_{\mathcal{A}_{j}}+T_{\mathcal{A}_{j}}+S_{kj}+T_{kj} \leq \sum_{a \in \mathcal{A}_{j}}\log |\mathcal{U}_{a}| + \log \pi_{j} - H(U_{\mathcal{A}_{j}},U_{ij}\oplus
U_{kj}|Q,U_{\mathcal{A}_{j}^{c}},V_{j},Y_{j}) \\
\lefteqn{S_{\mathcal{A}_{j}}+T_{\mathcal{A}_{j}}+K_{j}+L_{j} \leq \sum_{a \in \mathcal{A}_{j}}\log |\mathcal{U}_{a}|+H(V_{j}) - H(U_{\mathcal{A}_{j}},V_{j}|Q,U_{\mathcal{A}_{j}^{c}},U_{ij}\oplus
U_{kj},Y_{j})} 
\end{aligned} \nonumber
\end{equation}
\begin{equation}
\begin{aligned}
\lefteqn{S_{\mathcal{A}_{j}}+T_{\mathcal{A}_{j}}+K_{j}+L_{j}+S_{ij}+T_{ij} \leq \sum_{a \in \mathcal{A}_{j}}\log |\mathcal{U}_{a}| + \log \pi_{j}+H(V_{j}) -
H(U_{\mathcal{A}_{j}},V_{j},U_{ij}\oplus U_{kj}|Q,U_{\mathcal{A}_{j}^{c}},Y_{j})} \\
&S_{\mathcal{A}_{j}}+T_{\mathcal{A}_{j}}+K_{j}+L_{j}+S_{kj}+T_{kj} \leq \sum_{a \in \mathcal{A}_{j}}\log |\mathcal{U}_{a}| + \log \pi_{j}+H(V_{j}) - H(U_{\mathcal{A}_{j}},V_{j},U_{ij}\oplus
U_{kj}|Q,U_{\mathcal{A}_{j}^{c}},Y_{j}), 
\end{aligned} \nonumber
\end{equation}
for every $\mathcal{A}_{j} \subseteq \left\{ ji,jk\right\}$ with
distinct indices $i,j,k$ in $\left\{ 1,2,3 \right\}$, where
$S_{\mathcal{A}_{j}} \define \sum_{a \in \mathcal{A}_{j}}S_{a},
T_{\mathcal{A}_{j}} \define \sum_{a \in \mathcal{A}_{j}}T_{a},
U_{\mathcal{A}_{j}} = (U_{a}:a \in \mathcal{A}_{j})$. Let $\PCC$rate
region be defined as 
\begin{eqnarray}
 \label{Eqn:AchievableRateRegionFor3BCUsingNestedCosetCodes}
 \beta(W_{\underlineY|X},\kappa,\tau) = \cocl\left(\underset{\substack{p_{Q\underlineSemiPrivateRV\underlinePrivateRV X \underlineY} \in\\ \mathbb{D}^{f}(W_{\underlineY|X},\kappa,\tau)}
}{\bigcup}\beta(p_{Q\underlineSemiPrivateRV\underlinePrivateRV X \underlineY})\right).\nonumber
\end{eqnarray}
\end{definition}
\begin{thm}
 \label{Thm:AchievableRateRegionFor3BCUsingNestedCosetCodes}
For $3-$DBC $(\mathcal{X},\underlineSetY,W_{\underlineY|X},\kappa)$,
$\PCC$region $\beta(W_{\underlineY|X},\kappa,\tau)$ is achievable, i.e.,
$\beta(W_{\underlineY|X},\kappa,\tau) \subseteq
\mathbb{C}(W_{\underlineY|X},\kappa,\tau)$. 
\end{thm}
All the non-trivial elements of this proof being illustrated in considerable detail in the context of proof of theorem \ref{Thm:OneLayerAchievableRateRegionWithOneUserDecodingInterference}, we omit a proof of theorem \ref{Thm:AchievableRateRegionFor3BCUsingNestedCosetCodes}.
\begin{remark}
The $\PCC$region is a continuous function of the channel transition probability matrix. Therefore, gains obtained by the proposed coding technique are robust to small perturbations of $3-$DBC.
\end{remark}

\section{Enlarging $\UM$region using partitioned coset codes}
\label{Sec:EnlargingMarton'sRateRegionUsingNestedCosetCodes}
The natural question that arises is whether $\PCC$region $\beta(W_{\underlineY|X},\kappa,\tau)$
contains $\UM$region
$\alpha_{\mathscr{U}}(W_{\underlineY|X},\kappa,\tau)$. The coding
techniques based on structured codes do not substitute those based on unstructured codes, but
enhance the latter. Indeed, the technique proposed by K\"orner and
Marton \cite{197903TIT_KorMar}, in the context of distributed source
coding, is strictly suboptimal to that studied by Berger and Tung
\cite{Berger-MSC} if the function is not sufficiently compressive,
i.e., entropy of the sum is larger than one half of the joint entropy
of the sources.\footnote{If $X$ and $Y$ are the distributed binary
  sources whose modulo$-2$ sum is to be reconstructed at the decoder,
  then K\"orner and Marton technique is strictly suboptimal if $H(X
  \oplus Y) > \frac{H(X,Y)}{2}$.} The penalty paid in terms of the
binning rate for endowing structure is not sufficiently compensated
for by the function. This was recognized by Ahlswede and Han
\cite[Section VI]{198305TIT_AhlHan} for the problem studied by
K\"orner and Marton. 

We follow the approach of Ahlswede and Han \cite[Section
  VI]{198305TIT_AhlHan} to build upon $\UM$region by gluing to it the
coding technique proposed herein. In essence the coding techniques
studied in section \ref{SubSec:NaturalExtensionOfMartonTo3BC} and
\ref{Sec:AchievableRateRegionFor3BCUsingNestedCosetCodes} are glued
together.\footnote{This is akin to the use of superposition and
  binning in Marton's coding.} Indeed, a description of the resulting
rate region is quite involved and we do not provide it's characterization. The resulting coding technique will involve each user split
it's message into six parts - one public and private part each, two
semi-private and \textit{bivariate} parts each. This can be understood
by splitting the message as proposed in sections
\ref{SubSec:NaturalExtensionOfMartonTo3BC} and
\ref{Sec:AchievableRateRegionFor3BCUsingNestedCosetCodes} and
identifying the private parts. In essence each user decodes a
univariate component of every other user's transmission particularly
set apart for it, and furthermore decodes a bivariate component of the
other two user's transmissions.\footnote{An informed and inquisitive
  reader may begin to see a relationship emerge between the several
  layers of coding and common parts of a collection of random
  variables. Please refer to section \ref{Sec:ConcludingRemarks} for a
  discussion.} Please refer to figure
\ref{Fig:3BCMapOfRandomVariables} for an illustration of the coding
technique. Herein, $V$ denotes the private part, $U$, the
bivariate part,  $T$, the semi-private part and $W$, the public part.

\begin{figure}
\centering
\includegraphics[width=6.5in]{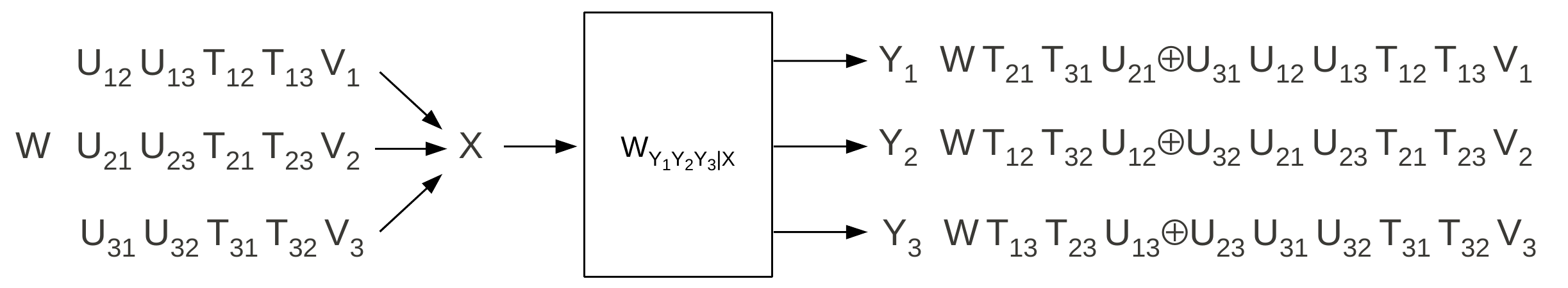}
\caption{Illustration of coding technique that incorporates unstructured and coset codes.}
\label{Fig:3BCMapOfRandomVariables}
\end{figure}

\section{Concluding remarks : Common parts of random variables and the need for structure}
\label{Sec:ConcludingRemarks}
Let us revisit Marton's coding technique for $2-$BC. Define the pair $\overline{\PrivateRV_{j}}\define (\PublicRV, \PrivateRV_{j}):j=1,2$ of random variables decoded by the two users and let $\overline{\PrivateRVSet_{j}}\define \PublicRVSet \times \PrivateRVSet_{j}:j=1,2$. Let us stack the collection of compatible codewords as $\overline{\PrivateRVSet_{1}}^{n} \times \overline{\PrivateRVSet_{2}}^{n}$. The encoder can work with this stack, being  oblivious to the distinction between $\PublicRVSet$ and $\PrivateRVSet_{j}:j=1,2$. In other words, it does not recognize that a symbol over $\overline{V_{j}}$ is indeed a pair of symbols. A few key observations of this stack of codewords is in order. Recognize that many pairs of compatible codewords agree in their `$\PublicRVSet-$coordinate'. In other words, they share the same codeword on the $\PublicRVSet-$codebook. $\PublicRV$ is the common part \cite{1972MMPCT_GacKor} of the pair $(\overline{\PrivateRV_{1}},\overline{\PrivateRV_{2}})$. Being a common part, it can be realized through univariate functions. Let us say $W=f_{1}(V_{1})=f_{2}(V_{2})$. This indicates that \textit{$\PublicRVSet-$codebook is built such that, the range of these univariate functions when applied on the collection of codewords in this stack, is contained.}

How did Marton accomplish this containment? Marton proposed building
the $\PublicRV-$codebook first, followed by conditional codebooks over
$\PrivateRV_{1},\PrivateRV_{2}$. Conditional coding with a careful
choice of order therefore contained the range under the action of
univariate function. How is all of this related to the need for
containing bivariate functions of a pair of random variables? The
fundamental underlying thread is the notion of common part
\cite{1972MMPCT_GacKor}. What are the common parts of a triple of
random variables? Clearly, one can simply extend the notion of common
part defined for a pair of random variables. This yields four common
parts - one part that is simultaneously common to all three random
variables and one common part corresponding to each pair in the
triple. 
Indeed, if
$\overline{V_{1}}=(W,U_{12},U_{31},V_{1}),\overline{V_{2}}=(W,U_{12},U_{23},V_{2}),\overline{V_{3}}=(W,U_{23},U_{31},V_{3})$,
then $W$ is the part simultaneously to common to
$\overline{V_{1}},\overline{V_{2}},\overline{V_{3}}$ and
$U_{ij}:ij\in\left\{12,23,31\right\}$ are the pairwise common
parts. The $\UM$technique suggests a way to handle these common parts.

This does not yet answer the need for containment under bivariate
function. We recognize a richer notion of common part for
a triple of random variables. Indeed, three nontrivial binary random
variables $X,Y, Z=X\oplus Y$ have no common parts as defined earlier.
Yet, the degeneracy in the joint probability
matrix hints at a common part. Indeed, they possess a
\textit{conferencing} common part. For example, the pair $(X,Y),Z$
have a common part. In other words, there exists a \textit{bivariate}
function of $X,Y$ and a univariate function of $Z$ that agree with
probability $1$. Containment of this bivariate function brings in the
need for structured codes. Indeed, the resemblance to the problem
studied by K\"orner and Marton \cite{197903TIT_KorMar} is striking. We
therefore believe the need for structured codes for three (multi) user
communication problems is closely linked to the notion of common parts
of a triple (collection) of random variables. Analogous to conditional
coding that contained univariate functions, endowing codebooks with
structure is an inherent need to carefully handle additional degrees
of freedom prevalent in larger dimensions. 

\section{Proof of theorem \ref{Thm:ConditionsUnderWhichMartonCodingTechniqueDoesNotAchieveRateTriple}}
\label{Sec:StrictSubOptimalityOfMartonCodingTechnique}

In this section, we prove strict sub-optimality of $\UM$technique for the $3-$DBC presented in example \ref{Ex:3-BCExample}. In particular, we prove that if parameters $\tau,\delta_{1},\delta_{2},\delta_{3}$ are such that $1+h_{b}(\delta_{1} * \tau) > h_{b}(\delta_{2})+h_{b}(\delta_{3})$ and $(R_{1},1-h_{b}(\delta_{2}),1-h_{b}(\delta_{3})) \in \alpha_{\mathscr{U}}(\tau)$, then $R_{1} < h_{b}(\tau*\delta_{1})-h_{b}(\delta_{1})$.

Why is $\UM$technique suboptimal for the case described above. As mentioned in section \ref{Sec:3To1BCAndTheNeedForStructuredCodes}, in this case, receiver $1$ is unable to decode the pair of codewords transmitted to users $2$ and $3$. Furthermore, based on unstructured independent coding, it does not attempt to decode a function of transmitted codewords - in this case the modulo$-2$ sum. This forces decoder $1$ to be content by decoding only individual components of user $2$ and $3$'s transmissions, leaving residual uncertainty in the interference. The encoder helps out by precoding for this residual uncertainty. However, as a consequence of the cost constraint on $X_{1}$, it is forced to live with a rate loss.

Our proof traces through the above arguments in three stages and is therefore instructive. In the first stage, we characterize all test channels $p_{\TimeSharingRV\PublicRV \underlineSemiPrivateRV \underlinePrivateRV X \underlineY}$ for which $(R_{1},1-h_{b}(\delta_{2}),1-h_{b}(\delta_{3})) \in \alpha_{\mathscr{U}}(p_{\TimeSharingRV\PublicRV \underlineSemiPrivateRV \underlinePrivateRV X \underlineY})$. This stage enables us identify `active' codebooks, their corresponding rates and characterize two upper bounds on $R_{1}$. One of these contains the rate loss due to precoding. In the second stage, we therefore characterize the condition under which there is no rate loss. As expected, it turns out that there is no rate loss only if decoder $1$ has decoded codewords of users $2$ and $3$. This gets us to the third stage, where we conclude that $1+h_{b}(\delta_{1} * \tau) > h_{b}(\delta_{2})+h_{b}(\delta_{3})$ precludes this possibility. The first stage is presented in lemma \ref{Lem:CharacterizationOfTestChannelThatContainsTheRateTriple}, second stage is stated in lemma \ref{Lem:CharacterizationOfConditionThatEnsuresNoRateLoss} and proved in appendices \ref{AppSec:CharacterizationForNoRateLossInPTP-STx} and \ref{AppSec:TheBinaryAdditiveDirtyPointToPointChannelSuffersARateLoss}. Third stage can be found in arguments following lemma \ref{Lem:CharacterizationOfConditionThatEnsuresNoRateLoss}.

We begin with a characterization of a test channel
$p_{\TimeSharingRV\PublicRV\underlineSemiPrivateRV\underlinePrivateRV
  X\underlineY}$ for which
$(R_{1},1-h_{b}(\delta_{2}),1-h_{b}(\delta_{3})) \in
\alpha_{\mathscr{U}}(p_{\TimeSharingRV\PublicRV\underlineSemiPrivateRV\underlinePrivateRV
  X\underlineY})$. Since independent information needs to be
communicated to users $2$ and $3$ at their respective PTP capacities,
it is expected that their codebooks are not precoded for each other's
signal, and moreover none of users $2$ and $3$ decode a part of the
other users' signal. The following lemma establishes this. We remind
the reader that $X_{1}X_{2}X_{3}=X$ denote the three digits at
the input, where $Y_{j}$, the output at receiver $j$ is obtained by
passing $X_{j}$ through a BSC with cross over probability
$\delta_{j}$ for $j=2,3$.
$Y_{1}$ is obtained by passing $X_{1}\oplus X_{2}\oplus X_{3}$
through a BSC with cross over probability $\delta_{1}$. Moreover, the
binary symmetric channels (BSCs) are independent. Input symbol
$X_{1}$ is constrained with respect to a Hamming cost function and the
constraint on the average cost per symbol is $\tau$. Formally,
$\kappa(x_{1}x_{2}x_{3})=1_{\left\{ x_{1}=1 \right\}}$ is the cost
function and the average cost per symbol is not to exceed $\tau$. 

\begin{lemma}
 \label{Lem:CharacterizationOfTestChannelThatContainsTheRateTriple}
If there exists a test channel $p_{\TimeSharingRV\PublicRV\underlineSemiPrivateRV\underlinePrivateRV X\underlineY} \in \SetOfDistributions_{\mathscr{U}}(\tau)$ and nonnegative numbers $K_{i},S_{ij},K_{ij},L_{ij},S_{i},T_{i}$ that satisfy (\ref{Eqn:3BCSourceCodingBoundNonnegativity})-(\ref{Eqn:3BCChannelCodingSextupleBound}) for each triple $(i,j,k) \in \left\{ (1,2,3),(2,3,1),(3,1,2) \right\}$ such that $R_{2}=K_{2}+K_{23}+L_{12}+T_{2}=1-h_{b}(\delta_{2}),R_{3}=K_{3}+K_{31}+L_{23}+T_{3}=1-h_{b}(\delta_{3})$, then
\begin{enumerate}
 \item \label{Item:RatesOfPublicAndSemiPrivateCodebooks}$K_1=K_2=K_3=K_{23}=L_{23}=K_{12}=L_{31}=S_2=S_3=0$ and $I(\SemiPrivateRV_{31} V_{1}V_{3};Y_{2}|\TimeSharingRV WU_{23}U_{12}V_{2})=0$,
\item \label{Item:BinningRatesOfSemiPrivateCodebooks}$S_{31}=I(U_{31};U_{23}|\TimeSharingRV W),S_{12}=I(U_{12};U_{23}|\TimeSharingRV W)$, $S_{23}=I(U_{12};U_{31}|\TimeSharingRV WU_{23})=0$,
\item \label{Item:WU23IsConditionallyIndependentOfY2Y3GivenQ}$I(V_{2}U_{12};V_{3}U_{31}|\TimeSharingRV WU_{23})=0$, $I(WU_{23};Y_{j}|\TimeSharingRV )=0:j=2,3,I(V_{2}U_{12};Y_{2}|\TimeSharingRV W U_{23})=1-h_{b}(\delta_{2})$ and $I(V_{3}U_{31};Y_{3}|\TimeSharingRV W U_{23})=1-h_{b}(\delta_{3})$, 
\item \label{ItemNumber:MarkovChains} $(V_3,X_3,V_1,U_{31}) - (\TimeSharingRV WU_{23}U_{12}V_2) - (X_2,Y_2) $ and $(V_2,X_2,V_1,U_{12}) - (\TimeSharingRV WU_{23}U_{31}V_3) - (X_3,Y_3)$ are Markov chains,
\item \label{Item:X2AndX3AreConditionallyIndependentGivenQWU12U23U31}$X_{2} - \TimeSharingRV \PublicRV \SemiPrivateRV_{12} \SemiPrivateRV_{23} \SemiPrivateRV_{31} - X_{3}$ is a Markov chain,
\item \label{Item:U31AndX2AreConditionallyIndependentGivenQWU12U23}$\SemiPrivateRV_{12}-\TimeSharingRV \PublicRV \SemiPrivateRV_{23} \SemiPrivateRV_{31}-X_{3}$ and $\SemiPrivateRV_{31}-\TimeSharingRV \PublicRV \SemiPrivateRV_{23} \SemiPrivateRV_{12}-X_{2}$ are Markov chains.
\end{enumerate}
\end{lemma}
\begin{proof}
 Substituting (i) $(2,3,1)$ for $(i,j,k)$ in (\ref{Eqn:3BCChannelCodingSextupleBound}), (ii) $(1,2,3)$ for $(i,j,k)$ in (\ref{Eqn:3BCSourceCodingPairwiseBound}) and combining the resulting bounds yields 
\begin{eqnarray}
I(WU_{23}U_{12}V_{2};Y_{2}|\TimeSharingRV) &\geq& I(WU_{23}U_{12}V_{2};Y_{2}|\TimeSharingRV)+I(U_{12};U_{23}|W,Q)-S_{12}-S_{23} \nonumber\\&\geq& R_{2}+K_{3}+K_{1}+L_{23}+K_{12}+S_{2}\geq R_{2}=1-h_{b}(\delta_{2}),\label{Eqn:UpperBoundOnR_2PlusABunchOfTerms}
\end{eqnarray}
where the second inequality follows from non-negativity of $K_{3},K_{1},L_{23},K_{12},S_{2}$. Moreover, 
\begin{eqnarray}
\label{Eqn:UsingDataProcessingInequalityInDerivingUpperBoundOnR2}
1-h_{b}(\delta_{2}) &\geq&  I(X_{2};Y_{2}) = I(\TimeSharingRV \PublicRV \underlineSemiPrivateRV\underlinePrivateRV X_{1}Y_{1}X_{3}Y_{3}X_{2};Y_{2}) \geq I(W \SemiPrivateRV_{23} \SemiPrivateRV_{12} \PrivateRV_{2} ; Y_{2}|\TimeSharingRV) \\
\label{Eqn:SubstitutingUpperBoundOnR_2}
&\geq& R_{2}\!+\! K_{3}\!+\! K_{1}\!+\! L_{23}\!+\! K_{12}\!+\! S_{2} \geq R_{2}=1-h_{b}(\delta_{2}),
\end{eqnarray}
where (i) equality in (\ref{Eqn:UsingDataProcessingInequalityInDerivingUpperBoundOnR2}) follows from Markov chain $\TimeSharingRV \PublicRV \underlineSemiPrivateRV \underlinePrivateRV X_{1}Y_{1}X_{3}Y_{3}- X_{2} -Y_{2}$. Since all the terms involved are non-negative, equality holds through the above chain of inequalities to yield
\begin{eqnarray}
\label{Eqn:BinningRatesOfUser2SemiPrivateCodebooks}
&S_{12}+S_{23}=I(U_{12};U_{23}|\TimeSharingRV W),K_{1}\!=\!K_{3}\!=\!L_{23}\!=\!K_{12}\!=\!S_{2}\!=\!
I(\TimeSharingRV;Y_{2})\!=\!0\\
\label{Eqn:K3K1L23K12S2AreZero}
&I(U_{31}V_{1}X_{1}Y_{1}V_{3}X_{3}Y_{3}X_{2};Y_{2}|QWU_{12}U_{23}V_{2})\!=\!0\!\\
\label{Eqn:PartOfmarkovChainUsedInGettingToRateLoss}
&\mbox{ and therefore }
(V_1,V_3,X_3,U_{31}) - (\TimeSharingRV WU_{12}U_{23}V_2) - Y_2\mbox{ is a Markov chain}
\end{eqnarray}
where the first equality in (\ref{Eqn:BinningRatesOfUser2SemiPrivateCodebooks}) follows from condition for equality in the first inequality of (\ref{Eqn:UpperBoundOnR_2PlusABunchOfTerms}). The above sequence of steps are repeated by substituting (i) $(3,1,2)$ for $(i,j,k)$ in (\ref{Eqn:3BCChannelCodingSextupleBound}), (ii) $(2,3,1)$ for $(i,j,k)$ in (\ref{Eqn:3BCSourceCodingPairwiseBound}). It can be verified that
\begin{eqnarray} 
\label{Eqn:BinningRatesOfUser3SemiPrivateCodebooks}
&S_{31}+S_{23}=I(U_{31};U_{23}|\TimeSharingRV W), K_{1}\!=\!K_{2}\!=\!L_{31}\!=\!K_{23}\!=\!S_{3}\!=\!I(\TimeSharingRV; Y_{3})\!=\!0,\\
\label{Eqn:K2K1L31K23S3AreZero}
&I(\SemiPrivateRV_{12} V_{1}X_{1}Y_{1}V_{2}X_{2}Y_{2}X_{3};Y_{3}|\TimeSharingRV WU_{23}U_{31}V_{3})\!=\!0\!\\
\label{Eqn:MarkovChainFoundInTestChannelThatIncludesY2ForUser2}
&\mbox{ and therefore }
(V_1,V_2,X_2,U_{12}) - (\TimeSharingRV WU_{23}U_{31}V_3) - Y_3\mbox{ is a Markov chain}.
\end{eqnarray}
The second set of equalities in (\ref{Eqn:BinningRatesOfUser2SemiPrivateCodebooks}), (\ref{Eqn:BinningRatesOfUser3SemiPrivateCodebooks}) lets us conclude
\begin{equation}
 \label{Eqn:RatesOfUsers12And3InTermsOfSimplifiedCodebooks}
R_{1}=T_{1},R_{2}=L_{12}+T_{2} \mbox{ and }R_{3}=K_{31}+T_{3}.
\end{equation}
From $I(U_{12};U_{23}|\TimeSharingRV W) + I(U_{31};U_{23}|\TimeSharingRV W) = S_{12}+S_{23}+S_{31}+S_{23}$, and (\ref{Eqn:3BCSourceCodingTripleBound}), we have $I(U_{12};U_{23}|\TimeSharingRV W) + I(U_{31};U_{23}|\TimeSharingRV W) \geq I(U_{12};U_{23};U_{31}|\TimeSharingRV W)+S_{23}$. The non-negativity of $S_{23}$ implies $S_{23}=0$ and $I(U_{31};U_{12}|\TimeSharingRV WU_{23})=0$. We therefore conclude
\begin{eqnarray}
\label{Eqn:SemiPrivateLayerBinningRates}
S_{12}=I(U_{12};U_{23}|\TimeSharingRV W),S_{31}=I(U_{31};U_{23}|\TimeSharingRV W),S_{23}=0,
I(U_{31};U_{12}|\TimeSharingRV W U_{23})=0
\end{eqnarray}
Substituting (\ref{Eqn:BinningRatesOfUser2SemiPrivateCodebooks}), (\ref{Eqn:BinningRatesOfUser3SemiPrivateCodebooks}), (\ref{Eqn:SemiPrivateLayerBinningRates}) in (\ref{Eqn:3BCSourceCodingQuadrapleBound}) for $(i,j,k)=(2,3,1)$ and $(i,j,k)=(3,1,2)$ and (\ref{Eqn:3BCSourceCodingPentaBound}) for $(i,j,k)=(2,3,1)$, we obtain
\begin{eqnarray}
\label{Eqn:IndependenceOfV2WithU31AndV3WithU12ConditionedOnWU23}
I(V_{2};U_{31}|\TimeSharingRV WU_{12}U_{23})=I(V_{3};U_{12}|\TimeSharingRV WU_{23}U_{31})
=I(V_{2};V_{3}|\TimeSharingRV WU_{12}U_{23}U_{31})=0.
\end{eqnarray}
(\ref{Eqn:IndependenceOfV2WithU31AndV3WithU12ConditionedOnWU23}) and last equality in (\ref{Eqn:SemiPrivateLayerBinningRates}) yield 
\begin{equation}
\label{Eqn:V2U12AndV3U31AreIndependentConditionedOnWU23}
I(V_{2}U_{12};V_{3}U_{31}|\TimeSharingRV WU_{23})=0.
\end{equation}
Substituting (\ref{Eqn:RatesOfUsers12And3InTermsOfSimplifiedCodebooks}), (\ref{Eqn:SemiPrivateLayerBinningRates}) in (\ref{Eqn:3BCChannelCodingDoubleBound})with $(i,j,k)=(2,3,1)$ yields the upper bound $R_2 \leq I(U_{12}V_2;Y_2|\TimeSharingRV WU_{23})$. Since
\begin{eqnarray}
&\!\!\!\!1-h_{b}(\delta_{2}) \!=\! R_{2} \!\leq\! I(\SemiPrivateRV_{12} \PrivateRV_{2};Y_{2}|\TimeSharingRV\PublicRV \SemiPrivateRV_{23}) \!\leq\! I(\PublicRV \SemiPrivateRV_{12}\SemiPrivateRV_{23} \PrivateRV_{2};Y_{2}|\TimeSharingRV) \leq 1-h_{b}(\delta_{2}),\nonumber
\end{eqnarray}
where the last inequality follows from (\ref{Eqn:UsingDataProcessingInequalityInDerivingUpperBoundOnR2}), equality holds in all of the above inequalities to yield $I(W\SemiPrivateRV_{23};Y_{2}|\TimeSharingRV )=0$ and $I(\SemiPrivateRV_{12} \PrivateRV_{2};Y_{2}|\TimeSharingRV\PublicRV \SemiPrivateRV_{23})=1-h_{b}(\delta_{2})$. A similar argument proves $I(W\SemiPrivateRV_{23};Y_{3}|\TimeSharingRV )=0$ and $I(\SemiPrivateRV_{31} \PrivateRV_{3};Y_{3}|\TimeSharingRV\PublicRV \SemiPrivateRV_{23})=1-h_{b}(\delta_{3})$.

We have proved the Markov chains in items (1)-(3). In order to prove Markov chains in item \ref{ItemNumber:MarkovChains}, we prove the following lemma.
\begin{lemma}
\label{Lem:ForBSCMarokovChainWithOutputImpliesMarkovChainWithInput}
If $A,B,X,Y$ are discrete random variables such that (i) $X,Y$ take values in $\left\{0,1\right\}$ with $P(Y=0|X=1)=P(Y=1|X=0)=\eta \in (0,\frac{1}{2})$, (ii) $A-B-Y$ and $AB-X-Y$ are Markov chains, then $A-B-XY$ is also a Markov chain.
\end{lemma}
Please refer to appendix
\ref{AppSec:ForBSCMarokovChainWithOutputImpliesMarkovChainWithInput}
for a proof. Markov chains in
(\ref{Eqn:PartOfmarkovChainUsedInGettingToRateLoss}),
(\ref{Eqn:MarkovChainFoundInTestChannelThatIncludesY2ForUser2}) in
conjunction with lemma
\ref{Lem:ForBSCMarokovChainWithOutputImpliesMarkovChainWithInput}
establishes Markov chains in item \ref{ItemNumber:MarkovChains}. 

(\ref{Eqn:V2U12AndV3U31AreIndependentConditionedOnWU23}) and (\ref{Eqn:K3K1L23K12S2AreZero}) imply $I(\SemiPrivateRV_{31}\PrivateRV_{3};\SemiPrivateRV_{12}\PrivateRV_{2}Y_{2}|\TimeSharingRV\PublicRV \SemiPrivateRV_{23})=0$. This in conjunction with (\ref{Eqn:K2K1L31K23S3AreZero}) implies
\begin{equation}
 \label{Eqn:U31V3Y3AndU12V2Y2AreIndependentConditionedOnWU23}
I(\SemiPrivateRV_{31}\PrivateRV_{3}Y_{3};\SemiPrivateRV_{12}\PrivateRV_{2}Y_{2}|\TimeSharingRV \PublicRV \SemiPrivateRV_{23})=0\mbox{ and thus }\SemiPrivateRV_{31}\PrivateRV_{3}Y_{3}-\TimeSharingRV \PublicRV \SemiPrivateRV_{23}-\SemiPrivateRV_{12}\PrivateRV_{2}Y_{2}\mbox{ is a Markov chain.}
\end{equation}
(\ref{Eqn:U31V3Y3AndU12V2Y2AreIndependentConditionedOnWU23}) implies that $\SemiPrivateRV_{31}Y_{3}-\TimeSharingRV \PublicRV \SemiPrivateRV_{23}-\SemiPrivateRV_{12}Y_{2}$ is a Markov chain, and therefore $Y_{3}-\TimeSharingRV \PublicRV \SemiPrivateRV_{12}\SemiPrivateRV_{23}\SemiPrivateRV_{31}-Y_{2}$ is a Markov chain. Employing lemma \ref{Lem:ForBSCMarokovChainWithOutputImpliesMarkovChainWithInput} twice we observe $Y_{3}X_{3}-\TimeSharingRV\PublicRV \SemiPrivateRV_{12}\SemiPrivateRV_{23}\SemiPrivateRV_{31}-X_{2}Y_{2}$ is a Markov chain and furthermore $X_{3}-\TimeSharingRV \PublicRV \SemiPrivateRV_{12}\SemiPrivateRV_{23}\SemiPrivateRV_{31}-X_{2}$ is a Markov chain, thus proving item \ref{Item:X2AndX3AreConditionallyIndependentGivenQWU12U23U31}.

Finally, we prove Markov chains in item \ref{Item:U31AndX2AreConditionallyIndependentGivenQWU12U23}. From Markov chain $(V_3,X_3,V_1,U_{31}) - (\TimeSharingRV WU_{23}U_{12}V_2) - (X_2,Y_2)$ proved in item \ref{ItemNumber:MarkovChains}, we have $I(X_{2};U_{31}|\TimeSharingRV WU_{23}U_{12}V_{2})=0$. From (\ref{Eqn:V2U12AndV3U31AreIndependentConditionedOnWU23}), we have $I(V_{2};U_{31}|\TimeSharingRV WU_{23}U_{12})=0$. Summing these two, we have $I(X_{2}V_{2};U_{31}|\TimeSharingRV WU_{23}U_{12})=0$ and therefore $I(X_{2};U_{31}|\TimeSharingRV WU_{23}U_{12})=0$ implying the Markov chain $X_{2}-\TimeSharingRV WU_{23}U_{12}-U_{31}$. Similarly, we get the Markov chain $X_{3}-\TimeSharingRV WU_{23}U_{31}-U_{12}$.
\end{proof}
Lemma \ref{Lem:CharacterizationOfTestChannelThatContainsTheRateTriple} enables us to simplify the bounds (\ref{Eqn:3BCSourceCodingBoundNonnegativity})-(\ref{Eqn:3BCChannelCodingSextupleBound}) for the particular test channel under consideration. Substituting (\ref{Eqn:BinningRatesOfUser2SemiPrivateCodebooks})-(\ref{Eqn:IndependenceOfV2WithU31AndV3WithU12ConditionedOnWU23}) in (\ref{Eqn:3BCSourceCodingBoundNonnegativity})-(\ref{Eqn:3BCChannelCodingSextupleBound}) and employing statements of lemma \ref{Lem:CharacterizationOfTestChannelThatContainsTheRateTriple}, we conclude that if $(R_{1},1-h_{b}(\delta_{2}),1-h_{b}(\delta_{3})) \in \alpha_{\mathscr{U}}(p_{\TimeSharingRV\PublicRV\underlineSemiPrivateRV\underlinePrivateRV X\underlineY})$, then there exists nonnegative numbers $S_{1},T_{1},L_{12},K_{31}$ that satisfy $R_{1}=T_{1},R_{2}=L_{12}+T_{2}=1-h_{b}(\delta_{2}),R_{3}=K_{31}+T_{3}=1-h_{b}(\delta_{3})$,
%
\begin{eqnarray}
 \label{Eqn:BoundsOnUser1CodebookAsAResultOfUser1Decoding}
&S_1  \geq I(V_1;U_{23}V_2V_3|\TimeSharingRV WU_{12}U_{31}),~~~T_1+S_1\leq I(V_1;Y_1|\TimeSharingRV WU_{12}U_{31})\\
\label{Eqn:BoundsOnRatesOfUser123CodebooksAsAResultOfUser1Decoding}
&L_{12}+K_{31}+T_1+S_1 \leq I(U_{12};U_{31}|\TimeSharingRV W)-I(U_{23};U_{12}|\TimeSharingRV W)+I(V_1U_{12}U_{31};Y_1|\TimeSharingRV W)-I(U_{23};U_{31}|\TimeSharingRV W)\\
\label{Eqn:BoundsOnRatesOfUser2CodebookAsAResultOfUser2Decoding}
&0 \leq T_{2} \leq I(V_2;Y_2|\TimeSharingRV WU_{12}U_{23}),~~~1-h_{b}(\delta_{2})=T_{2}+L_{12} = I(U_{12}V_2;Y_2|\TimeSharingRV WU_{23})\\
\label{Eqn:BoundsOnRatesOfUser3CodebookAsAResultOfUser3Decoding}
&0 \leq T_{3}\leq I(V_3;Y_3|\TimeSharingRV WU_{31}U_{23}),~~~1-h_{b}(\delta_{3})=T_{3}+K_{31}= I(U_{31}V_3;Y_3|\TimeSharingRV WU_{23}).
\end{eqnarray}
(\ref{Eqn:BoundsOnRatesOfUser2CodebookAsAResultOfUser2Decoding}), (\ref{Eqn:BoundsOnRatesOfUser3CodebookAsAResultOfUser3Decoding}) imply
\begin{equation}
 \label{Eqn:LowerBoundsOnRateOfCodebooksSharedWithUser1}
L_{12}\geq I(U_{12};Y_{2}|\TimeSharingRV \PublicRV U_{23}),~~~~~~~~~~~ K_{31} \geq I(U_{31};Y_3|\TimeSharingRV WU_{23}),
\end{equation}
(\ref{Eqn:BoundsOnUser1CodebookAsAResultOfUser1Decoding}) implies
\begin{eqnarray}
T_{1}\!\!\!\!&=&\!\!\!\!R_1 \leq I(V_1;Y_1|\TimeSharingRV W\SemiPrivateRV_{12}\SemiPrivateRV_{31})
 -I(V_1;U_{23}V_2V_3|\TimeSharingRV WU_{12}U_{31}), \nonumber\\
\label{Eqn:UpperBoundOnRateOfUser1ThatContainsRateLoss}
&\leq&\!\!\!\! I(V_1;Y_1U_{23}|\TimeSharingRV W\SemiPrivateRV_{12}\SemiPrivateRV_{31})
 -I(V_1;U_{23}V_2V_3|\TimeSharingRV WU_{12}U_{31})=I(V_1;Y_1|\TimeSharingRV W\underlineSemiPrivateRV)
 -I(V_1;V_2V_3|\TimeSharingRV W\underlineSemiPrivateRV),
\end{eqnarray}
and (\ref{Eqn:BoundsOnRatesOfUser123CodebooksAsAResultOfUser1Decoding}) in conjunction with (\ref{Eqn:LowerBoundsOnRateOfCodebooksSharedWithUser1}), and the lower bound on $S_{1}$ in (\ref{Eqn:BoundsOnUser1CodebookAsAResultOfUser1Decoding}) imply
\begin{eqnarray}
R_1 \!\!\!\!\!\!&\leq&\!\!\!\! I(U_{12}U_{31}V_1;Y_1|\TimeSharingRV W)
-I(V_1;U_{23}V_2V_3|\TimeSharingRV WU_{12}U_{31})-I(U_{12};Y_2|\TimeSharingRV WU_{23})-I(U_{31};Y_3|\TimeSharingRV WU_{23}) \nonumber\\
&&\!\!\!\!+I(U_{12};U_{31}|\TimeSharingRV W)-I(U_{23};U_{12}|\TimeSharingRV W)-I(U_{23};U_{31}|\TimeSharingRV W)\nonumber\\
\!\!\!\!&\leq&\!\!\!\! I(U_{12}U_{31}V_{1};Y_1 U_{23}|\TimeSharingRV W)
-I(V_1;U_{23}V_2V_3|\TimeSharingRV WU_{12}U_{31})-I(U_{12};Y_2|\TimeSharingRV WU_{23})-I(U_{31};Y_3|\TimeSharingRV WU_{23})\nonumber\\
&&\!\!\!\!+I(U_{12};U_{31}|\TimeSharingRV W)-I(U_{23};U_{12}|\TimeSharingRV W)-I(U_{23};U_{31}|\TimeSharingRV W)\nonumber\\\label{Eqn:BoundOnUser1RateAsAConsequenceOfDecoderDecoding3Codebooks}
\!\!\!\!\!\!\!\!&=&\!\!\!\!I(V_1;Y_1|\TimeSharingRV
W\underlineSemiPrivateRV) \!-\!I(V_1;V_2V_3|\TimeSharingRV
W\underlineSemiPrivateRV)\! +\!I(U_{12}U_{31};Y_1|\TimeSharingRV W
U_{23}) \!-\!I(U_{12};Y_2|\TimeSharingRV WU_{23}) \nonumber \\
&&  - I(U_{31};Y_3|\TimeSharingRV WU_{23}\!),
\end{eqnarray}
where (\ref{Eqn:BoundOnUser1RateAsAConsequenceOfDecoderDecoding3Codebooks}) follows from the last equality in (\ref{Eqn:SemiPrivateLayerBinningRates}). Combining (\ref{Eqn:UpperBoundOnRateOfUser1ThatContainsRateLoss}) and (\ref{Eqn:BoundOnUser1RateAsAConsequenceOfDecoderDecoding3Codebooks}), we have
\begin{equation}
 \label{Eqn:3To1DBCAnalyticalConvTwoUpperBoundsCombined}
R_{1} \leq I(V_1;Y_1|\TimeSharingRV W\underlineSemiPrivateRV)
-I(V_1;V_2V_3|\TimeSharingRV W\underlineSemiPrivateRV) + \min
\left\{ \begin{array}{c} 0,I(U_{12}U_{31};Y_1|\TimeSharingRV W U_{23})
  -I(U_{12};Y_2|\TimeSharingRV WU_{23})\\-I(U_{31};Y_3|\TimeSharingRV WU_{23})
\end{array} 
\right\}.
\end{equation}
We have thus obtained (\ref{Eqn:UpperBoundOnRateOfUser1ThatContainsRateLoss}) and (\ref{Eqn:BoundOnUser1RateAsAConsequenceOfDecoderDecoding3Codebooks}), two upper bounds on $R_{1}$ we were seeking, and this concludes the first stage of our proof. In the sequel, we prove the minimum of the above upper bounds on $R_{1}$ is strictly lesser than $h_{b}(\tau * \delta_{1})-h_{b}(\delta_{1})$. Towards, that end, note that upper bound (\ref{Eqn:UpperBoundOnRateOfUser1ThatContainsRateLoss}) contains the rate loss due to precoding. In the second stage, we work on (\ref{Eqn:UpperBoundOnRateOfUser1ThatContainsRateLoss}) and derive conditions under which there is \textit{no} rate loss.

Markov chains of item (\ref{ItemNumber:MarkovChains}) in lemma \ref{Lem:CharacterizationOfTestChannelThatContainsTheRateTriple} imply $V_{1}-QW\underlineSemiPrivateRV V_{2}V_{3}-X_{2}$ and $V_{1}-QW\underlineSemiPrivateRV V_{2}V_{3}X_{2}-X_{3}$ are Markov chains. Therefore, $I(V_{1};X_{2}|QW\underlineSemiPrivateRV V_{2}V_{3})=0$ and $I(V_{1};X_{3}|QW\underlineSemiPrivateRV V_{2}V_{3}X_{2})=0$. Summing these, we have $I(V_{1};X_{2}X_{3}|QW\underlineSemiPrivateRV V_{2}V_{3})=0$. Employing this in (\ref{Eqn:UpperBoundOnRateOfUser1ThatContainsRateLoss}),we note
\begin{eqnarray}
\label{eq:inequality1}
R_1 &\leq& I(V_1;Y_1|\TimeSharingRV W\underlineSemiPrivateRV) -I(V_1;V_2V_3|\TimeSharingRV W\underlineSemiPrivateRV) =I(V_1;Y_1|\TimeSharingRV W\underlineSemiPrivateRV) -I(V_1;V_2V_3X_2X_3|\TimeSharingRV W\underlineSemiPrivateRV) \\
&\leq& I(V_1;Y_1|\TimeSharingRV W\underlineSemiPrivateRV) - I(V_1;X_2,X_3|\TimeSharingRV W\underlineSemiPrivateRV) \leq I(V_1;Y_1 |\TimeSharingRV W\underlineSemiPrivateRV) -I(V_1;X_2\oplus X_3|\TimeSharingRV W\underlineSemiPrivateRV) \label{eq:inequality3}
\end{eqnarray}
By now, an informed reader must have made the connection to capacity of the PTP channel with non-causal state \cite{1980MMPCT_GelPin}. In the sequel, we state the import of this connection.\footnote{The proof is relegated to appendix \ref{AppSec:CharacterizationForNoRateLossInPTP-STx}} This will require us to define a few mathematical objects that may initially seem unrelated to a reader unaware of findings in \cite{1980MMPCT_GelPin}. Very soon, we argue the relevance. An informed reader will find the following development natural.

Let $\mathbb{D}_{T}(\tau,\delta,\epsilon)$ denote the collection of all probability mass functions $p_{\GPV\GPS\GPX\GPY}$ defined on $\GPSetV \times \left\{ 0,1 \right\}\times \left\{ 0,1 \right\}\times \left\{ 0,1 \right\}$, where $\GPSetV$ is an arbitrary finite set such that (i) $p_{\GPY|\GPX\GPS\GPV}(x\oplus s|x,s,v)=p_{\GPY|\GPX\GPS}(x\oplus s|x,s)=1-\delta$, where $\delta \in (0,\frac{1}{2})$, (ii) $p_{\GPS}(1)=\epsilon \in [0,1]$, and (iii) $p_{\GPX}(1)\leq \tau \in (0,\frac{1}{2})$. For $p_{\GPV\GPS\GPX\GPY} \in \mathbb{D}_{T}(\tau,\delta,\epsilon)$, let \begin{equation}\alpha_{T}(p_{\GPV\GPS\GPX\GPY})=I(\GPV;\GPY)-I(\GPV;\GPS)\mbox{ and }\alpha_{T}(\tau,\delta,\epsilon)=\sup_{p_{\GPV\GPS\GPX\GPY} \in \mathbb{D}_{T}(\tau,\delta,\epsilon)}\alpha_{T}(p_{\GPV\GPS\GPX\GPY})\nonumber.\end{equation} For every $(q,w,\underline{u}) \in \mathcal{Q} \times \mathcal{W} \times \underlineSemiPrivateRV$ that satisfies $p_{\TimeSharingRV\PublicRV\underlineSemiPrivateRV}(q,w,\underline{u})>0$, we note $p_{Y_{1}|X_{1},X_{2}\oplus X_{3}\PrivateRV_{1}\TimeSharingRV\PublicRV\underlineSemiPrivateRV}(x_{1}\oplus x_{2}\oplus x_{3}|x_{1},x_{2}\oplus x_{3},v_{1},q,w,\underline{u})=p_{Y_{1}|X_{1},X_{2}\oplus X_{3}\TimeSharingRV\PublicRV\underlineSemiPrivateRV}(x_{1}\oplus x_{2}\oplus x_{3}|x_{1},x_{2}\oplus x_{3},q,w,\underline{u})=1-\delta_{1}$. In other words, conditioned on the event $\left\{ (\TimeSharingRV,\PublicRV,\underlineSemiPrivateRV)= (q,w,\underline{u})\right\}$, $V_{1}-X_{1},X_{2}\oplus X_{3}-Y_{1}$ is a Markov chain. We conclude $p_{V_{1}X_{2}\oplus X_{3}X_{1}Y_{1}|QW\underlineSemiPrivateRV}(\cdots|q,w,\underline{u}) \in \mathbb{D}_{T}(\tau_{q,w,\underline{u}},\delta_{1},\epsilon_{q,w,\underline{u}})$, where $\tau_{q,w,\underline{u}}=p_{X_{1}|QW\underlineSemiPrivateRV}(1|q,w,\underline{u})$, $\epsilon_{q,w,\underline{u}}=p_{X_{2}\oplus X_{3}|QW\underlineSemiPrivateRV}(1|q,w,\underline{u})$. Hence
\begin{equation}
 \label{Eqn:CharacterizationOfAchievableRateForUser1InTermsOfGelfandPinskerCapacity}
I(V_1;Y_1|(\TimeSharingRV ,\PublicRV, \underlineSemiPrivateRV)=(q,w,\underline{u}))
\!-\!I(V_1;X_{2}\oplus X_{3}|(\TimeSharingRV ,\PublicRV ,\underlineSemiPrivateRV)=(q,w,\underline{u}))\leq \alpha_{T}(\tau_{q,w,\underline{u}},\delta_{1},\epsilon_{q,w,\underline{u}}).
\end{equation}
We now characterize $\alpha_{T}(\tau,\delta,\epsilon)$. Verify that
$\alpha_{T}(\tau,\delta,0)=\alpha_{T}(\tau,\delta,1)= h_b(\tau *
\delta)-h_b(\delta)$. The following lemma states that
$\alpha_{T}(\tau,\delta,\epsilon)$ is strictly lower for non-trivial
values of $\epsilon$. Please refer to appendices
\ref{AppSec:CharacterizationForNoRateLossInPTP-STx} and \ref{AppSec:TheBinaryAdditiveDirtyPointToPointChannelSuffersARateLoss} for a proof.
\begin{lemma} 
\label{lem:GP}
If $\tau ,\delta \in (0,\frac{1}{2})$ and $\epsilon \in (0,1)$, then $\alpha_{T}(\tau,\delta,\epsilon) < h_b(\tau * \delta)-h_b(\delta)$. Alternatively, if $\tau ,\delta \in (0,\frac{1}{2})$ and $\epsilon \in [0,1]$, then either $\alpha_{T}(\tau,\delta,\epsilon) < h_b(\tau * \delta)-h_b(\delta)$ or $\epsilon \in \left\{ 0,1\right\}$.\end{lemma}
(\ref{eq:inequality3}), (\ref{Eqn:CharacterizationOfAchievableRateForUser1InTermsOfGelfandPinskerCapacity}) and lemma \ref{lem:GP} in conjunction with Jensen's inequality enables us to conclude
\begin{eqnarray}
 \label{Eqn:PluggingInImportOfRateLossLemma}\!\!\!\!\!\!\!\!\!
\lefteqn{\!\!\!\!\!\!\!\!\!\!\!\!\!\!\!\!\!\!\!\!\!\!R_{1} \leq I(V_1;Y_1 |\TimeSharingRV W\underlineSemiPrivateRV) -I(V_1;X_2\oplus X_3|\TimeSharingRV W\underlineSemiPrivateRV)\overset{(i)}{\leq} \!\!\sum_{(q,w,\underline{u})} p_{QW\underline{U}}(q,w,\underline{u})h_{b}(\tau_{q,w,\underline{u}} * \delta_{1})-h_{b}(\delta_{1})}\\ \lefteqn{\!\!\!\!\!\!\!\!\!\!\!\!\!\!\! \overset{(ii)}{\leq} h_{b}[\delta_{1}+(1-2\delta_{1})\!\!\!\sum_{(q,w,\underline{u})}\!\!\!\!p_{QW\underline{U}}(q,w,\underline{u})\tau_{q,w,\underline{u}}]-h_{b}(\delta_{1}) }\nonumber\\
\label{Eqn:PluggingInImportOfRateLossLemmaAndJensenIneq}
\!\!\!\!\!\!\!\!\!\!\!\!&&\!\!\!\!\!\!\!\!\!\!\!\!\!\!\!\!\!\!\!\!\!\!\!\!\!\!\!\!\!\!\!\!\!\!\!\!= h_{b}\left(\delta_{1}+(1-2\delta_{1})p_{X_{1}}(1)\right)-h_{b}(\delta_{1}) \overset{(iii)}{\leq} h_{b}\left(\delta_{1}+(1-2\delta_{1})\tau\right)-h_{b}(\delta_{1}) = h_{b}\left(\tau  *\delta_{1}\right)-h_{b}(\delta_{1}), \nonumber
\end{eqnarray}
where equality holds in (\ref{Eqn:PluggingInImportOfRateLossLemma})(i), (ii) and (iii) only if $\epsilon_{q,w,\underline{u}} \in \{0,1\}$ and $\tau_{q,w,\underline{u}}=p_{X_{1}|QW\underlineSemiPrivateRV}(1|q,w,\underline{u})=p_{X_{1}}(1)=\tau$ for every $(q,w,\underline{u})$ for which $p_{QW\underline{U}}(q,w,\underline{u})>0$. We conclude that $R_{1}=h_{b}\left(\tau  *\delta_{1}\right)-h_{b}(\delta_{1})$ only if $\tau_{q,w,\underline{u}}=\tau$ for every such $(q,w,\underline{u})$, and 
\begin{equation}
 \label{Eqn:CaseOfNoRateLoss}
I(V_1;Y_1|\TimeSharingRV W\underlineSemiPrivateRV)
-I(V_1;X_2\oplus X_3|\TimeSharingRV W\underlineSemiPrivateRV)  = h_b(\tau * \delta_{1}) -h_b(\delta_{1})\mbox{ and }H(X_2 \oplus  X_3|\TimeSharingRV W\underlineSemiPrivateRV)=0.
\end{equation}
This has got us to the third and final stage. Here we argue (\ref{Eqn:CaseOfNoRateLoss}) implies RHS of (\ref{Eqn:BoundOnUser1RateAsAConsequenceOfDecoderDecoding3Codebooks}) is strictly smaller than $h_b(\tau * \delta_{1}) -h_b(\delta_{1})$. Towards that end, note that Markov chain $X_{2}-\TimeSharingRV WU_{23}U_{12}U_{31}-X_{3}$ proved in lemma \ref{Lem:CharacterizationOfTestChannelThatContainsTheRateTriple}, item \ref{Item:X2AndX3AreConditionallyIndependentGivenQWU12U23U31}
and (\ref{Eqn:CaseOfNoRateLoss}) imply $H(X_{2}|\TimeSharingRV W\underlineSemiPrivateRV)=H(X_{3}|\TimeSharingRV W\underlineSemiPrivateRV)=0$.\footnote{Indeed, for any $(q,w,\underline{u}) \in \TimeSharingRVSet \times \PublicRVSet \times \underline{\SemiPrivateRVSet}$ that satisfies $P((\TimeSharingRV,\PublicRV,\underlineSemiPrivateRV)=(q,w,\underline{u}))>0$, if $P(X_{j}=1|(\TimeSharingRV,\PublicRV,\underlineSemiPrivateRV)=(q,w,\underline{u}))=\alpha_{j}:j=2,3$, then $0=H(X_{2}\oplus X_{3}|(\TimeSharingRV,\PublicRV,\underlineSemiPrivateRV)=(q,w,\underline{u}))=h_{b}(\alpha_{2}*\alpha_{3}) \geq \alpha_{2}h_{b}(1-\alpha_{3})+(1-\alpha_{2})h_{b}(\alpha_{3})=\alpha_{2}h_{b}(\alpha_{3})+(1-\alpha_{2})h_{b}(\alpha_{3})=h_{b}(\alpha_{3})\geq 0$, where the first inequality follows from concavity of binary entropy function, and similarly, interchanging the roles of $\alpha_{2}, \alpha_{3}$, we obtain $0=H(X_{2}\oplus X_{3}|(\TimeSharingRV,\PublicRV,\underlineSemiPrivateRV)=(q,w,\underline{u})) \geq h_{b}(\alpha_{2})\geq 0$.} Furthermore, Markov chains $\SemiPrivateRV_{12}-\TimeSharingRV\PublicRV \SemiPrivateRV_{23} \SemiPrivateRV_{31}-X_{3}$ and $\SemiPrivateRV_{31}-\TimeSharingRV\PublicRV \SemiPrivateRV_{23} \SemiPrivateRV_{12}-X_{2}$ proved in lemma \ref{Lem:CharacterizationOfTestChannelThatContainsTheRateTriple} item \ref{Item:U31AndX2AreConditionallyIndependentGivenQWU12U23} imply
\begin{equation}
\label{Eqn:InTheCaseOfNoRateLossReeciever1KnowsX2AndX3}
H(X_2|\TimeSharingRV WU_{23}U_{12})=H(X_3|\TimeSharingRV WU_{23}U_{31})=0.
\end{equation}
Observe that
\begin{eqnarray}
\label{Eqn:SubstitutingForNoRateLossCaseInUpperBound1}
h_b(\tau * \delta_{1}) -h_b(\delta_{1})\!\!&=&\!\!I(V_1;Y_1|\TimeSharingRV W\underlineSemiPrivateRV) -I(V_1;X_2\oplus X_3|\TimeSharingRV  W\underlineSemiPrivateRV) = I(V_1;Y_1|\TimeSharingRV  W \underlineSemiPrivateRV ) =I(V_1;Y_1|\TimeSharingRV  W \underlineSemiPrivateRV ,X_2,X_3) \\
\!\!&=&\!\! H(Y_{1}|\TimeSharingRV W\underlineSemiPrivateRV X_{2}X_{3})-H(Y_{1}|\TimeSharingRV W\underlineSemiPrivateRV V_{1}X_{2}X_{3})
\leq H(Y_{1}|\TimeSharingRV W\underlineSemiPrivateRV X_{2}X_{3})-H(Y_{1}|\TimeSharingRV W\underlineSemiPrivateRV V_{1}X_{1}X_{2}X_{3})\nonumber\\
\label{Eqn:EntropyOfY1ConditionedOnX2AndX3}
\!\!&=&\!\!H(Y_{1}|\TimeSharingRV  W \underlineSemiPrivateRV ,X_2,X_3)-h_b(\delta_{1})
\end{eqnarray}
where the first two equalities in (\ref{Eqn:SubstitutingForNoRateLossCaseInUpperBound1}) follows from (\ref{Eqn:CaseOfNoRateLoss}) and the last equality follows from (\ref{Eqn:InTheCaseOfNoRateLossReeciever1KnowsX2AndX3}). (\ref{Eqn:EntropyOfY1ConditionedOnX2AndX3}) and first equality in (\ref{Eqn:SubstitutingForNoRateLossCaseInUpperBound1}) enables us to conclude
\begin{equation}
 \label{Eqn:LowerBoundOnEntropyOfY1ConditionedOnX2AndX3}
H(Y_{1}|\TimeSharingRV  W \underlineSemiPrivateRV ,X_2,X_3) \geq h_b(\tau * \delta_{1})
\end{equation}
We now provide an upper bound on the right hand side of (\ref{Eqn:BoundOnUser1RateAsAConsequenceOfDecoderDecoding3Codebooks}). Note that it suffices to prove $I(U_{12}U_{31};Y_1|\TimeSharingRV WU_{23})-I(U_{12};Y_2|\TimeSharingRV WU_{23}) -I(U_{31};Y_3|\TimeSharingRV WU_{23})$ is negative. Observe that
\begin{eqnarray}
\lefteqn{I(U_{12}U_{31};Y_1|\TimeSharingRV WU_{23})-I(U_{12};Y_2|\TimeSharingRV WU_{23})
-I(U_{31};Y_3|\TimeSharingRV WU_{23})}\nonumber\\ 
&&=H(Y_{1}|\TimeSharingRV WU_{23})-H(Y_{1}|\TimeSharingRV W\underlineSemiPrivateRV)-H(Y_{2}|\TimeSharingRV WU_{23})+H(Y_{2}|\TimeSharingRV WU_{23}U_{12})-H(Y_{3}|\TimeSharingRV WU_{23})+H(Y_{3}|\TimeSharingRV WU_{23}U_{31})
\nonumber\\
\label{Eqn:X2IsAFunctionOfWU23U12AndX3IsAFunctionOfWU23U31}
&&=H(Y_{1}|\TimeSharingRV WU_{23})-H(Y_{1}|\TimeSharingRV WX_{2}X_{3}\underlineSemiPrivateRV)-H(Y_{2})+H(Y_{2}|\TimeSharingRV WU_{23}U_{12}X_{2})-H(Y_{3})+H(Y_{3}|\TimeSharingRV WU_{23}U_{31}X_{3})
\\
&&= H(Y_{1}|\TimeSharingRV WU_{23})-H(Y_{1}|\TimeSharingRV WX_{2}X_{3}\underlineSemiPrivateRV)-2+h_{b}(\delta_{2})+h_{b}(\delta_{3})
\nonumber\\
\label{Eqn:SubstitutingLowerBoundOnEntropyOfY1}
&&\leq 1-H(Y_{1}|\TimeSharingRV WX_{2}X_{3}\underlineSemiPrivateRV)-2+h_{b}(\delta_{2})+h_{b}(\delta_{3})
\leq h_{b}(\delta_{2})+h_{b}(\delta_{3})-h_b(\delta_{1}*\tau)-1
\end{eqnarray}
where (\ref{Eqn:X2IsAFunctionOfWU23U12AndX3IsAFunctionOfWU23U31}) follows from (\ref{Eqn:CaseOfNoRateLoss}) and (\ref{Eqn:InTheCaseOfNoRateLossReeciever1KnowsX2AndX3}), second inequality in (\ref{Eqn:SubstitutingLowerBoundOnEntropyOfY1}) follows from (\ref{Eqn:LowerBoundOnEntropyOfY1ConditionedOnX2AndX3}).
If $\tau,\delta_{1},\delta_{2},\delta_{3}$ are such that $h_b(\delta_{2})+h_{b}(\delta_{3}) < 1+h_b(\delta_{1}*\tau)
$, then right hand side of
(\ref{Eqn:SubstitutingLowerBoundOnEntropyOfY1}) is negative. This
concludes the proof. 

\appendices
\section{Upper bound on $P(\epsilon_{l})$}
\label{AppSec:AnalysisOfEncoderErrorEventFor3To1DBC}
From (\ref{Eqn:3To1DBCOneLayerEncoderErrorEventUseOfCheybyshevIneq}), it suffices to derive upper and lower bounds on $\Var\left\{ \phi(M_{1},M_{2}^{t_{2}},M_{3}^{t_{3}})\right\}$ and $\Expectation\left\{ \phi(M_{1},M_{2}^{t_{2}},M_{3}^{t_{3}}) \right\}$ respectively. Note that $\Expectation\left\{\phi^{2}(m_{1},m_{2}^{t_{2}},m_{3}^{t_{3}})\right\} = \sum_{l=0}^{7}\mathscr{T}_{l}$, where
\begin{eqnarray}
 \label{Eqn:3To1DBCOneLayerEncoderErrorEventVarianceOfPhi}
\mathscr{T}_{0}&=&\Expectation\left\{ \phi(M_{1},M_{2}^{t_{2}},M_{3}^{t_{3}})  \right\} = \sum_{\substack{(b_{1},a_{2}^{s_{2}},a_{3}^{s_{3}}) \in \\ \mathcal{B}_{1} \times \fieldpi^{s_{2}} \times \fieldpi^{s_{3}}}}\sum_{\substack{(v_{1}^{n},u_{2}^{n},u_{3}^{n}) \in\\T_{2\eta_{2}}(V_{1},U_{2},U_{3}|q^{n})}} P\left( \substack{V_{1}^{n}(m_{1},b_{1})=v_{1}^{n},U_{j}(a_{j}^{s_{j}})=u_{j}^{n},I(a_{j}^{s_{j}})=m_{j}^{t_{j}}:j=2,3  }\right),\\
 \mathscr{T}_{1} &=& \sum_{\substack{(b_{1},a_{2}^{s_{2}},a_{3}^{s_{3}}) \in \\ \mathcal{B}_{1} \times \fieldpi^{s_{2}} \times \fieldpi^{s_{3}}}}\sum_{\substack{\tilde{a}_{3}^{s_{3}}\in \fieldpi^{s_{3}}\\\tilde{a}_{3}^{s_{3}}\neq a_{3}^{s_{3}}}}\sum_{\substack{(v_{1}^{n},u_{2}^{n},u_{3}^{n}) \in\\T_{2\eta_{2}}(V_{1},\underline{U}|q^{n})}}\sum_{\substack{\tilde{u}_{3}^{n} \in\\T_{2\eta_{2}}(U_{3}|q^{n},v_{1}^{n},u_{2}^{n})}} P\left(\substack{ V_{1}^{n}(m_{1},b_{1})=v_{1}^{n},U_{j}(a_{j}^{s_{j}})=u_{j}^{n},I(a_{j}^{s_{j}})=m_{j}^{t_{j}}:j=2,3,\\U_{3}^{n}(\tilde{a}_{3}^{s_{3}})=\tilde{u}_{3}^{n},I(\tilde{a}_{3}^{s_{3}})=m_{3}^{t_{3}}  }\right) \nonumber
 \end{eqnarray}
\begin{eqnarray}
\mathscr{T}_{2} &=& \sum_{\substack{(b_{1},a_{2}^{s_{2}},a_{3}^{s_{3}}) \in \\ \mathcal{B}_{1} \times \fieldpi^{s_{2}} \times \fieldpi^{s_{3}}}}\sum_{\substack{\tilde{a}_{2}^{s_{2}}\in \fieldpi^{s_{2}}\\\tilde{a}_{2}^{s_{2}}\neq a_{2}^{s_{2}}}}\sum_{\substack{(v_{1}^{n},u_{2}^{n},u_{3}^{n}) \in\\T_{2\eta_{2}}(V_{1},\underline{U}|q^{n})}}\sum_{\substack{\tilde{u}_{2}^{n} \in\\T_{2\eta_{2}}(U_{2}|q^{n},v_{1}^{n},u_{3}^{n})}} P\left(\substack{ V_{1}^{n}(m_{1},b_{1})=v_{1}^{n},U_{j}(a_{j}^{s_{j}})=u_{j}^{n},I(a_{j}^{s_{j}})=m_{j}^{t_{j}}:j=2,3,\\U_{2}^{n}(\tilde{a}_{2}^{s_{2}})=\tilde{u}_{2}^{n},I(\tilde{a}_{2}^{s_{2}})=m_{2}^{t_{2}}  }\right) \nonumber\\
 \mathscr{T}_{3} &=&
 \sum_{\substack{(b_{1},a_{2}^{s_{2}},a_{3}^{s_{3}}) \in
     \\ \mathcal{B}_{1} \times \fieldpi^{s_{2}} \times
     \fieldpi^{s_{3}}}}\sum_{\substack{(\tilde{a}_{2}^{s_{2}},\tilde{a}_{3}^{s_{3}})\in
     \fieldpi^{s_{2}}\times
     \fieldpi^{s_{3}}\\\tilde{a}_{2}^{s_{2}}\neq
     a_{2}^{s_{2}},\tilde{a}_{3}^{s_{3}}\neq
     a_{3}^{s_{3}}}}\sum_{\substack{(v_{1}^{n},u_{2}^{n},u_{3}^{n})
     \in\\T_{2\eta_{2}}(V_{1},\underline{U}|q^{n})}}\sum_{\substack{(\tilde{u}_{2}^{n},\tilde{u}_{3}^{n})
     \in\\T_{2\eta_{2}}(\underline{U}|q^{n},v_{1}^{n})}}
 P\left(\substack{
   V_{1}^{n}(m_{1},b_{1})=v_{1}^{n},U_{j}(a_{j}^{s_{j}})=u_{j}^{n},I(a_{j}^{s_{j}})=m_{j}^{t_{j}}:j=2,3,\\U_{j}(\tilde{a}_{j}^{s_{j}})=\tilde{u}_{j}^{n},I(\tilde{a}_{j}^{s_{j}})=m_{j}^{t_{j}}:j=2,3
 }\right)\nonumber
\end{eqnarray}
\begin{eqnarray}
 \mathscr{T}_{4} &=& \sum_{\substack{(b_{1},a_{2}^{s_{2}},a_{3}^{s_{3}}) \in \\ \mathcal{B}_{1} \times \fieldpi^{s_{2}} \times \fieldpi^{s_{3}}}}\sum_{\substack{\tilde{b}_{1} \in \mathcal{B}_{1}\\\tilde{b}_{1}\neq b_{1}}}\sum_{\substack{(v_{1}^{n},u_{2}^{n},u_{3}^{n}) \in\\T_{2\eta_{2}}(V_{1},\underline{U}|q^{n})}}\sum_{\substack{\tilde{v}_{1}^{n} \in\\T_{2\eta_{2}}(V_{1}|q^{n},\underline{u}^{n})}} P\left( \substack{V_{1}^{n}(m_{1},b_{1})=v_{1}^{n},U_{j}(a_{j}^{s_{j}})=u_{j}^{n},I(a_{j}^{s_{j}})=m_{j}^{t_{j}}:j=2,3,V_{1}^{n}(m_{1},\tilde{b}_{1})=\tilde{v}_{1}^{n}  }\right) \nonumber\\
 \mathscr{T}_{5} &=& \sum_{\substack{(b_{1},a_{2}^{s_{2}},a_{3}^{s_{3}}) \in \\ \mathcal{B}_{1} \times \fieldpi^{s_{2}} \times \fieldpi^{s_{3}}}}\sum_{\substack{(\tilde{b}_{1},\tilde{a}_{3}^{s_{3}})\in \mathcal{B}_{1}\times\fieldpi^{s_{3}}\\\tilde{b}_{1}\neq b_{1}\tilde{a}_{3}^{s_{3}}\neq a_{3}^{s_{3}}}}\sum_{\substack{(v_{1}^{n},u_{2}^{n},u_{3}^{n}) \in\\T_{2\eta_{2}}(V_{1},\underline{U}|q^{n})}}\sum_{\substack{\tilde{v}_{1}^{n},\tilde{u}_{3}^{n} \in\\T_{2\eta_{2}}(V_{1},U_{3}|q^{n},u_{2}^{n})}} P\left(\substack{ V_{1}^{n}(m_{1},b_{1})=v_{1}^{n},U_{j}(a_{j}^{s_{j}})=u_{j}^{n},I(a_{j}^{s_{j}})=m_{j}^{t_{j}}:j=2,3,\\V_{1}(m_{1},\tilde{b}_{1})=\tilde{v}_{1}^{n},U_{3}^{n}(\tilde{a}_{3}^{s_{3}})=\tilde{u}_{3}^{n},I(\tilde{a}_{3}^{s_{3}})=m_{3}^{t_{3}} } \right)  \nonumber\\
 \mathscr{T}_{6} &=& \sum_{\substack{(b_{1},a_{2}^{s_{2}},a_{3}^{s_{3}}) \in \\ \mathcal{B}_{1} \times \fieldpi^{s_{2}} \times \fieldpi^{s_{3}}}}\sum_{\substack{(\tilde{b}_{1},\tilde{a}_{2}^{s_{2}})\in \mathcal{B}_{1}\times\fieldpi^{s_{2}}\\\tilde{b}_{1}\neq b_{1}\tilde{a}_{2}^{s_{2}}\neq a_{2}^{s_{2}}}}\sum_{\substack{(v_{1}^{n},u_{2}^{n},u_{3}^{n}) \in\\T_{2\eta_{2}}(V_{1},\underline{U}|q^{n})}}\sum_{\substack{\tilde{v}_{1}^{n},\tilde{u}_{2}^{n} \in\\T_{2\eta_{2}}(V_{1},U_{2}|q^{n},u_{3}^{n})}} P\left(\substack{ V_{1}^{n}(m_{1},b_{1})=v_{1}^{n},U_{j}(a_{j}^{s_{j}})=u_{j}^{n},I(a_{j}^{s_{j}})=m_{j}^{t_{j}}:j=2,3,\\V_{1}(m_{1},\tilde{b}_{1})=\tilde{v}_{1}^{n},U_{2}^{n}(\tilde{a}_{2}^{s_{2}})=\tilde{u}_{2}^{n},I(\tilde{a}_{2}^{s_{2}})=m_{2}^{t_{2}} } \right)\nonumber\\  
\mathscr{T}_{7} &=& \sum_{\substack{(b_{1},a_{2}^{s_{2}},a_{3}^{s_{3}}) \in \\ \mathcal{B}_{1} \times \fieldpi^{s_{2}} \times \fieldpi^{s_{3}}}}\sum_{\substack{(\tilde{b}_{1},\tilde{a}_{2}^{s_{2}},\tilde{a}_{3}^{s_{3}})\in \mathcal{B}_{1}\times\fieldpi^{s_{2}}\times \fieldpi^{s_{3}}\\\tilde{b}_{1}\neq b_{1},\tilde{a}_{2}^{s_{2}}\neq a_{2}^{s_{2}},\tilde{a}_{3}^{s_{3}}\neq a_{3}^{s_{3}}}}\sum_{\substack{(v_{1}^{n},u_{2}^{n},u_{3}^{n}) \in\\T_{2\eta_{2}}(V_{1},\underline{U}|q^{n})}}\sum_{\substack{(\tilde{v}_{1}^{n},\tilde{u}_{2}^{n},\tilde{u}_{3}^{n}) \in\\T_{2\eta_{2}}(V_{1},\underline{U}|q^{n})}} P\left(\substack{ V_{1}^{n}(m_{1},b_{1})=v_{1}^{n},U_{j}(a_{j}^{s_{j}})=u_{j}^{n},I(a_{j}^{s_{j}})=m_{j}^{t_{j}}:j=2,3,\\V_{1}(m_{1},\tilde{b}_{1})=\tilde{v}_{1}^{n},U_{j}^{n}(\tilde{a}_{j}^{s_{j}})=\tilde{u}_{j}^{n},I(\tilde{a}_{j}^{s_{j}})=m_{j}^{t_{j}}:j=2,3 } \right) .\nonumber
\end{eqnarray}
We have
\begin{equation}
 \label{Eqn:3To1DBCOneLayerEncoderErrorEventInTermsOfExpectationsOfT1ToT7}
 \frac{4\Var\left\{ \phi(M_{1},M_{2}^{t_{2}},M_{3}^{t_{3}}) \right\}}{\left(\Expectation\left\{ \phi(M_{1},M_{2}^{t_{2}},M_{3}^{t_{3}}) \right\}\right)^{2}} = 4\frac{\left(\sum_{l=0}^{7} \mathscr{T}_{l}\right) - \mathscr{T}_{0}^{2}}{\mathscr{T}_{0}^{2}}.\nonumber
\end{equation}
We take a closer look at $\mathscr{T}_{7}$. For $\theta \in \fieldpi$, let \[\mathscr{D}_{\theta}(a_{2}^{s_{2}},a_{3}^{s_{3}}) \define \left\{ (\tilde{a}_{2}^{s_{2}},\tilde{a}_{3}^{s_{3}}): \tilde{a}_{3l}^{s_{3}}-{a}_{3l}^{s_{3}} = \theta (\tilde{a}_{2l}^{s_{2}}-a_{2l}^{s_{2}}) \mbox{ for }1\leq l \leq s_{2}\mbox{ and }\tilde{a}_{3l}^{s_{3}}-{a}_{3l}^{s_{3}} =0 \mbox{ for } s_{2}+1 \leq l \leq s_{3} \right\},\]
$\mathscr{D}(a_{2}^{s_{2}},a_{3}^{s_{3}}) \define \underset{\theta \in \fieldpi}{\cup}\mathscr{D}_{\theta}(a_{2}^{s_{2}},a_{3}^{s_{3}})$ and $\mathscr{I}(a_{2}^{s_{2}},a_{3}^{s_{3}}) = \fieldpi^{s_{2}}\times \fieldpi^{s_{3}} \setminus \mathscr{D}(a_{2}^{s_{2}},a_{3}^{s_{3}})$. The reader may verify that for $(\tilde{a}_{2}^{s_{2}},\tilde{a}_{3}^{s_{3}}) \in \mathscr{D}_{\theta}(a_{2}^{s_{2}},a_{3}^{s_{3}})$
\begin{eqnarray}
 \label{Eqn:3To1DBCOneLayerEncoderErrorEventTheNonZeroTermsInT7}
 P\left( \substack{ V_{1}^{n}(m_{1},b_{1})=v_{1}^{n},U_{j}(a_{j}^{s_{j}})=u_{j}^{n},I(a_{j}^{s_{j}})=m_{j}^{t_{j}}:j=2,3,\\V_{1}(m_{1},\tilde{b}_{1})=\tilde{v}_{1}^{n},U_{j}^{n}(\tilde{a}_{j}^{s_{j}})=\tilde{u}_{j}^{n},I(\tilde{a}_{j}^{s_{j}})=m_{j}^{t_{j}}:j=2,3  } \right) = \left\{ \begin{array}{lr} \frac{P(V_{1}^{n}(m_{1},b_{1})=v_{1}^{n},V_{1}(m_{1},\tilde{b}_{1})=\tilde{v}_{1}^{n})}{\pi^{3n+2t_{2}+2t_{3}}}&\mbox{ if }\tilde{u}_{3}^{n}\ominus\theta \tilde{u}_{2}^{n} = u_{3}^{n}\ominus\theta u_{2}^{n}\\0&\mbox{otherwise} \end{array}\right.\nonumber
\end{eqnarray}
For $(\tilde{a}_{2}^{s_{2}},\tilde{a}_{3}^{s_{3}}) \in \mathscr{I}(a_{2}^{s_{2}},a_{3}^{s_{3}})$, we claim
\begin{eqnarray}
 \label{Eqn:3To1DBCOneLayerEncoderErrorEventTheNonZeroTermsInT7WithIndependentPicks}
 P\left( \substack{ V_{1}^{n}(m_{1},b_{1})=v_{1}^{n},U_{j}(a_{j}^{s_{j}})=u_{j}^{n},I(a_{j}^{s_{j}})=m_{j}^{t_{j}}:j=2,3,\\V_{1}(m_{1},\tilde{b}_{1})=\tilde{v}_{1}^{n},U_{j}^{n}(\tilde{a}_{j}^{s_{j}})=\tilde{u}_{j}^{n},I(\tilde{a}_{j}^{s_{j}})=m_{j}^{t_{j}}:j=2,3  } \right) = P\left( \substack{V_{1}^{n}(m_{1},b_{1})=v_{1}^{n},U_{j}(a_{j}^{s_{j}})=u_{j}^{n}\\I(a_{j}^{s_{j}})=m_{j}^{t_{j}}:j=2,3}  \right) P\left( \substack{V_{1}(m_{1},\tilde{b}_{1})=\tilde{v}_{1}^{n},U_{j}^{n}(\tilde{a}_{j}^{s_{j}})=\tilde{u}_{j}^{n}\\I(\tilde{a}_{j}^{s_{j}})=m_{j}^{t_{j}}:j=2,3} \right). \nonumber
\end{eqnarray}
In order to prove this claim, it suffices to prove
\begin{eqnarray}
 \label{Eqn:3To1DBCOneLayerEncoderErrorEventTheNonZeroTermsInT7WithIndependentPicks}
 P\left( \substack{ V_{1}^{n}(m_{1},b_{1})=v_{1}^{n},U_{j}(a_{j}^{s_{j}})=u_{j}^{n},I(a_{j}^{s_{j}})=m_{j}^{t_{j}}:j=2,3,\\V_{1}(m_{1},\tilde{b}_{1})=\tilde{v}_{1}^{n},U_{j}^{n}(\tilde{a}_{j}^{s_{j}})=\tilde{u}_{j}^{n},I(\tilde{a}_{j}^{s_{j}})=m_{j}^{t_{j}}:j=2,3  } \right) = \frac{P(V_{1}^{n}(m_{1},b_{1})=v_{1}^{n},V_{1}(m_{1},\tilde{b}_{1})=\tilde{v}_{1}^{n})}{\pi^{4n+2t_{2}+2t_{3}}}. \nonumber
\end{eqnarray}
which can be verified through a counting process. We therefore have $\mathscr{T}_{7}=\mathscr{T}_{7I}+\mathscr{T}_{7D}$, where
\begin{eqnarray}
 \label{Eqn:3To1DBCEncoderErrorEventTermsT7I}
 \mathscr{T}_{7I} =\!\!\!\! \sum_{\substack{(b_{1},a_{2}^{s_{2}},a_{3}^{s_{3}}) \in \\ \mathcal{B}_{1} \times \fieldpi^{s_{2}} \times \fieldpi^{s_{3}}}}\sum_{\substack{(\tilde{b}_{1},\tilde{a}_{2}^{s_{2}},\tilde{a}_{3}^{s_{3}})\in \\\mathcal{B}_{1} \times\mathscr{I}(a_{2}^{s_{2}},a_{3}^{s_{3}})}}\sum_{\substack{(v_{1}^{n},u_{2}^{n},u_{3}^{n}) \in\\T_{2\eta_{2}}(V_{1},\underline{U}|q^{n})}}\sum_{\substack{(\tilde{v}_{1}^{n},\tilde{u}_{2}^{n},\tilde{u}_{3}^{n}) \in\\T_{2\eta_{2}}(V_{1},\underline{U}|q^{n})}} \!\!\!\!\!P\left( \substack{V_{1}^{n}(m_{1},b_{1})=v_{1}^{n},U_{j}(a_{j}^{s_{j}})=u_{j}^{n}\\I(a_{j}^{s_{j}})=m_{j}^{t_{j}}:j=2,3}  \right) P\left( \substack{V_{1}(m_{1},\tilde{b}_{1})=\tilde{v}_{1}^{n},U_{j}^{n}(\tilde{a}_{j}^{s_{j}})=\tilde{u}_{j}^{n}\\I(\tilde{a}_{j}^{s_{j}})=m_{j}^{t_{j}}:j=2,3} \right) \\
 \mathscr{T}_{7D} =\!\!\!\!\!\!\! \sum_{\substack{(b_{1},a_{2}^{s_{2}},a_{3}^{s_{3}}) \in \\ \mathcal{B}_{1} \times \fieldpi^{s_{2}} \times \fieldpi^{s_{3}}}}\sum_{\substack{(\tilde{b}_{1},\tilde{a}_{2}^{s_{2}},\tilde{a}_{3}^{s_{3}})\in \\\mathcal{B}_{1}\times \mathscr{D}(a_{2}^{s_{2}},a_{3}^{s_{3}})}}\sum_{\substack{u^{n} \in\\T_{2\eta_{2}}(U_{3}\ominus\theta U_{2}|q^{n})}}\sum_{\substack{({v}_{1}^{n},{u}_{2}^{n},u^{n}\oplus\theta{u}_{2}^{n}) \in\\T_{2\eta_{2}}(V_{1},\underline{U}|q^{n})}}\sum_{\substack{(\tilde{v}_{1}^{n},\tilde{u}_{2}^{n},u^{n}\oplus\theta\tilde{u}_{2}^{n}) \in\\T_{2\eta_{2}}(V_{1},\underline{U}|q^{n})}} \!\!\!\!\!\frac{P(V_{1}^{n}(m_{1},b_{1})=v_{1}^{n},V_{1}(m_{1},\tilde{b}_{1})=\tilde{v}_{1}^{n})}{\pi^{3n+2t_{2}+2t_{3}}}. \nonumber
\end{eqnarray}
Verify that $\mathscr{T}_{7I} \leq \mathscr{T}_{0}^{2}$. We therefore have
\begin{equation}
 \label{Eqn:3To1DBCOneLayerEncoderErrorEventInTermsOfT1ToT7D}
 \frac{4\Var\left\{ \phi(M_{1},M_{2}^{t_{2}},M_{3}^{t_{3}}) \right\}}{\Expectation\left\{ \phi(M_{1},M_{2}^{t_{2}},M_{3}^{t_{3}}) \right\}} \leq 4\frac{\left(\sum_{l=0}^{6} \mathscr{T}_{l}\right) + \mathscr{T}_{7D}}{\mathscr{T}_{0}^{2}}.
\end{equation}
and it suffices to derive lower bound on $\mathscr{T}_{0}$ and upper bounds on $\mathscr{T}_{l}:l \in [6]$ and $\mathscr{T}_{7D}$. 

Just as we split $\mathscr{T}_{7}$, we split $\mathscr{T}_{3}$ as $\mathscr{T}_{3} = \mathscr{T}_{3I} + \mathscr{T}_{3D} $. We let the reader fill in the details and confirm the following bounds. From conditional typicality results, there exists $N_{2}(\eta_{2}) \in \naturals$, such that for all $n \geq N_{2}(\eta_{2})$,
\begin{eqnarray}
 \label{Eqn:3To1DBCEncoderErrorEventBoundsOnT0ToT7}
 \mathscr{T}_{0} & \geq & \frac{|\mathcal{B}_{1}|\pi^{s_{2}+s_{3}}\exp\left\{ nH(V_{1},\underline{U}|Q)-4n\eta_{2} \right\}}{\pi^{2n+t_{2}+t_{3}}\exp\left\{ nH(V_{1}|Q)+4n\eta_{2} \right\}}\nonumber\\
\mathscr{T}_{1} & \leq & \frac{|\mathcal{B}_{1}|\pi^{s_{2}+2s_{3}}\exp\left\{ nH(V_{1},\underline{U}|Q)+4n\eta_{2}+nH(U_{3}|Q,V_{1},U_{2})+8n\eta_{2} \right\}}{\pi^{3n+t_{2}+2t_{3}}\exp\left\{ nH(V_{1}|Q)-4n\eta_{2} \right\}}\nonumber\\
 \mathscr{T}_{2} & \leq & \frac{|\mathcal{B}_{1}|\pi^{2s_{2}+s_{3}}\exp\left\{ nH(V_{1},\underline{U}|Q)+4n\eta_{2}+nH(U_{2}|Q,V_{1},U_{3})+8n\eta_{2} \right\}}{\pi^{3n+2t_{2}+t_{3}}\exp\left\{ nH(V_{1}|Q)-4n\eta_{2} \right\}}\nonumber\\
 \mathscr{T}_{3I} & \leq & \frac{|\mathcal{B}_{1}|\pi^{2s_{2}+2s_{3}}\exp\left\{ nH(V_{1},\underline{U}|Q)+4n\eta_{2}+nH(U_{2},U_{3}|Q,V_{1})+8n\eta_{2} \right\}}{\pi^{4n+2t_{2}+2t_{3}}\exp\left\{ nH(V_{1}|Q)-4n\eta_{2} \right\}}\nonumber\\
 \mathscr{T}_{3D} & \leq & \pi\frac{|\mathcal{B}_{1}|\pi^{2s_{2}+s_{3}}\exp\left\{ nH(V_{1},\underline{U}|Q,U_{3}\ominus \theta U_{2})+8n\eta_{2}+nH(U_{3}\ominus \theta U_{2}|Q)+4n\eta_{2} \right\}}{\pi^{3n+2t_{2}+2t_{3}}\exp\left\{ nH(V_{1}|Q)-4n\eta_{2}-nH(\underline{U}|Q,V_{1},U_{3}\ominus \theta U_{2})-16n\eta_{2} \right\}}\nonumber\\
 \mathscr{T}_{4} & \leq &\frac{|\mathcal{B}_{1}|^{2}\pi^{s_{2}+s_{3}}\exp\left\{ nH(V_{1},\underline{U}|Q)+4n\eta_{2}+nH(V_{1}|Q,U_{2},U_{3})+8n\eta_{2} \right\}}{\pi^{2n+t_{2}+t_{3}}\exp\left\{ 2nH(V_{1}|Q)-8n\eta_{2} \right\}}\nonumber\\
 \mathscr{T}_{5} & \leq &\frac{|\mathcal{B}_{1}|^{2}\pi^{s_{2}+2s_{3}}\exp\left\{ nH(V_{1},\underline{U}|Q)+4n\eta_{2}+nH(V_{1},U_{3}|Q,U_{2})+8n\eta_{2} \right\}}{\pi^{3n+t_{2}+2t_{3}}\exp\left\{ 2nH(V_{1}|Q)-8n\eta_{2} \right\}}\nonumber\\
 \mathscr{T}_{6} & \leq &\frac{|\mathcal{B}_{1}|^{2}\pi^{2s_{2}+s_{3}}\exp\left\{ nH(V_{1},\underline{U}|Q)+4n\eta_{2}+nH(V_{1},U_{2}|Q,U_{3})+8n\eta_{2} \right\}}{\pi^{3n+2t_{2}+t_{3}}\exp\left\{ 2nH(V_{1}|Q)-8n\eta_{2} \right\}}\nonumber\\
  \mathscr{T}_{7D} & \leq & \frac{|\mathcal{B}_{1}|^{2}\pi^{2s_{2}+s_{3}}\exp\left\{ 2nH(V_{1},\underline{U}|Q,U_{3}\ominus \theta U_{2})+16n\eta_{2}+nH(U_{3}\ominus \theta U_{2}|Q)+4n\eta_{2} \right\}}{\pi^{3n+2t_{2}+2t_{3}}\exp\left\{ 2nH(V_{1}|Q)-8n\eta_{2} \right\}}\nonumber
\end{eqnarray}
We now employ the bounds on the parameters of the code ((\ref{Eqn:3To1DBCOneLayerAchievabilityLowerAndUpperBoundsOnSj}) - (\ref{Eqn:3To1DBCOneLayerAchievabilityLowerAndUpperBoundsOnR1AndK1})). It maybe verified that, for $n \geq \max\{ N_{1}(\eta), N_{2}(\eta_{2})\}$,
\begin{eqnarray}
 \label{Eqn:3DBCProofTwiceTheListSize}
 \frac{\mathscr{T}_{0}}{\mathscr{T}_{0}^{2}} &\leq &\exp \left\{\! -n \left( \frac{\log|\mathcal{B}_{1}|}{n}+\left(\sum_{l=2}^{3}\frac{s_{l}-t_{l}}{n}\right)\log\pi  - \left[ 2\log\pi-H(\underline{U}|Q,V_{1})+16\eta_{2}\right]\right)  \right\} \leq \exp \left\{\! -n\left( \substack{\delta_{1}+\frac{\eta}{8}\\-16\eta_{2}}\right) \right\}\\
 \frac{\mathscr{T}_{1}}{\mathscr{T}_{0}^{2}} & \leq & \exp \left\{ -n \left( \frac{\log|\mathcal{B}_{1}|}{n}+\frac{s_{2}-t_{2}}{n}\log\pi - \left[ \log\pi-H(U_{2}|Q,V_{1})+32\eta_{2}\right]\right)  \right\}\leq \exp \left\{ -n\left( \delta_{1}+\frac{\eta}{8}-32\eta_{2}\right) \right\} \nonumber\\
 \frac{\mathscr{T}_{2}}{\mathscr{T}_{0}^{2}} & \leq & \exp \left\{ -n \left( \frac{\log|\mathcal{B}_{1}|}{n}+\frac{s_{3}-t_{3}}{n}\log\pi - \left[ \log\pi-H(U_{3}|Q,V_{1})+32\eta_{2}\right]\right)  \right\}\leq \exp \left\{ -n\left( \delta_{1}+\frac{\eta}{8}-32\eta_{2}\right) \right\}\nonumber\\
 \frac{\mathscr{T}_{3I}}{\mathscr{T}_{0}^{2}} & \leq & \exp \left\{ -n \left( \frac{\log|\mathcal{B}_{1}|}{n}- 32\eta_{2}\right)  \right\}\leq \exp \left\{ -n\left( \delta_{1}+\frac{\eta}{8}-32\eta_{2}\right) \right\}\nonumber
\end{eqnarray}
\begin{eqnarray}
 \frac{\mathscr{T}_{3D}}{\mathscr{T}_{0}^{2}} & \leq & \max_{\theta \neq 0}\exp \left\{ -n \left( \frac{\log|\mathcal{B}_{1}|}{n}+\frac{s_{3}}{n}\log\pi - \left[ \log\pi-H(U_{3}\ominus \theta U_{2}|Q,V_{1})+48\eta_{2}\right]\right)  \right\}\leq \pi\exp \left\{ -n\left( \delta_{1}-48\eta_{2}\right) \right\}\nonumber\\
 \frac{\mathscr{T}_{4}}{\mathscr{T}_{0}^{2}} & \leq & \exp \left\{ -n \left( \left(\sum_{l=2}^{3}\frac{s_{l}-t_{l}}{n}\right)\log\pi - \left[ \log\pi-H(\underline{U}|Q)+36\eta_{2}\right]\right)  \right\}\leq \exp \left\{ -n\left( \delta_{1}-36\eta_{2}\right) \right\}\nonumber
\end{eqnarray}
\begin{eqnarray}
 \frac{\mathscr{T}_{5}}{\mathscr{T}_{0}^{2}} & \leq & \exp \left\{ -n \left( \frac{s_{2}-t_{2}}{n}\log\pi - \left[ \log\pi-H(U_{2}|Q)+36\eta_{2}\right]\right)  \right\}\leq \exp \left\{ -n\left( \delta_{1}-36\eta_{2}\right) \right\}\nonumber\\
 \frac{\mathscr{T}_{6}}{\mathscr{T}_{0}^{2}} & \leq & \exp \left\{ -n \left( \frac{s_{3}-t_{3}}{n}\log\pi - \left[ \log\pi-H(U_{3}|Q)+36\eta_{2}\right]\right)  \right\}\leq \exp \left\{ -n\left( \delta_{1}-36\eta_{2}\right) \right\}\nonumber\\
 \frac{\mathscr{T}_{7D}}{\mathscr{T}_{0}^{2}} & \leq & \max_{\theta \neq 0}\exp \left\{ -n \left( \frac{s_{3}}{n}\log\pi - \left[ \log\pi-H(U_{3}\ominus \theta U_{2}|Q)+48\eta_{2}\right]\right)  \right\}\leq \exp \left\{ -n\left( \delta_{1}-\frac{\eta}{8}-48\eta_{2}\right) \right\}.\nonumber
\end{eqnarray}
Substituting, the above bounds in (\ref{Eqn:3To1DBCOneLayerEncoderErrorEventInTermsOfT1ToT7D}), we conclude $P(\epsilon_{l}) \leq (28+8\log\pi)\exp \left\{ -n\left( \delta_{1}-\frac{\eta}{8}-48\eta_{2}\right) \right\}$ for $n \geq \max \{ N_{1}(\eta),N_{2}(\eta_{2}) \}$. In the sequel, we derive a lower bound on $\mathcal{L}(n)$ and prove that for large $n$, $\mathcal{L}(n) > 1$, thereby establishing $\epsilon_{1} \subseteq \epsilon_{l}$. From the definition of $\mathcal{L}(n)$, (\ref{Eqn:3To1DBCOneLayerEncoderErrorEventVarianceOfPhi}), we have
\begin{eqnarray}
\label{Eqn:3DBCProofValueOfListSize}
 \mathcal{L}(n)=\frac{\mathscr{T}_{0}}{2} \geq \frac{|\mathcal{B}_{1}|\pi^{s_{2}+s_{3}}|T_{2\eta_{2}}(V_{1},\underline{U}|q^{n})|}{2\pi^{2n+t_{2}+t_{3}}\exp\left\{ nH(V_{1}|Q)+4n\eta_{2} \right\}},
\end{eqnarray}
for sufficiently large $n$. Moreover, from (\ref{Eqn:3DBCProofTwiceTheListSize}), we note that $\mathcal{L}(n) \geq \frac{1}{2}\exp \left\{ n\left( \delta_{1}+\frac{\eta}{8}-16\eta_{2}\right)\right\}$ for $n \geq \max \{ N_{1}(\eta),$ $N_{2}(\eta_{2}) \}$. By our choice of $\eta, \eta_{2}$, for sufficiently large $n$, we have $\mathcal{L}(n)>1$.

\section{Upper bound on $P((\epsilon_{1}\cup \epsilon_{2}\cup \epsilon_{3})^{c}\cap\epsilon_{41})$}
\label{AppSec:AnalysisOfDecoder1ErrorEventFor3To1DBC}

We begin by introducing some compact notation. We let $\underline{M}^{\underline{t}}$ denote the pair $(M_{2}^{t_{2}},M_{3}^{t_{3}})$ of message random variables. We let $\underline{m}^{\underline{t}}$ denote a generic element $(m_{2}^{t_{2}},m_{3}^{t_{3}}) \in \fieldpi^{\underline{t}}\define \fieldpi^{t_{2}} \times \fieldpi^{t_{3}}$, and similarly $\underline{a}^{\underline{s}}$ denote $(a_{2}^{s_{2}},a_{3}^{s_{3}})\in \fieldpi^{\underline{s}} \define  \fieldpi^{s_{2}} \times \fieldpi^{s_{3}}$. We abbreviate $T_{8\eta_{2}}(V_{1},U_{2}\oplus U_{3}|q^{n},y_{1}^{n})$ as $T_{8\eta_{2}}(V_{1},\oplus|q^{n},y_{1}^{n})$ and the vector $X^{n}(M_{1},M_{2}^{t_{2}},M_{3}^{t_{3}})$ input on the channel as $X^{n}$. Let
\begin{eqnarray}
 \label{Eqn:3To1DBCOneLayerAchievabilityTypicalitySet}
 &\!\!\!\tilde{T}_{\eta_{2}}(q^{n}) \define \left\{ (v_{1}^{n},\underline{u}^{n},x^{n},y_{1}^{n}) \in T_{8\eta_{2}}(V_{1},\underline{U},X,Y_{1}|q^{n}): (v_{1}^{n},\underline{u}^{n}) \in T_{2\eta_{2}}(V_{1},\underline{U}|q^{n}) , (v_{1}^{n},\underline{u}^{n},x^{n}) \in T_{4\eta_{2}}(V_{1},\underline{U},X|q^{n})\right\},
  \nonumber\\
 \lefteqn{\tilde{T}_{\eta_{2}}(q^{n}|v_{1}^{n},\underline{u}^{n}) = \left\{  (x^{n},y_{1}^{n}) : (v_{1}^{n},\underline{u}^{n},x^{n},y_{1}^{n}) \in \tilde{T}_{\eta_{2}}(q^{n})\right\}}.
  \nonumber
\end{eqnarray}
We begin by characterizing the event under question. Denoting $\tilde{\epsilon}_{41} = (\epsilon_{l}\cup\epsilon_{2}\cup\epsilon_{3})^{c}\cap \epsilon_{41}$, we have
\begin{eqnarray}
 \label{Eqn:3To1DBCOneLayerAchievabilityDecoder1ErrorEventCharacterization}
 P(\tilde{\epsilon}_{41}) \leq \sum_{m_{1}}\sum_{\hat{m}_{1}\neq m_{1} } \sum_{\hat{b}_{1} \in \mathcal{B}_{1}}\sum_{\hat{a}_{3}^{s_{3}} } \sum_{\substack{(v_{1}^{n},\underline{u}^{n},x^{n},y_{1}^{n})\\\in \tilde{T}_{\eta_{2}}(q^{n})}} \sum_{\substack{(\hat{v}_{1}^{n},\hat{u}^{n})\in \\ T_{8\eta_{2}}(V_{1},\oplus|q^{n},y_{1}^{n}) }}\!\!\!\!\!\!\!\!P\left(\left\{ \substack{M_{1}=m_{1}, V_{1}^{n}(m_{1},B_{1})=v_{1}^{n},U_{l}^{n}(A_{l}^{s_{l}})=u_{l}^{n}\\I_{l}(A_{l}^{s_{l}})=M_{l}^{t_{l}}:l=2,3,Y_{1}^{n}=y_{1}^{n},X^{n}=x^{n}\\U_{\oplus}^{n}(\hat{a}_{3}^{s_{3}})=\hat{u}^{n},V_{1}^{n}(\hat{m}_{1},\hat{b}_{1})=\hat{v}_{1}^{n}   }\right\}\cap \epsilon_{l}^{c}  \right)
\end{eqnarray}
We consider a generic term in the above sum. Observe that
\begin{eqnarray}
 \label{Eqn:3To1DBCOneLayerAchievabilityGenericTermInSum}
 &\!\!\!\!\!\!\!\!\!P\left( \substack{Y_{1}^{n}=y_{1}^{n}\\X^{n}=x^{n}}\middle|\left\{\substack{M_{1}=m_{1},V_{1}^{n}(m_{1},B_{1})=v_{1}^{n},U_{l}^{n}(A_{l}^{s_{l}})=u_{l}^{n}\\I_{l}(A_{l}^{s_{l}})=M_{l}^{t_{l}}:l=2,3,U_{\oplus}^{n}(\hat{a}_{3}^{s_{3}})=\hat{u}^{n},V_{1}^{n}(\hat{m}_{1},\hat{b}_{1})=\hat{v}_{1}^{n}   }\right\} \cap \epsilon_{l}^{c} \right) = P\left(\substack{Y_{1}^{n}=y_{1}^{n}\\X^{n}=x^{n}}\middle|\substack{V_{1}^{n}(M_{1},B_{1})=v_{1}^{n}\\U_{l}^{n}(A_{l}^{s_{l}})=u_{l}^{n}:l=2,3}\right) =: \theta(y_{1}^{n},x^{n}|v_{1}^{n},\underline{u}^{n}),\\
 \label{Eqn:3To1DBCOneLayerAchievabilityGenericTermInTheSumEvent}
 \lefteqn{P\left( \left\{\substack{M_{1}=m_{1},V_{1}^{n}(m_{1},B_{1})=v_{1}^{n}\\U_{l}^{n}(A_{l}^{s_{l}})=u_{l}^{n},I_{l}(A_{l}^{s_{l}})=M_{l}^{t_{l}}:l=2,3\\U_{\oplus}^{n}(\hat{a}_{3}^{s_{3}})=\hat{u}^{n},V_{1}^{n}(\hat{m}_{1},\hat{b}_{1})=\hat{v}_{1}^{n}   }\right\} \cap \epsilon_{l}^{c}\right) = \sum_{\underline{m}^{\underline{t}}\in \fieldpi^{\underline{t}}}\sum_{\substack{(b_{1},\underline{a}^{\underline{s}})\in\\ \mathcal{B}_{1} \times \fieldpi^{\underline{s}}}}P \left(\left\{\substack{M_{1}=m_{1},V_{1}^{n}(m_{1},b_{1})=v_{1}^{n},U_{l}^{n}(a_{l}^{s_{l}})=u_{l}^{n}\\M_{l}^{t_{l}}=m_{l}^{t_{l}},A_{l}^{s_{l}}=a_{l}^{s_{l}},I_{l}(a_{l}^{s_{l}})=m_{l}^{t_{l}}:l=2,3\\B_{1}=b_{1},U_{\oplus}^{n}(\hat{a}_{3}^{s_{3}})=\hat{u}^{n},V_{1}^{n}(\hat{m}_{1},\hat{b}_{1})=\hat{v}_{1}^{n}   }\right\} \cap \epsilon_{l}^{c} \right),}
\end{eqnarray}
and the product of left hand sides of (\ref{Eqn:3To1DBCOneLayerAchievabilityGenericTermInSum}) and (\ref{Eqn:3To1DBCOneLayerAchievabilityGenericTermInTheSumEvent}) is a generic term in (\ref{Eqn:3To1DBCOneLayerAchievabilityDecoder1ErrorEventCharacterization}). We now consider a generic term on the right hand side of (\ref{Eqn:3To1DBCOneLayerAchievabilityGenericTermInTheSumEvent}). Note that
\begin{eqnarray}
 P \left(E \cap \left\{\substack{B_{1}=b_{1},A_{l}^{s_{l}}=a_{l}^{s_{l}} }\right\} \cap \epsilon_{l}^{c} \right) \leq P(E)P(\left\{\substack{B_{1}=b_{1},A_{l}^{s_{l}}=a_{l}^{s_{l}} }\right\}| E \cap  \epsilon_{l}^{c} ) \leq \frac{P(E)}{\mathcal{L}(n)},\nonumber
 \end{eqnarray}
where $E$ abbreviates the event $
 \left\{\substack{M_{1}=m_{1},V_{1}^{n}(m_{1},b_{1})=v_{1}^{n},U_{l}^{n}(a_{l}^{s_{l}})=u_{l}^{n},M_{l}^{t_{l}}=m_{l}^{t_{l}},I_{l}(a_{l}^{s_{l}})=m_{l}^{t_{l}}:l=2,3,U_{\oplus}^{n}(\hat{a}_{3}^{s_{3}})=\hat{u}^{n},V_{1}^{n}(\hat{m}_{1},\hat{b}_{1})=\hat{v}_{1}^{n}   }\right\}$.
Substituting the above in (\ref{Eqn:3To1DBCOneLayerAchievabilityGenericTermInTheSumEvent}), we have
\begin{eqnarray}
 \label{Eqn:3DBCProofTheListEntersIntoPicture}
 P\left( \left\{\substack{M_{1}=m_{1},V_{1}^{n}(m_{1},B_{1})=v_{1}^{n}\\U_{l}^{n}(A_{l}^{s_{l}})=u_{l}^{n},I_{l}(A_{l}^{s_{l}})=M_{l}^{t_{l}}:l=2,3\\U_{\oplus}^{n}(\hat{a}_{3}^{s_{3}})=\hat{u}^{n},V_{1}^{n}(\hat{m}_{1},\hat{b}_{1})=\hat{v}_{1}^{n}   }\right\} \cap \epsilon_{l}^{c}\right) \leq \frac{1}{\mathcal{L}(n)}\sum_{\underline{m}^{\underline{t}}\in \fieldpi^{\underline{t}}}\sum_{\substack{(b_{1},\underline{a}^{\underline{s}})\\\in \mathcal{B}_{1} \times \mathscr{D}(\hat{a}^{s_{3}})}} P\left(  \substack{M_{1}=m_{1},V_{1}^{n}(m_{1},b_{1})=v_{1}^{n},U_{l}^{n}(a_{l}^{s_{l}})=u_{l}^{n},M_{l}^{t_{l}}=m_{l}^{t_{l}}\\I_{l}(a_{l}^{s_{l}})=m_{l}^{t_{l}}:l=2,3,U_{\oplus}^{n}(\hat{a}_{3}^{s_{3}})=\hat{u}^{n},V_{1}^{n}(\hat{m}_{1},\hat{b}_{1})=\hat{v}_{1}^{n}   } \right)
\nonumber\\
  \label{Eqn:3To1DBCOneLayerAchievabilityGenericTermInTheSumEventRestated}
  +\frac{1}{\mathcal{L}(n)}\sum_{\underline{m}^{\underline{t}}\in \fieldpi^{\underline{t}}}\sum_{\substack{(b_{1},\underline{a}^{\underline{s}})\\\in \mathcal{B}_{1} \times \mathscr{I}(\hat{a}^{s_{3}})}}\!\!\!\!\!\!\!\! P\left(  \substack{M_{1}=m_{1},V_{1}^{n}(m_{1},b_{1})=v_{1}^{n},U_{l}^{n}(a_{l}^{s_{l}})=u_{l}^{n},M_{l}^{t_{l}}=m_{l}^{t_{l}}\\I_{l}(a_{l}^{s_{l}})=m_{l}^{t_{l}}:l=2,3,U_{\oplus}^{n}(\hat{a}_{3}^{s_{3}})=\hat{u}^{n},V_{1}^{n}(\hat{m}_{1},\hat{b}_{1})=\hat{v}_{1}^{n}  } \right).
\end{eqnarray}
where $\mathscr{D}(\hat{a}^{s_{3}}) \define \left\{ \underline{a}^{\underline{s}}:(a_{2}^{s_{2}}0^{s_{+}})\oplus a_{3}^{s_{3}}=\hat{a}^{s_{3}} \right\}$, $s_{+} = s_{3}-s_{2}$ and $\mathscr{I}(\hat{a}^{s_{3}}) \define \fieldpi^{s_{2}} \times \fieldpi^{s_{3}} \setminus \mathscr{D}(\hat{a}^{s_{3}})$. Let us evaluate a generic term in the right hand side of (\ref{Eqn:3To1DBCOneLayerAchievabilityGenericTermInTheSumEventRestated}). The collection $M_{1},M_{2}^{t_{2}},M_{3}^{t_{3}},V_{1}^{n}(m_{1},b_{1}),I_{2}(a^{s_{2}}),I_{3}(a^{s_{3}}),(U_{l}(a_{l}^{s_{l}}):l=2,3,U_{\oplus}(\hat{a}_{3}^{s_{3}})),V_{1}^{n}(\hat{m}_{1},\hat{b}_{1})$ are mutually independent, where $(U_{l}(a_{l}^{s_{l}}):l=2,3,U_{\oplus}(\hat{a}_{3}^{s_{3}}))$ is treated as a single random object. If $(a_{2}^{s_{2}},a_{3}^{s_{3}}) \in \mathscr{D}(\hat{a}^{s_{3}})$, then \begin{eqnarray}P(U_{l}(a_{l}^{s_{l}})=u_{l}^{n}:l=2,3,U_{\oplus}(\hat{a}_{3}^{s_{3}})=\hat{u}^{n})=\left\{ \begin{array}{ll} \frac{1}{\pi^{2n}}&\mbox{ if } u_{2}^{n}\oplus u_{3}^{n}=\hat{u}^{n}\\0&\mbox{ otherwise.}\end{array}.\right.\nonumber\end{eqnarray}
Otherwise, i.e., $(a_{2}^{s_{2}},a_{3}^{s_{3}}) \in \mathscr{I}(\hat{a}^{s_{3}})$, a counting argument similar to that employed in appendix \ref{AppSec:AnalysisOfDecoder2And3ErrorEventsFor3To1DBC} proves $P(U_{l}(a_{l}^{s_{l}})=u_{l}^{n}:l=2,3,U_{\oplus}(\hat{a}_{3}^{s_{3}})=\hat{u}^{n})=\frac{1}{\pi^{3n}}$. We therefore have
\begin{eqnarray}
 \label{Eqn:3To1DBCOneLayerAchievabilityDecoder1ErrorEventOneCoreElement}
 P\left(  \substack{M_{1}=m_{1},V_{1}^{n}(m_{1},b_{1})=v_{1}^{n},U_{l}^{n}(a_{l}^{s_{l}})=u_{l}^{n},M_{l}^{t_{l}}=m_{l}^{t_{l}}\\I_{l}(a_{l}^{s_{l}})=m_{l}^{t_{l}}:l=2,3,U_{\oplus}^{n}(\hat{a}_{3}^{s_{3}})=\hat{u}^{n},V_{1}^{n}(\hat{m}_{1},\hat{b}_{1})=\hat{v}_{1}^{n}} \right) = \left\{  \begin{array}{ll}
\frac{P\left(\substack{M_{1}=m_{1},V_{1}^{n}(m_{1},b_{1})=v_{1}^{n}\\\underline{M}^{\underline{t}}=\underline{m}^{\underline{t}},V_{1}^{n}(\hat{m}_{1},\hat{b}_{1})=\hat{v}_{1}^{n}}\right)}{\pi^{2n+t_{2}+t_{3}}} &\mbox{ if }(a_{2}^{s_{2}},a_{3}^{s_{3}}) \in \mathscr{D}(\hat{a}^{s_{3}})\\&\mbox{ and }u_{2}^{n}\oplus u_{3}^{n}=\hat{u}^{n}\\
\frac{P\left(\substack{M_{1}=m_{1},V_{1}^{n}(m_{1},b_{1})=v_{1}^{n}\\\underline{M}^{\underline{t}}=\underline{m}^{\underline{t}},V_{1}^{n}(\hat{m}_{1},\hat{b}_{1})=\hat{v}_{1}^{n}}\right)}{\pi^{3n+t_{2}+t_{3}}} &\mbox{ if }(a_{2}^{s_{2}},a_{3}^{s_{3}}) \in \mathscr{I}(\hat{a}^{s_{3}})
\end{array}
  \right.
\end{eqnarray}
Substituting (\ref{Eqn:3To1DBCOneLayerAchievabilityDecoder1ErrorEventOneCoreElement}) in (\ref{Eqn:3To1DBCOneLayerAchievabilityGenericTermInTheSumEventRestated}) and recognizing that product of right hand sides of (\ref{Eqn:3To1DBCOneLayerAchievabilityGenericTermInTheSumEvent}), (\ref{Eqn:3To1DBCOneLayerAchievabilityGenericTermInSum}) is a generic term in the sum (\ref{Eqn:3To1DBCOneLayerAchievabilityDecoder1ErrorEventCharacterization}), we have
\begin{eqnarray}
 \label{Eqn:3To1DBCOneLayerAchievabilityDecoder1SumBlownUp}
P(\tilde{\epsilon}_{41})\leq \sum_{(m_{1},\underline{m}^{\underline{t}})}\sum_{\hat{m}_{1}\neq m_{1} } \sum_{\hat{b}_{1} \in \mathcal{B}_{1}}\sum_{\hat{a}_{3}^{s_{3}} } \sum_{\substack{(b_{1},\underline{a}^{\underline{s}})\\\in \mathcal{B}_{1} \times \mathscr{D}(\hat{a}^{s_{3}})}}\sum_{\substack{(v_{1}^{n},\underline{u}^{n},x^{n},y_{1}^{n})\\\in \tilde{T}_{\eta_{2}}(q^{n})}} \!\!\!\!\!\!\theta(y_{1}^{n},x^{n}|v_{1}^{n},\underline{u}^{n})\!\!\!\!\!\!\sum_{\substack{(\hat{v}_{1}^{n},u_{2}^{n}\oplus u_{3}^{n})\in \\ T_{8\eta_{2}}(V_{1},\oplus|q^{n},y_{1}^{n}) }}\!\!\!\!\!\!\!\!\!\frac{P\left(\substack{M_{1}=m_{1},V_{1}^{n}(m_{1},b_{1})=v_{1}^{n}\\\underline{M}^{\underline{t}}=\underline{m}^{\underline{t}},V_{1}^{n}(\hat{m}_{1},\hat{b}_{1})=\hat{v}_{1}^{n}}\right)}{\pi^{2n+t_{2}+t_{3}}\mathcal{L}(n)}
\nonumber\\+
\sum_{(m_{1},\underline{m}^{\underline{t}})}\sum_{\hat{m}_{1}\neq m_{1} } \sum_{\hat{b}_{1} \in \mathcal{B}_{1}}\sum_{\hat{a}_{3}^{s_{3}} } \sum_{\substack{(b_{1},\underline{a}^{\underline{s}})\\\in \mathcal{B}_{1} \times \mathscr{I}(\hat{a}^{s_{3}})}}\sum_{\substack{(v_{1}^{n},\underline{u}^{n},x^{n},y_{1}^{n})\\\in \tilde{T}_{\eta_{2}}(q^{n})}} \!\!\!\!\!\!\theta(y_{1}^{n},x^{n}|v_{1}^{n},\underline{u}^{n})\!\!\!\!\!\!\sum_{\substack{(\hat{v}_{1}^{n},\hat{u}^{n})\in \\ T_{8\eta_{2}}(V_{1},\oplus|q^{n},y_{1}^{n}) }}\!\!\!\!\!\!\!\!\!\frac{P\left(\substack{M_{1}=m_{1},V_{1}^{n}(m_{1},b_{1})=v_{1}^{n}\\\underline{M}^{\underline{t}}=\underline{m}^{\underline{t}},V_{1}^{n}(\hat{m}_{1},\hat{b}_{1})=\hat{v}_{1}^{n}}\right)}{\pi^{3n+t_{2}+t_{3}}\mathcal{L}(n)}
\nonumber
\end{eqnarray}
The codewords over $\mathcal{V}^{n}$ are picked independently and identically with respect to $p^{n}_{V_{1}|Q}(\cdot|q^{n})$ and hence by conditional frequency typicality, we have \[P\left(M_{1}=m_{1},V_{1}^{n}(m_{1},b_{1})=v_{1}^{n},\underline{M}^{\underline{t}}=\underline{m}^{\underline{t}},V_{1}^{n}(\hat{m}_{1},\hat{b}_{1})=\hat{v}_{1}^{n}\right) \leq \exp\left\{ -n(2H(V_{1}|Q)-20\eta_{2}) \right\}P(M_{1}=m_{1},\underline{M}^{\underline{t}}=\underline{m}^{\underline{t}})\]for the pairs $(v_{1}^{n},\hat{v}_{1}^{n})$ in question. This upper bound being independent of the arguments in the summation, we only need to compute the number of terms in the summations. For a fixed pair $(u_{2}^{n},u_{3}^{n})$, conditional frequency typicality results guaranty existence of $N_{4}(\eta_{2}) \in \naturals$ such that for all $n \geq N_{4}(\eta_{2})$, we have $|\left\{v_{1}^{n} : (v_{1}^{n},u_{2}\oplus u_{3}^{n}) \in T_{8\eta_{2}}(V_{1},U_{2}\oplus U_{3}|q^{n},y_{1}^{n}) \right\}| \leq \exp \left\{ n(H(V_{1}|Q,U_{2}\oplus U_{3},Y_{1})+32\eta_{2}) \right\}$ and $|T_{8\eta_{2}}(V_{1},U_{2}\oplus U_{3}|q^{n},y_{1}^{n})| \leq \exp \left\{ n(H(V_{1},U_{2}\oplus U_{3}|Q,Y_{1})+32\eta_{2}) \right\}$.  Substituting this upper bound, the inner most summation turns out to be
\begin{eqnarray}
\sum_{\substack{(\hat{v}_{1}^{n},u_{2}^{n}\oplus u_{3}^{n})\in \\ T_{8\eta_{2}}(V_{1},\oplus|q^{n},y_{1}^{n}) }}\!\!\!\!\!\!\!\!\!\frac{P\left(\substack{M_{1}=m_{1},V_{1}^{n}(m_{1},b_{1})=v_{1}^{n}\\\underline{M}^{\underline{t}}=\underline{m}^{\underline{t}},V_{1}^{n}(\hat{m}_{1},\hat{b}_{1})=\hat{v}_{1}^{n}}\right)}{\pi^{2n+t_{2}+t_{3}}} \leq \exp\left\{ -n\left(\substack{2H(V_{1}|Q)-52\eta_{2}\\-H(V_{1}|Q,U_{2}\oplus U_{3},Y_{1})}\right) \right\}\frac{P(M_{1}=m_{1},\underline{M}^{\underline{t}}=\underline{m}^{\underline{t}})}{\pi^{2n+t_{2}+t_{3}}\mathcal{L}(n)}=:\beta_{1},
 \nonumber\\
 \sum_{\substack{(\hat{v}_{1}^{n},\hat{u}^{n})\in \\ T_{8\eta_{2}}(V_{1},\oplus|q^{n},y_{1}^{n}) }}\!\!\!\!\!\!\!\!\!\frac{P\left(\substack{M_{1}=m_{1},V_{1}^{n}(m_{1},b_{1})=v_{1}^{n}\\\underline{M}^{\underline{t}}=\underline{m}^{\underline{t}},V_{1}^{n}(\hat{m}_{1},\hat{b}_{1})=\hat{v}_{1}^{n}}\right)}{\pi^{3n+t_{2}+t_{3}}} \leq \exp\left\{ -n\left(\substack{2H(V_{1}|Q)-52\eta_{2}\\-H(V_{1},U_{2}\oplus U_{3}|Q,Y_{1})}\right) \right\}\frac{P(M_{1}=m_{1},\underline{M}^{\underline{t}}=\underline{m}^{\underline{t}})}{\pi^{3n+t_{2}+t_{3}}\mathcal{L}(n)}=:\beta_{2}
 \nonumber
\end{eqnarray}
Substituting $\beta_{1}$ and $\beta_{2}$, we have
\begin{eqnarray}
 & \!\!\!\!\!\!\!\!\!\!\!\!\!\!\!\!\!\!\!\!\!\!\!\!\!\!\!\!\!\!\!\!\!\!\!\!\!\!\!\!\!\!\!\!\!\!\!\!\!\!\!\!\!\!\!\!\!\!\!\!\!\!\!\!\!\!\!\!\!\!\displaystyle \!\!\!\!\!\!\!\!\!\!P(\tilde{\epsilon}_{41}) \leq \sum_{(m_{1},\underline{m}^{\underline{t}})}\sum_{\hat{m}_{1}\neq m_{1} } \sum_{\substack{\hat{b}_{1} \in \mathcal{B}_{1}\\\hat{a}^{s_{3}}\in\fieldpi^{s_{3}}}} \sum_{\substack{(b_{1},\underline{a}^{\underline{s}})\in \\\mathcal{B}_{1} \times \mathscr{D}(\hat{a}^{s_{3}})}}\sum_{\substack{(v_{1}^{n},\underline{u}^{n})\in\\ T_{2\eta_{2}}(V_{1},\underline{U}|q^{n})}}\sum_{\substack{(x^{n},y_{1}^{n}) \in \\ \tilde{T}_{\eta_{2}}(q^{n}|v_{1}^{n},\underline{u}^{n})}}\theta(y_{1}^{n},x^{n}|v_{1}^{n},\underline{u}^{n})\beta_{1}\nonumber\\
 \lefteqn{~~~~~~~~~~~~~~+\sum_{(m_{1},\underline{m}^{\underline{t}})}\sum_{\hat{m}_{1}\neq m_{1} } \sum_{\hat{b}_{1} \in \mathcal{B}_{1}}\sum_{\hat{a}_{3}^{s_{3}} } \sum_{\substack{(b_{1},\underline{a}^{\underline{s}})\in \\\mathcal{B}_{1} \times \mathscr{I}(\hat{a}^{s_{3}})}}\sum_{\substack{(v_{1}^{n},\underline{u}^{n})\in\\ T_{2\eta_{2}}(V_{1},\underline{U}|q^{n})}}\sum_{\substack{(x^{n},y_{1}^{n}) \in \\ \tilde{T}_{\eta_{2}}(q^{n}|v_{1}^{n},\underline{u}^{n})}}\theta(y_{1}^{n},x^{n}|v_{1}^{n},\underline{u}^{n})\beta_{2}}
 \nonumber\\
&\displaystyle \!\!\!\!\!\leq \sum_{(m_{1},\underline{m}^{\underline{t}})}\sum_{\hat{m}_{1}\neq m_{1} } \sum_{\substack{\hat{b}_{1} \in \mathcal{B}_{1}\\\hat{a}^{s_{3}}\in\fieldpi^{s_{3}}}} \sum_{\substack{(b_{1},\underline{a}^{\underline{s}})\in \\\mathcal{B}_{1} \times \mathscr{D}(\hat{a}^{s_{3}})}}\sum_{\substack{(v_{1}^{n},\underline{u}^{n})\in\\ T_{2\eta_{2}}(V_{1},\underline{U}|q^{n})}}\beta_{1}
 +\sum_{(m_{1},\underline{m}^{\underline{t}})}\sum_{\hat{m}_{1}\neq m_{1} } \sum_{\hat{b}_{1} \in \mathcal{B}_{1}}\sum_{\hat{a}_{3}^{s_{3}} } \sum_{\substack{(b_{1},\underline{a}^{\underline{s}})\in \\\mathcal{B}_{1} \times \mathscr{I}(\hat{a}^{s_{3}})}}\sum_{\substack{(v_{1}^{n},\underline{u}^{n})\in\\ T_{2\eta_{2}}(V_{1},\underline{U}|q^{n})}}\beta_{2}\nonumber
\end{eqnarray}
The terms in the first and second summation are identical to $\beta_{1}$ and $\beta_{2}$ respectively. Multiplying each with the corresponding number of terms, employing the lower bound for $\mathcal{L}(n)$ derived in (\ref{Eqn:3DBCProofValueOfListSize}), it maybe verified that $P(\tilde{\epsilon}_{41}) \leq \mathscr{T}_{1}+\mathscr{T}_{2}$, where 
\begin{eqnarray}
 \mathscr{T}_{1} &=& 2\exp \left\{ -n \left( \left[ {I(V_{1};U_{2}\oplus U_{3},Y_{1}|Q)-56\eta_{2}}\right] - \left[ \frac{\log |\mathcal{B}_{1}|}{n}+\frac{\log |\mathcal{M}_{1}|}{n} \right]\right)\right\} \nonumber\\
 \mathscr{T}_{2}&=& 2\exp \left\{ -n \left( \left[\log\pi+H(V_{1}|Q)-H(V_{1},U_{2}\oplus U_{3}|Q,Y_{1})-56\eta_{2} \right] - \left[ \frac{\log |\mathcal{B}_{1}|}{n}+\frac{\log |\mathcal{M}_{1}|}{n} +\frac{s_{3}\log\pi}{n}\right]  \right)  \right\}.
 \nonumber
\end{eqnarray}
From bounds on the parameters of the code ((\ref{Eqn:3To1DBCOneLayerAchievabilityLowerAndUpperBoundsOnSj}) - (\ref{Eqn:3To1DBCOneLayerAchievabilityLowerAndUpperBoundsOnR1AndK1})), it maybe verified that for $n\geq \max \{ N_{1}(\eta),N_{j}(\eta_{2}):j=2,3,4 \}$, $P((\epsilon_{l}\cup \epsilon_{2}\cup \epsilon_{3})^{c}\cap \epsilon_{41})   \leq  4 \exp \left\{ -n\left(\delta_{1}+\frac{\eta}{4}-56\eta_{2} \right) \right\}$.

\section{Upper bound on $P((\epsilon_{1}\cup \epsilon_{2}\cup \epsilon_{3})^{c}\cap\epsilon_{4j})$ for $3-$DBC}
\label{AppSec:AnalysisOfDecoder2And3ErrorEventsFor3To1DBC}
We begin by introducing some compact notation similar to that introduced in appendix \ref{AppSec:AnalysisOfDecoder1ErrorEventFor3To1DBC}. We let $\underline{M}^{\underline{t}}$ denote the pair $(M_{2}^{t_{2}},M_{3}^{t_{3}})$ of message random variables. We let $\underline{m}^{\underline{t}}$ denote a generic element $(m_{2}^{t_{2}},m_{3}^{t_{3}}) \in \fieldpi^{\underline{t}}\define \fieldpi^{t_{2}} \times \fieldpi^{t_{3}}$, and similarly $\underline{a}^{\underline{s}}$ denote $(a_{2}^{s_{2}},a_{3}^{s_{3}})\in \fieldpi^{\underline{s}} \define  \fieldpi^{s_{2}} \times \fieldpi^{s_{3}}$. We let
\begin{eqnarray}
 \label{Eqn:3To1DBCOneLayerAchievabilityTypicalitySet}
 &\!\!\!\hat{T}_{\eta_{2}}(q^{n}) \define \left\{ (v_{1}^{n},\underline{u}^{n},x^{n},y_{j}^{n}) \in T_{8\eta_{2}}(V_{1},\underline{U},X,Y_{j}|q^{n}): (v_{1}^{n},\underline{u}^{n}) \in T_{2\eta_{2}}(V_{1},\underline{U}|q^{n}) , (v_{1}^{n},\underline{u}^{n},x^{n}) \in T_{4\eta_{2}}(V_{1},\underline{U},X|q^{n})\right\},
  \nonumber\\
 \lefteqn{~\hat{T}_{\eta_{2}}(q^{n}|v_{1}^{n},\underline{u}^{n}) = \left\{  (x^{n},y_{j}^{n}) : (v_{1}^{n},\underline{u}^{n},x^{n},y_{j}^{n}) \in \hat{T}_{\eta_{2}}(q^{n})\right\}}
  \nonumber
\end{eqnarray}
We begin by characterizing the event under question. For $j=2,3$, denoting $\tilde{\epsilon}_{4j}\define (\epsilon_{l}\cup\epsilon_{2}\cup\epsilon_{3})^{c}\cap \epsilon_{4j}$, we have
\begin{eqnarray}
 \label{Eqn:3To1DBCOneLayerAchievabilityDecoder2ErrorEventCharacterization}
 P(\tilde{\epsilon}_{4j}) \leq \sum_{(m_{1},\underline{m}^{\underline{t}})}\sum_{\hat{m}_{j}^{t_{j}} \neq M_{j}^{t_{j}} } \sum_{\hat{a}_{j}^{s_{j}} } \sum_{\substack{(v_{1}^{n},\underline{u}^{n},x^{n},y_{j}^{n})\\\in \hat{T}_{\eta_{2}}(q^{n})}} \sum_{\substack{\hat{u}_{j}^{s_{j}}\in \\ T_{8\eta_{2}}(U_{j}|q^{n},y_{j}^{n}) }}\!\!\!\!\!\!P\left( \left\{\substack{M_{1}=m_{1},\underline{M}^{\underline{t}}=\underline{m}^{\underline{t}},V_{1}^{n}(m_{1},B_{1})=v_{1}^{n}\\U_{l}^{n}(A_{l}^{s_{l}})=u_{l}^{n},I_{l}(A^{s_{l}})=m_{l}^{t_{l}}:l=2,3,Y_{j}^{n}=y_{j}^{n}\\X^{n}=x^{n},U_{j}^{n}(\hat{a}_{j}^{s_{j}})=\hat{u}_{j}^{n},I_{j}(\hat{a}_{j}^{s_{j}})=\hat{m}_{j}^{t_{j}}   } \right\} \cap \epsilon_{l}^{c} \right),
\end{eqnarray}
where $X^{n}$ abbreviates $X^{n}(M_{1},\underline{M}^{\underline{t}})$, the random vector input on the channel. We consider a generic term in the above sum. Observe that
\begin{eqnarray}
 \label{Eqn:3To1DBCOneLayerAchievabilityGenericTermInTheSum}
 P\left( \substack{Y_{j}^{n}=y_{j}^{n}\\X^{n}=x^{n}}\middle|\left\{\substack{M_{1}=m_{1},\underline{M}^{\underline{t}}=\underline{m}^{\underline{t}},V_{1}^{n}(m_{1},B_{1})=v_{1}^{n}\\U_{l}^{n}(A_{l}^{s_{l}})=u_{l}^{n},I_{l}(A^{s_{l}})=m_{l}^{t_{l}}:l=2,3\\U_{j}^{n}(\hat{a}_{j}^{s_{j}})=\hat{u}_{j}^{n},I_{j}(\hat{a}_{j}^{s_{j}})=\hat{m}_{j}^{t_{j}}   }  \right\} \cap \epsilon_{l}^{c}\right) = P\left(\substack{Y_{j}^{n}=y_{j}^{n}\\X^{n}=x^{n}}\middle|\substack{V_{1}^{n}(M_{1},B_{1})=v_{1}^{n}\\U_{l}^{n}(A_{l}^{s_{l}})=u_{l}^{n}:l=2,3}\right) =: \theta(y^{n},x^{n}|v_{1}^{n},\underline{u}^{n}),
 \\
 \label{Eqn:3To1DBCOneLayerAchievabilityGenericTermInTheSumConditioningEvent}
 P\left(  \left\{\substack{M_{1}=m_{1},\underline{M}^{\underline{t}}=\underline{m}^{\underline{t}},V_{1}^{n}(m_{1},B_{1})=v_{1}^{n}\\U_{l}^{n}(A_{l}^{s_{l}})=u_{l}^{n},I_{l}(A^{s_{l}})=m_{l}^{t_{l}}:l=2,3\\U_{j}^{n}(\hat{a}_{j}^{s_{j}})=\hat{u}_{j}^{n},I_{j}(\hat{a}_{j}^{s_{j}})=\hat{m}_{j}^{t_{j}}   }  \right\} \cap \epsilon_{l}^{c} \right) = \sum_{\substack{(b_{1},\underline{a}^{\underline{s}})\\\in \mathcal{B}_{1} \times \fieldpi^{\underline{s}}}} \!\!\!P\left( E \cap \left\{\substack{B_{1}=b_{1}\\\underline{A}^{\underline{s}}=\underline{a}^{\underline{s}}}\right\}\cap \epsilon_{l}^{c}  \right) \leq \sum_{\substack{(b_{1},\underline{a}^{\underline{s}})\\\in \mathcal{B}_{1} \times \fieldpi^{\underline{s}}}}\!\!\!P(E)P\left( \left\{\substack{B_{1}=b_{1}\\\underline{A}^{\underline{s}}=\underline{a}^{\underline{s}} }\right\}| E \cap \epsilon_{l}^{c}\right),
\end{eqnarray}
where $E$ abbreviates the event $\left\{\substack{M_{1}=m_{1},~\underline{M}^{\underline{t}}=\underline{m}^{\underline{t}},~V_{1}^{n}(m_{1},b_{1})=v_{1}^{n},~U_{l}^{n}(a_{l}^{s_{l}})=u_{l}^{n},~I_{l}(a^{s_{l}})=m_{l}^{t_{l}}:l=2,3,~U_{j}^{n}(\hat{a}_{j}^{s_{j}})=\hat{u}_{j}^{n},~I_{j}(\hat{a}_{j}^{s_{j}})=\hat{m}_{j}^{t_{j}}   } \right\}$. We now focus on the terms on the right hand side of (\ref{Eqn:3To1DBCOneLayerAchievabilityGenericTermInTheSumConditioningEvent}). By the encoding rule, $P\left( \left\{\substack{B_{1}=b_{1},~\underline{A}^{\underline{s}}~=\underline{a}^{\underline{s}} }\right\}| E \cap \epsilon_{l}^{c}\right)= \frac{1}{\mathcal{L}(n)}$. We are left to evaluate $P(E)$. The collection $M_{1},M_{2}^{t_{2}},M_{3}^{t_{3}},V_{1}^{n}(m_{1},b_{1}),I_{2}(a^{s_{2}}),I_{3}(a^{s_{3}}),I_{j}(\hat{a}^{s_{j}}),(U_{l}(a_{l}^{s_{l}}):l=2,3,U_{j}(\hat{a}_{j}^{s_{j}}))$ are mutually independent, where $(U_{l}(a_{l}^{s_{l}}):l=2,3,U_{j}(\hat{a}_{j}^{s_{j}}))$ is treated as a single random object. The following counting argument proves the triplet $U_{l}(a_{l}^{s_{l}}):l=2,3,U_{j}(\hat{a}_{j}^{s_{j}})$ also to be mutually independent. Let $\left\{ j,\msout{j} \right\}=\left\{2,3\right\}$. For any $u_{j}^{n},u_{\msout{j}}^{n}$ and $\hat{u}_{j}^{n}$, let us study
\begin{eqnarray}
 \label{Eqn:3To1DBCOneLayerAchievabilityCountingArgument}
 \left| \left\{ (g_{2},g_{3/2},b_{2}^{n},b_{3}^{n}): a_{j}^{s_{j}}g_{j}\oplus b_{j}^{n}=u_{j}^{n}, a_{\msout{j}}^{s_{\msout{j}}}g_{\msout{j}}\oplus b_{\msout{j}}^{n}=u_{\msout{j}}^{n}, (\hat{a}_{j}^{s_{j}}\ominus a_{j}^{s_{j}})g_{j}=\hat{u}_{j}^{n}-u_{j}^{n} \right\}  \right|.
 \nonumber
\end{eqnarray}
There exists a $t$ such that $\hat{a}_{jt}^{s_{j}} \neq a_{jt}^{s_{j}}$. For any choice of rows $1,2,\cdots, t-1, t+1,\cdots,s_{3}$ of $g_{3}$, one can choose the $t$th row of $g_{j}$ and $b_{2}^{n},b_{3}^{n}$ such that the above conditions are satisfied. The cardinality of the above set is $\pi^{(s_{3}-1)n}$. The uniform distribution and mutual independence guarantee $P(U_{l}(a_{l}^{s_{l}})=u_{l}^{n}:l=2,3,U_{j}(\hat{a}_{j}^{s_{j}})=\hat{u}_{j}^{n})=\frac{1}{\pi^{3n}}$.

We therefore have
\begin{eqnarray}
\label{Eqn:3To1DBCOneLayerAchievabilityGenTermInTheSumConditioningEventSimplified}
 P\left(  \substack{M_{1}=m_{1},\underline{M}^{\underline{t}}=\underline{m}^{\underline{t}},V_{1}^{n}(m_{1},b_{1})=v_{1}^{n},\\U_{l}^{n}(a_{l}^{s_{l}})=u_{l}^{n},I_{l}(a^{s_{l}})=m_{l}^{t_{l}}:l=2,3,\\U_{j}^{n}(\hat{a}_{j}^{s_{j}})=\hat{u}_{j}^{n},I_{j}(\hat{a}_{j}^{s_{j}})=\hat{m}_{j}^{t_{j}}   } \right)=\frac{P(M_{1}=m_{1},\underline{M}^{\underline{t}}=\underline{m}^{\underline{t}},V_{1}^{n}(m_{1},b_{1})=v_{1}^{n})}{\pi^{3n+t_{2}+t_{3}+t_{j}}}
\end{eqnarray}
Substituting (\ref{Eqn:3To1DBCOneLayerAchievabilityGenTermInTheSumConditioningEventSimplified}), (\ref{Eqn:3To1DBCOneLayerAchievabilityGenericTermInTheSumConditioningEvent}) and (\ref{Eqn:3To1DBCOneLayerAchievabilityGenericTermInTheSum}) in (\ref{Eqn:3To1DBCOneLayerAchievabilityDecoder2ErrorEventCharacterization}), we have
\begin{eqnarray}
 \label{Eqn:Eqn:3To1DBCOneLayerAchievabilityDecoder2ErrorEventAfterUnionBound}
P(\tilde{\epsilon}_{4j}) \leq  \nonumber \!\!\!\! \sum_{\left(m_{1},\underline{m}^{\underline{t}}\right)} \sum_{(b_{1},\underline{a}^{\underline{s}})}\sum_{\hat{m}_{j}^{t_{j}} \neq m_{j}^{t_{j}} }\sum_{\substack{\hat{a}_{j}^{s_{j}}}}\sum_{\substack{(v_{1}^{n},\underline{u}^{n},x^{n},y_{j}^{n})\\\in \hat{T}_{\eta_{2}}(q^{n})}}\!\!\!\!\!\!\!\!\!\!\!\!\theta(y^{n},x^{n}|v_{1}^{n},\underline{u}^{n})\!\!\!\!\!\!\!\!\!\!\!\! \sum_{\substack{\hat{u}_{j}^{n} \in \\ T_{16\eta_{2}}(U_{j}|q^{n},y_{j}^{n})}}\!\!\!\!\!\!\!\!\!\!\!\frac{P(M_{1}=m_{1},\underline{M}^{\underline{t}}=\underline{m}^{\underline{t}},V_{1}^{n}(m_{1},b_{1})=v_{1}^{n})}{\pi^{3n+t_{2}+t_{3}+t_{j}}\mathcal{L}(n)}. \nonumber
\end{eqnarray}
Note that terms in the innermost sum do not depend on the arguments of the sum. We now employ the bounds on the cardinality of conditional typical sets. There exists $N_{5}(\eta_{2}) \in \naturals$ such that for all $n \geq N_{5}(\eta_{2})$, we have $|T_{16\eta_{2}}(U_{j}|q^{n},y_{j}^{n})| \leq \exp \{ n(H(U_{j}|Q,Y_{j})+32\eta_{2}) \}$ for all $(q^{n},y_{j}^{n}) \in T_{8\eta_{2}}(Q,Y_{j})$. For $n \geq \max \{ N_{1}(\eta),N_{5}(\eta_{2}) \}$, we therefore have
\begin{eqnarray}
&&\!\!\!\!\!\!\!\!\!\!\!\!\!\!\!\!\!\!\!\!\!P(\tilde{\epsilon}_{4j})\leq \nonumber \!\!\!\!\sum_{\left(m_{1},\underline{m}^{\underline{t}}\right)} \sum_{(b_{1},\underline{a}^{\underline{s}})}\sum_{\hat{m}_{j}^{t_{j}} \neq m_{j}^{t_{j}} }\sum_{\substack{\hat{a}_{j}^{s_{j}}}}\sum_{\substack{(v_{1}^{n},\underline{u}^{n})\\\in T_{2\eta_{2}}(V_{1},\underline{U}|q^{n})}}\frac{P\left( \substack{V_{1}(m_{1},b_{1})=v_{1}^{n},M_{1}=m_{1}\\M_{l}^{t_{l}}=m_{l}^{t_{l}}:l=2,3}\right)\exp \left\{ n32\eta_{2} \right\}}{\pi^{3n+t_{2}+t_{3}+t_{j}}\exp \left\{ -nH(U_{j}|Q,Y_{j}) \right\}} \sum_{\substack{(x^{n},y_{j}^{n}) \in \\ \hat{T}_{\eta_{2}}(q^{n}|v_{1}^{n},\underline{u}^{n})}}\!\!\!\!\!\!\frac{\theta(y^{n},x^{n}|v_{1}^{n},\underline{u}^{n})}{\mathcal{L}(n)}\nonumber\\
&&\leq \nonumber \sum_{\left(m_{1},\underline{m}^{\underline{t}}\right)} \sum_{(b_{1},\underline{a}^{\underline{s}})}\sum_{\hat{m}_{j}^{t_{j}} \neq m_{j}^{t_{j}} }\sum_{\substack{\hat{a}_{j}^{s_{j}}}}\sum_{\substack{(v_{1}^{n},\underline{u}^{n})\\\in T_{2\eta_{2}}(V_{1},\underline{U}|q^{n})}}\frac{P\left( \substack{V_{1}(m_{1},b_{1})=v_{1}^{n},M_{1}=m_{1}\\M_{l}^{t_{l}}=m_{l}^{t_{l}}:l=2,3}\right)\exp \left\{ n32\eta_{2} \right\}}{\pi^{3n+t_{2}+t_{3}+t_{j}}\exp \left\{ -nH(U_{j}|Q,Y_{j}) \right\}}\frac{1}{\mathcal{L}(n)}\nonumber\\
\label{Eqn:3DBCProofFinalInequalityInDecoder2ErrorEvent}
\lefteqn{~~~~~~\leq  2\exp \left\{ s_{j}\log\pi -n\left( \log\pi-H(U_{j}|Q,Y_{j})-32\eta_{2} \right)\right\} \leq 2\exp \left\{ -n(\delta_{1}-32\eta_{2})\right\},}
\end{eqnarray}
where (\ref{Eqn:3DBCProofFinalInequalityInDecoder2ErrorEvent}) follows from definition of $\mathcal{L}(n)$, (\ref{Eqn:3To1DBCOneLayerEncoderErrorEventVarianceOfPhi}) and the bounds on the parameters of the code derived in (\ref{Eqn:3To1DBCOneLayerAchievabilityLowerAndUpperBoundsOnSj}) - (\ref{Eqn:3To1DBCOneLayerAchievabilityLowerAndUpperBoundsOnR1AndK1}).

\section{Characterization for no rate loss in point-to-point channels with channel state information}
\label{AppSec:CharacterizationForNoRateLossInPTP-STx}

We now develop the connection between upper bound
(\ref{Eqn:UpperBoundOnRateOfUser1ThatContainsRateLoss}) and the
capacity of a PTP channel with non-causal state
\cite{1980MMPCT_GelPin}. We only describe the relevant additive
channel herein and refer the interested reader to either to
\cite{1980MMPCT_GelPin} or \cite[Chapter 7]{201201NIT_ElgKim} for a
detailed study. The notation employed in this section and appendix
\ref{AppSec:TheBinaryAdditiveDirtyPointToPointChannelSuffersARateLoss}
are specific to these sections. 

Consider the  discrete memoryless PTP channel with binary input and output alphabets
$\mathcal{X}=\mathcal{Y}=\left\{0,1\right\}$. The channel transition
probabilities depend on a random parameter, called state that takes
values in the binary alphabet $\mathcal{S}=\left\{ 0,1 \right\}$. 
The channel is additive, i.e., if $S,X$ and $Y$ denote
channel state, input and output respectively, then $P(Y=x \oplus s |
X=x, S=s)=1-\delta$, where $\oplus$ denotes addition in binary field
and $\delta \in (0,\frac{1}{2})$. The state is independent and
identically distributed across time with $P(S=1)=\epsilon \in
(0,1)$.\footnote{Through appendices
  \ref{AppSec:CharacterizationForNoRateLossInPTP-STx},\ref{AppSec:TheBinaryAdditiveDirtyPointToPointChannelSuffersARateLoss}
  we prove if $\delta,\tau \in (0,\frac{1}{2})$ and $\epsilon \in
  (0,1)$, then $\alpha_{T}(\tau,\eta,\epsilon) < h_b(\tau *
  \eta)-h_b(\eta)$. This implies statement of lemma
  \ref{Lem:CharacterizationOfConditionThatEnsuresNoRateLoss}.} The
input is constrained by an additive Hamming cost, i.e., the cost of
transmitting $x^{n} \in \InputAlphabet^{n}$ is
$\sum_{t=1}^{n}1_{\left\{x_{t}=1\right\}}$ and average cost of input
per symbol is constrained to be $\tau \in (0,\frac{1}{2})$.  

The quantities of interest - left and right hand sides of (\ref{Eqn:PluggingInImportOfRateLossLemma})(i) - are related to two scenarios with regard to knowledge of state for the above channel. In the first scenario we assume the state sequence is available to the encoder non-causally and the decoder has no knowledge of the same. In the second scenario, we assume knowledge of state is available to both the encoder and decoder non-causally. Let $\mathcal{C}_{T}(\tau,\delta,\epsilon),\mathcal{C}_{TR}(\tau,\delta,\epsilon)$ denote the capacity of the channel in the first and second scenarios respectively. It turns out, the left hand side of (\ref{Eqn:PluggingInImportOfRateLossLemma})(i) is upper bounded by $\mathcal{C}(\tau,\delta,\epsilon)$ and the right hand side of (\ref{Eqn:PluggingInImportOfRateLossLemma})(i) is $\mathcal{C}_{TR}(\tau,\delta,\epsilon)$. A necessary condition for (\ref{Eqn:PluggingInImportOfRateLossLemma})(i) to hold, is therefore $\mathcal{C}_{T}(\tau,\delta,\epsilon)=\mathcal{C}_{TR}(\tau,\delta,\epsilon)$. For the PTP channel with non-causal state, this equality is popularly referred to as \textit{no rate loss}. We therefore seek the condition for no rate loss.

The objective of this section and appendix \ref{AppSec:TheBinaryAdditiveDirtyPointToPointChannelSuffersARateLoss} is to study the condition under which $\mathcal{C}_{T}(\tau,\delta,\epsilon)=\mathcal{C}_{TR}(\tau,\delta,\epsilon)$. In this section, we characterize each of these quantities, in the standard information theoretic way, in terms of a maximization of an objective function over a particular collection of probability mass functions.

We begin with a characterization of $\mathcal{C}_{T}(\tau,\delta,\epsilon)$ and $\mathcal{C}_{TR}(\tau,\delta,\epsilon)$.

\begin{definition}
 \label{Defn:TestChannelsForPTP-STx}
Let $\SetOfDistributions_{T}(\tau,\delta,\epsilon)$ denote the set of all probability mass functions $p_{USXY}$ defined on $\AuxiliaryAlphabet \times \StateAlphabet \times \InputAlphabet \times \OutputAlphabet$ that satisfy (i) $p_{S}(1) = \epsilon$, (ii) $p_{Y|XSU}(x\oplus s|x,s,u)=p_{Y|XS}(x\oplus s | x,s)=1-\delta$, (iii) $P(X=1)\leq \tau$. For $p_{USXY} \in \SetOfDistributions_{T}(\tau,\delta,\epsilon)$, let $\alpha_{T}(p_{USXY}) = I(U;Y)-I(U;S)$ and $\alpha_{T}(\tau,\delta,\epsilon) = \underset{p_{USXY} \in \SetOfDistributions_{T}(\tau,\delta,\epsilon)}{\sup} \alpha_{T}(p_{USXY})$.
\end{definition}

\begin{thm}
 \label{Thm:CapacityOfBinaryAdditivePTP-STx}
$\mathcal{C}_{T}(\tau,\delta,\epsilon) = \alpha_{T}(\tau,\delta,\epsilon)$
\end{thm}
This is a well known result in information theory and we refer the reader to \cite{1980MMPCT_GelPin} or \cite[Section 7.6, Theorem 7.3]{201201NIT_ElgKim} for a proof.
\begin{definition}
 \label{Defn:TestChannelsForPTP-STxRx}
Let $\SetOfDistributions_{TR}(\tau,\delta,\epsilon)$ denote the set of all probability mass functions $p_{SXY}$ defined on $\StateAlphabet \times \InputAlphabet \times \OutputAlphabet$ that satisfy (i) $p_{S}(1) = \epsilon$, (ii) $p_{Y|XS}(x\oplus s|x,s)=1-\delta$, (iii) $P(X=1)\leq \tau$. For $p_{SXY} \in \SetOfDistributions_{TR}(\tau,\delta,\epsilon)$, let $\alpha_{TR}(p_{SXY}) = I(X;Y|S)$ and $\alpha_{TR}(\tau,\delta,\epsilon) = \underset{p_{SXY} \in \SetOfDistributions_{TR}(\tau,\delta,\epsilon)}{\sup} \alpha_{TR}(p_{SXY})$.
\end{definition}

\begin{thm}
 \label{Thm:CapacityOfBinaryAdditivePTP-STxRx}
$\mathcal{C}_{TR}(\tau,\delta,\epsilon) = \alpha_{TR}(\tau,\delta,\epsilon)$
\end{thm}
This can be argued using Shannon's characterization of PTP channel capacity \cite{194807BSTJ_Sha} and we refer the reader to \cite[Section 7.4.1]{201201NIT_ElgKim} for a proof.
\begin{remark}
\label{Rem:SideInformationAtBothEncoderAndDecoderResultsInLargerCapacity}
 From the definition of $\mathcal{C}_{T}(\tau,\delta,\epsilon)$ and $\mathcal{C}_{TR}(\tau,\delta,\epsilon)$, it is obvious that $\mathcal{C}_{T}(\tau,\delta,\epsilon) \leq \mathcal{C}_{TR}(\tau,\delta,\epsilon)$, we provide an alternative argument based on theorems \ref{Thm:CapacityOfBinaryAdditivePTP-STx}, \ref{Thm:CapacityOfBinaryAdditivePTP-STxRx}. For any $p_{USXY} \in \SetOfDistributions_{T}(\tau,\delta,\epsilon)$, it is easy to verify the corresponding marginal $p_{SXY} \in \SetOfDistributions_{TR}(\tau,\delta,\epsilon)$ and moreover $\alpha_{T}(p_{USXY})=I(U;Y)-I(U;S) \leq I(U;YS)-I(U;S) = I(U;Y|S)=H(Y|S)-H(Y|US) \leq H(Y|S)-H(Y|USX) \overset{(a)}{=} H(Y|S)-H(Y|SX)=I(X;Y|S)=\alpha_{TR}(p_{SXY})\leq \mathcal{C}_{TR}(\tau,\delta,\epsilon)$, where (a) follows from Markov chain $U-(S,X)-Y$ ((ii) of definition \ref{Defn:TestChannelsForPTP-STx}). Since this this true for every $p_{USXY} \in \SetOfDistributions_{T}(\tau,\delta,\epsilon)$, we have $\mathcal{C}_{T}(\tau,\delta,\epsilon) \leq \mathcal{C}_{TR}(\tau,\delta,\epsilon)$.
\end{remark}

We provide an alternate characterization for $\mathcal{C}_{TR}(\tau,\delta,\epsilon)$.
\begin{lemma}
 \label{Lem:AlternateCharacterizationForCapacityOfBinaryAdditivePTP-STxRx}
For $p_{USXY} \in \SetOfDistributions_{T}(\tau,\delta,\epsilon)$, let $\beta_{TR}(p_{USXY}) = I(U;Y|S)$ and $\beta_{TR}(\tau,\delta,\epsilon) = \underset{p_{USXY} \in \SetOfDistributions_{T}(\tau,\delta,\epsilon)}{\sup} \beta_{TR}(p_{USXY})$. Then $\beta_{TR}(\tau,\delta,\epsilon) = \alpha_{TR}(\tau,\delta,\epsilon) = \mathcal{C}_{TR}(\tau,\delta,\epsilon)$.
\end{lemma}
\begin{proof}
 We first prove $\beta_{TR}(\tau,\delta,\epsilon) \leq \alpha_{TR}(\tau,\delta,\epsilon)$. Note that for any $p_{USXY} \in \SetOfDistributions_{T}(\tau,\delta,\epsilon)$, the corresponding marginal $p_{SXY} \in \SetOfDistributions_{TR}(\tau,\delta,\epsilon)$. Moreover, $\beta_{TR}(p_{USXY})=I(U;Y|S) = H(Y|S)-H(Y|US) \leq H(Y|S)-H(Y|USX) \overset{(a)}{=} H(Y|S)-H(Y|SX)=I(X;Y|S)=\alpha_{TR}(p_{SXY})$, where (a) follows from Markov chain $U-(S,X)-Y$ ((ii) of definition \ref{Defn:TestChannelsForPTP-STx}). Therefore, $\beta_{TR}(\tau,\delta,\epsilon) \leq \alpha_{TR}(\tau,\delta,\epsilon)$. Conversely, given $p_{SXY} \in \SetOfDistributions_{TR}(\tau,\delta,\epsilon)$, define $\AuxiliaryAlphabet=\left\{  0,1\right\}$ and a probability mass function $q_{USXY}$ defined on $\AuxiliaryAlphabet \times \StateAlphabet \times \InputAlphabet \times \OutputAlphabet$ as $q_{USXY}(u,s,x,y)=p_{SXY}(s,x,y)1_{\left\{ u=x \right\}}$. Clearly $q_{SXY}=p_{SXY}$ and hence (i) and (iii) of definition \ref{Defn:TestChannelsForPTP-STx} are satisfied. Note that $q_{USX}(x,s,x)=p_{SX}(s,x)$, and hence $q_{Y|XSU}(y|x,s,x)=p_{Y|XS}(y|x,s)=W_{Y|XS}(y|x,s)$. Hence $q_{USXY} \in \SetOfDistributions_{TR}(\tau,\delta,\epsilon)$. It is easy to verify $\beta_{TR}(q_{USXY}) = \alpha_{TR}(p_{SXY})$ and therefore $\beta_{TR}(\tau,\delta,\epsilon) \geq \alpha_{TR}(\tau,\delta,\epsilon)$.
\end{proof}
We now derive a characterization of the condition under which $\mathcal{C}_{TR}(\tau,\delta,\epsilon)=\mathcal{C}_{T}(\tau,\delta,\epsilon)$. Towards that end, we first prove uniqueness of the PMF that achieves $\mathcal{C}_{TR}(\tau,\delta,\epsilon)$.
\begin{lemma}
 \label{Eqn:UniquenessOfPMFThatAchievesCapacityOfPTP-STxRx}
Suppose $p_{SXY},q_{SXY} \in \SetOfDistributions_{TR}(\tau,\delta,\epsilon)$ are such that $\alpha_{TR}(p_{SXY})=\alpha_{TR}(q_{SXY})=\mathcal{C}_{TR}(\tau,\delta,\epsilon)$, then $p_{SXY}=q_{SXY}$. Moreover, if $\alpha_{TR}(p_{SXY})=\mathcal{C}_{TR}(\tau,\delta,\epsilon)$, then $p_{SX}=p_{S}p_{X}$, i.e., $S$ and $X$ are independent.
\end{lemma}
\begin{proof}
Clearly, if $q_{SXY} \in \SetOfDistributions_{TR}(\tau,\delta,\epsilon)$ satisfies $q_{SX}=q_{S}q_{X}$ with $q_{X}(1)=\tau$, then $\alpha_{TR}(q_{SXY})=h_{b}(\tau * \delta)-h_{b}(\delta)$ and since $\mathcal{C}_{TR}(\tau,\delta,\epsilon) \leq h_{b}(\tau * \delta)-h_{b}(\delta)$,\footnote{This can be easily verified using standard information theoretic arguments.} we have $\mathcal{C}_{TR}(\tau,\delta,\epsilon) = h_{b}(\tau * \delta)-h_{b}(\delta)$. Let $p_{SXY} \in \SetOfDistributions_{TR}(\tau,\delta,\epsilon)$ be another PMF for which $\alpha_{TR}(p_{SXY})=h_{b}(\tau * \delta)-h_{b}(\delta)$. Let $\chi_{0}\define p_{X|S}(1|0)$ and $\chi_{1}\define p_{X|S}(1|1)$. $\alpha_{TR}(p_{SXY})=I(X;Y|S)=H(Y|S)-H(Y|X,S)=H(X\oplus S \oplus N|S)-h_{b}(\delta)$. We focus on the first term
\begin{eqnarray}
\lefteqn{H(X\oplus S \oplus N|S) = (1-\epsilon)H(X\oplus 0 \oplus N|S=0)+\epsilon H(X\oplus 1 \oplus N|S=1)}\nonumber\\
&=&(1-\epsilon)h_{b}(\chi_{0}(1-\delta)+(1-\chi_{0})\delta)+\epsilon h_{b}(\chi_{1} (1-\delta)+(1-\chi_{1})\delta)
 \nonumber\\
\label{Eqn:UniquenessOfPMFThatAchievesCapacityOfPTPWithStateFollowsFromConcavity}
&\leq& h_{b}((1-\epsilon)\chi_{0}(1-\delta)+(1-\epsilon)(1-\chi_{0})\delta+\epsilon\chi_{1} (1-\delta)+\epsilon(1-\chi_{1})\delta)\\
\label{Eqn:UniquenessOfPMFThatAchievesCapacityOfPTPWithStateFollowsFromRangeOfEta}
&=&h_{b}(p_{X}(1)(1-\delta)+(1-p_{X}(1))\delta)=h_{b}(\delta+p_{X}(1)(1-2\delta)) \leq h_{b}(\delta+\tau(1-2\delta))=h_{b}(\tau * \delta)
\end{eqnarray}
where (\ref{Eqn:UniquenessOfPMFThatAchievesCapacityOfPTPWithStateFollowsFromConcavity}) follows from concavity of binary entropy function $h_{b}(\cdot)$ and inequality in (\ref{Eqn:UniquenessOfPMFThatAchievesCapacityOfPTPWithStateFollowsFromRangeOfEta}) follows from $\delta \in (0,\frac{1}{2})$. We therefore have $\alpha_{TR}(p_{SXY})=h_{b}(\tau * \delta)-h_{b}(\delta)$ if and only if equality holds in (\ref{Eqn:UniquenessOfPMFThatAchievesCapacityOfPTPWithStateFollowsFromConcavity}), (\ref{Eqn:UniquenessOfPMFThatAchievesCapacityOfPTPWithStateFollowsFromRangeOfEta}). $h_{b}(\cdot)$ being strictly concave, equality holds in (\ref{Eqn:UniquenessOfPMFThatAchievesCapacityOfPTPWithStateFollowsFromConcavity}) if and only if $\epsilon \in \left\{ 0,1\right\}$ or $\chi_{0}=\chi_{1}$. The range of $\epsilon$ precludes the former and therefore $\chi_{0}=\chi_{1}$. This proves $p_{SX}=p_{S}p_{X}$ and $p_{X}(1)=\tau$. Given $p_{SXY} \in \SetOfDistributions_{TR}(\tau,\delta,\epsilon)$, these constrains completely determine $p_{SXY}$ and we have $p_{SXY}=q_{SXY}$.
\end{proof}
Following is the main result of this section.
\begin{lemma}
 \label{Lem:CharacterizationOfConditionThatEnsuresNoRateLoss}
$\mathcal{C}_{TR}(\tau,\delta,\epsilon)=\mathcal{C}_{T}(\tau,\delta,\epsilon)$ if and only if there exists a PMF $p_{USXY} \in \SetOfDistributions_{T}(\tau,\delta,\epsilon)$ such that
\begin{enumerate}
 \item the corresponding marginal achieves $\mathcal{C}_{TR}(\tau,\delta,\epsilon)$, i.e., $\alpha_{TR}(p_{SXY})=\mathcal{C}_{TR}(\tau,\delta,\epsilon)$,
\item $S-Y-U$ is a Markov chain.
\item $X-(U,S)-Y$ is a Markov chain.
\end{enumerate}
\end{lemma}
\begin{proof}
We first prove the reverse implication, i.e., the if statement. Note that $\mathcal{C}_{TR}(\tau,\delta,\epsilon)=\alpha_{TR}(p_{SXY})=I(X;Y|S)=H(Y|S)-H(Y|XS)\overset{(a)}{=}H(Y|S)-H(Y|XSU)\overset{(b)}{=}H(Y|S)-H(Y|US)=I(U;Y|S)=I(U;YS)-I(U;S)\overset{(c)}{=}I(U;Y)-I(U;S) \leq \mathcal{C}_{T}(\tau,\delta,\epsilon)$, where (a) follows from (ii) of definition \ref{Defn:TestChannelsForPTP-STx}, (b) follows from hypothesis 3) and (c) follows from hypothesis 2). We therefore have $\mathcal{C}_{TR}(\tau,\delta,\epsilon)\leq \mathcal{C}_{T}(\tau,\delta,\epsilon)$, and the reverse inequality follows from remark \ref{Rem:SideInformationAtBothEncoderAndDecoderResultsInLargerCapacity}.

Conversely, let $p_{USXY} \in \SetOfDistributions_{T}(\tau,\delta,\epsilon)$ achieve $\mathcal{C}_{T}(\tau,\delta,\epsilon)$, i.e., $\alpha_{T}(p_{USXY})=\mathcal{C}_{T}(\tau,\delta,\epsilon)$. We have $\mathcal{C}_{T}(\tau,\delta,\epsilon)=\alpha_{T}(p_{USXY}) = I(U;Y)-I(U;S) \overset{(b)}{\leq} I(U;YS)-I(U;S) = I(U;Y|S)=H(Y|S)-H(Y|US) \overset{(c)}{\leq} H(Y|S)-H(Y|USX) \overset{(a)}{=} H(Y|S)-H(Y|SX)=I(X;Y|S)=\alpha_{TR}(p_{SXY})\leq \mathcal{C}_{TR}(\tau,\delta,\epsilon)$, where (a) follows from Markov chain $U-(S,X)-Y$ ((ii) of definition \ref{Defn:TestChannelsForPTP-STx}). Equality of $\mathcal{C}_{TR}(\tau,\delta,\epsilon),\mathcal{C}_{T}(\tau,\delta,\epsilon)$ implies equality in (b), (c) and thus $I(U;S|Y)=0$ and $H(Y|US)=H(Y|USX)$ and moreover $\alpha_{TR}(p_{SXY})=\mathcal{C}_{TR}(\tau,\delta,\epsilon)$.
\end{proof}

For the particular binary additive PTP channel with state, we strengthen the condition for no rate loss in the following lemma.
\begin{lemma}
 \label{Lem:XShouldBeAFunctionOfUAndS}
If $p_{USXY} \in \SetOfDistributions_{T}(\tau,\delta,\epsilon)$ satisfies
(i) $S-Y-U$ is a Markov chain, and (ii) $X-(U,S)-Y$ is a Markov chain, 
then $H(X|U,S)=0$, or in other words, there exists a function $f : \AuxiliaryAlphabet \times \StateAlphabet \rightarrow \InputAlphabet$ such that $P(X=f(U,S))=1$.
\end{lemma}
\begin{proof}
We prove this by contradiction. In particular, we prove $H(X|U,S) > 0$ violates Markov chain $X-(U,S)-Y$.
If $H(X|U,S) > 0$, then $H(X \oplus S |U,S) >0$. Indeed, $0 < H(X|U,S)\leq H(X,S|U,S)=H(X\oplus S, S |U,S) = H(S|U,S)+H(X\oplus S | U,S)=H(X \oplus S | U, S)$. Since $(U,S,X)$ is independent of $X\oplus S \oplus Y$ and in particular, $(U,S,S\oplus X)$ is independent of $X\oplus S \oplus Y$, we have $H((X\oplus S)\oplus(X\oplus S \oplus Y)|U,S)> H(X\oplus S \oplus Y|U,S)=h_{b}(\delta)=H(Y|U,S,X)$, where the first inequality follows from concavity of binary entropy function. But $(X\oplus S)\oplus(X\oplus S \oplus Y)=Y$ and we have therefore proved $H(Y|U,S)>H(Y|U,S,X)$ contradicting Markov chain $X-(U,S)-Y$.
\end{proof}
We summarize the conditions for no rate loss below.
\begin{thm}
 \label{Thm:CharacterizationOfConditionThatEnsuresNoRateLossForBinaryAdditive}
$\mathcal{C}_{TR}(\tau,\delta,\epsilon)=\mathcal{C}_{T}(\tau,\delta,\epsilon)$ if and only if there exists a PMF $p_{USXY} \in \SetOfDistributions_{T}(\tau,\delta,\epsilon)$ such that
\begin{enumerate}
 \item the corresponding marginal achieves $\mathcal{C}_{TR}(\tau,\delta,\epsilon)$, i.e., $\alpha_{TR}(p_{SXY})=\mathcal{C}_{TR}(\tau,\delta,\epsilon)$, and in particular $S$ and $X$ are independent,
\item $S-Y-U$ is a Markov chain.
\item $X-(U,S)-Y$ is a Markov chain,
\item $H(X|U,S)=0$, or in other words, there exists a function $f : \AuxiliaryAlphabet \times \StateAlphabet \rightarrow \InputAlphabet$ such that $P(X=f(U,S))=1$.
\end{enumerate}
\end{thm}
\section{The binary additive dirty point-to-point channel suffers a rate loss}
\label{AppSec:TheBinaryAdditiveDirtyPointToPointChannelSuffersARateLoss}
This section is dedicated to proving proposition \ref{Prop:PropositionRegardindRateLoss}. We begin with an upper bound on cardinality of auxiliary set involved in characterization of $\mathcal{C}_{T}(\tau,\delta,\epsilon)$.
\begin{lemma}
 \label{Lem:CardinalityBoundOnCostConstrainedGelfandPinskerAuxiliaryRV}
Consider a PTP channel with state information available at transmitter. Let $\mathcal{S}, \mathcal{X}$ and $\mathcal{Y}$ denote state, input and output alphabets respectively. Let $W_{S},W_{Y|XS}$ denote PMF of state, channel transition probabilities respectively. The input is constrained with respect to a cost function $\kappa : \mathcal{X} \times \mathcal{S} \rightarrow [0,\infty)$. Let $\mathbb{D}_{T}(\tau)$ denote the collection of all probability mass functions $p_{UXSY}$ defined on $\mathcal{U} \times \mathcal{X} \times \mathcal{S} \times \mathcal{Y}$, where $\mathcal{U}$ is an arbitrary set, such that (i) $p_{S}=W_{S}$, (ii) $p_{Y|XSU}=p_{Y|XS}=W_{Y|XS}$ and (iii) $\Expectation \left\{ \kappa(X,S) \right\} \leq \tau$. Moreover, let \[\overline{\mathbb{D}_{T}}(\tau) = \left\{ p_{UXSY} \in \mathbb{D}_{T}(\tau): |\mathcal{U}| \leq \min \left\{ \substack{|\mathcal{X}|\cdot |\mathcal{S}|, \\|\mathcal{X}|+|\mathcal{S}|+|\mathcal{Y}|-2 }\right\} \right\}.\] For $p_{UXSY} \in \mathbb{D}_{T}(\tau)$, let $\alpha(p_{UXSY})=I(U;Y)-I(U;S)$. Let 
\begin{eqnarray} 
 \alpha_{T}(\tau)=\sup_{p_{UXSY} \in \mathbb{D}_{T}(\tau)}\alpha(p_{UXSY}) ,~~~~ \overline{\alpha_{T}}(\tau)=\sup_{p_{UXSY} \in \overline{\mathbb{D}_{T}}(\tau)}\alpha(p_{UXSY}). \nonumber
\end{eqnarray}
Then $\alpha_{T}(\tau)=\overline{\alpha_{T}}(\tau)$.
\end{lemma}
\begin{proof}
 The proof is based on Fenchel-Eggelston-Carath\'eodory \cite{Egg-Convexity1958}, \cite[Appendix C]{201201NIT_ElgKim} theorem which is stated here for ease of reference.
\begin{lemma}
\label{Lem:FenchelEgglestonCaratheodory}
 let $\mathcal{A}$ be a finite set and $\mathcal{Q}$ be an arbitrary set. Let $\mathcal{P}$ be a connected compact subset of PMF's on $\mathcal{A}$ and $p_{A|Q}(\cdot | q) \in \mathcal{P}$ for each $q \in \mathcal{Q}$. For $j=1,2,\cdots,d$ let $g_{j}: \mathcal{P} \rightarrow \reals$ be continuous functions. Then for every $Q \sim F_{Q}$ defined on $\mathcal{Q}$, there exist a random variable $\overline{Q} \sim p_{\overline{Q}}$ with $|\overline{\mathcal{Q}}| \leq d$ and a collection of PMF's $p_{A|\overline{Q}}(\cdot|\overline{q}) \in \mathcal{P}$, one for each $\overline{q} \in \overline{\mathcal{Q}}$, such that
\begin{eqnarray}
 \int_{\mathcal{Q}}g_{j}(p_{A|Q}(a|q))dF_{Q}(q) = \sum_{\overline{q} \in \overline{\mathcal{Q}}}g_{j}(p_{A|\overline{Q}}(a|\overline{q}))p_{\overline{Q}}(\overline{q}).\nonumber
\end{eqnarray}
\end{lemma}
The proof involves identifying $g_{j}:j=1,2\cdots,d$ such that rate achievable and cost expended are preserved. We first prove the bound $|\mathcal{U}| \leq |\mathcal{X}|\cdot|\mathcal{S}|$.

Set $\mathcal{Q}=\mathcal{U}$ and $\mathcal{A}=\mathcal{X} \times \mathcal{S}$ and $\mathcal{P}$ denote the connected compact subset of PMF's on $\mathcal{X} \times \mathcal{S}$. Without loss of generality, let $\mathcal{X} = \left\{1,2,\cdots, |\mathcal{X}|  \right\}$ and $\mathcal{S} = \left\{1,2,\cdots, |\mathcal{S}|  \right\}$. For $i=1,2, \cdots, |\mathcal{X}|$ and $k=1,2,\cdots,|\mathcal{S}|-1$, let $g_{i,k}(\pi_{X,S})=\pi_{X,S}(i,k)$ and $g_{l,|\mathcal{S}|}(\pi_{X,S})=\pi_{X,S}(l,|\mathcal{S}|)$ for $l=1,2,\cdots, |\mathcal{X}|-1$. Let $g_{|\mathcal{X}|\cdot |\mathcal{S}|}(\pi_{X,S})=H(S)-H(Y)$. It can be verified that
\begin{flalign}
 g_{|\mathcal{X}|\cdot |\mathcal{S}|}(\pi_{X,S}) =& -\sum_{s \in \mathcal{S}} (\sum_{x \in \mathcal{X}}\pi_{X,S}(x,s))\log_{2}(\sum_{x \in \mathcal{X}}\pi_{X,S}(x,s))+ \sum_{y \in \mathcal{Y}} \theta(y)\log_{2}(\theta(y)),\mbox{ where }\nonumber\\
\label{Eqn:ThetaOfy}
\!\!\!\!\!\!\!\theta(y)=&\sum_{(x,s) \in \mathcal{X}\times \mathcal{S}}\pi_{X,S}(x,s)W_{Y|XS}(y|x,s)
\end{flalign}
where, is continuous. An application of lemma \ref{Lem:FenchelEgglestonCaratheodory} using the above set of functions, the upper bound $|\mathcal{X}|\cdot|\mathcal{S}|$ on $|\mathcal{U}|$ can be verified.

We now outline proof of upper bound $|\mathcal{X}|+|\mathcal{S}|+|\mathcal{Y}|-2$ on $|\mathcal{U}|$. Without loss of generality, we assume $\mathcal{X}=\left\{ 1,\cdots, |\mathcal{X}| \right\}$, $\mathcal{S}=\left\{ 1,\cdots, |\mathcal{S}| \right\}$ and $\mathcal{Y}=\left\{ 1,\cdots, |\mathcal{Y}| \right\}$. As earlier, set $\mathcal{Q}=\mathcal{U}$ and $\mathcal{A}=\mathcal{X} \times \mathcal{S}$ and $\mathcal{P}$ denote the connected compact subset of PMF's on $\mathcal{X} \times \mathcal{S}$. For $j=1,\cdots,|\mathcal{S}|-1$, let $g_{j}(\pi_{X,S})=\sum_{x \in \mathcal{X}} \pi_{X,S}(x,j)$. For $j=|\mathcal{S}|, \cdots, |\mathcal{S}|+|\mathcal{Y}|-2$, let $g_{j}(\pi_{X,S})=\sum_{(x,s) \in \mathcal{X} \times \mathcal{S}}\pi_{X,S}(x,s)W_{Y|X,S}(j-|\mathcal{S}|+1|x,s)$. For $j=|\mathcal{S}|+|\mathcal{Y}|-1, \cdots, |\mathcal{S}|+|\mathcal{Y}|+|\mathcal{X}|-3$, let $g_{j}(\pi_{X,S})=\sum_{s \in \mathcal{S}}\pi_{X,S}(j-|\mathcal{S}|-|\mathcal{Y}|+2,s)$. Let $g_{t}(\pi_{X,S})=H(S)-H(Y)$, i.e.,
\begin{eqnarray}
g_{t}(\pi_{X,S}) =  -\sum_{s \in \mathcal{S}} (\sum_{x \in \mathcal{X}}\pi_{X,S}(x,s))\log_{2}(\sum_{x \in \mathcal{X}}\pi_{X,S}(x,s))+ \sum_{y \in \mathcal{Y}} \theta(y)\log_{2}(\theta(y)),\nonumber
\end{eqnarray}
where $t=|\mathcal{S}|+|\mathcal{Y}|+|\mathcal{X}|-2$, and $\theta(y)$ as is in (\ref{Eqn:ThetaOfy}). The rest of the proof follows by simple verification.
\end{proof}

\begin{prop}
 \label{Prop:PropositionRegardindRateLoss}
There exists no probability mass function $p_{UXSY}$ defined on $\AuxiliaryAlphabet \times \StateAlphabet \times \InputAlphabet \times \OutputAlphabet$ where $\AuxiliaryAlphabet = \left\{  0,1,2,3 \right\}, \InputAlphabet = \StateAlphabet = \OutputAlphabet = \left\{ 0,1 \right\}$, such that
\begin{enumerate}
 \item $X$ and $S$ are independent with $P(S=1)=\epsilon$, $P(X=1)=\tau$, where $\epsilon \in (0,1)$, $\tau \in (0,\frac{1}{2})$,
\item $p_{Y|X,S,U}(x \oplus s|x,s,u)=p_{Y|X,S}(x\oplus s |x,s)=1-\delta$ for every $(u,x,s,y) \in \AuxiliaryAlphabet \times \StateAlphabet \times \InputAlphabet \times \OutputAlphabet$, where $\delta \in (0,\frac{1}{2})$,
\item $U-Y-S$ and $X-(U,S)-Y$ are Markov chains, and
\item $p_{X|US}(x|u,s) \in \left\{ 0,1 \right\}$ for each $(u,s,x) \in \AuxiliaryAlphabet \times \StateAlphabet \times \InputAlphabet$.
\end{enumerate}
\end{prop}
\begin{proof}
 The proof is by contradiction. If there exists such a PMF $p_{USXY}$ then conditions 1) and 2) completely specify it's marginal on $\StateAlphabet \times \InputAlphabet \times \OutputAlphabet$ and it maybe verified that $p_{SY}(0,0)=(1-\epsilon)(1-\theta), p_{SY}(0,1)=(1-\epsilon)\theta, p_{SY}(1,0)=\epsilon\theta,p_{SY}(1,1)=\epsilon(1-\theta)$, where $\theta \define \delta(1-\tau)+(1-\delta)\tau$ takes a value in $(0,1)$. Since $\epsilon \in (0,1)$, $p_{SY}(s,y) \in (0,1)$ for each $(s,y) \in \StateAlphabet \times \OutputAlphabet$.
If we let $\beta_{i}\define p_{U|Y}(i|0) :i=0,1,2,3$ and $\gamma_{j}\define p_{U|Y}(j|1):j=0,1,2,3$, then Markov chain $U-Y-S$ implies $p_{USY}$ is as in table \ref{Table:p_USY}.
\begin{table} \begin{center}
\begin{tabular}{|c|c|c|c|} \hline
USY & $p_{USY}$&USY & $p_{USY}$\\
\hline
000 & $(1-\epsilon)(1-\theta)\beta_{0}$&200 & $(1-\epsilon)(1-\theta)\beta_{2}$ \\
\hline
001 & $(1-\epsilon)\theta\gamma_{0}$&201 & $(1-\epsilon)\theta \gamma_{2}$ \\
\hline
010 & $\epsilon\theta\beta_{0}$&210 & $\epsilon\theta\beta_{2}$ \\
\hline
011 & $\epsilon(1-\theta) \gamma_{0}$ &211 & $\epsilon(1-\theta)\gamma_{2}$ \\
\hline
100 & $(1-\epsilon)(1-\theta)\beta_{1}$ &300 & $(1-\epsilon)(1-\theta)\beta_{3}$ \\
\hline
101 & $(1-\epsilon)\theta \gamma_{1}$ &301 & $(1-\epsilon)\theta \gamma_{3}$ \\
\hline
110 & $\epsilon\theta\beta_{1}$ &310 & $\epsilon\theta\beta_{3}$ \\
\hline
111 & $\epsilon(1-\theta) \gamma_{1}$ &311 & $\epsilon(1-\theta)\gamma_{3}$ \\
\hline
\end{tabular} \end{center}\caption{$p_{USY}$} \label{Table:p_USY} 
\end{table}  
Since $X$ is a function of $(U,S)$\footnote{With probability $1$}, there exist $z_{i} \in \left\{ 0,1 \right\}:i=0,1,\cdots,7$ such that entries of table \ref{Table:p_USX} hold true.
\begin{table} \begin{center}
\begin{tabular}{|c|c|} \hline
$p_{USX}(0,0,0) =p_{US}(0,0)z_{0}$&$p_{USX}(0,1,0) =p_{US}(0,1)z_{4}$\\
\hline
$p_{USX}(1,0,0) =p_{US}(1,0)z_{1}$&$p_{USX}(1,1,0) =p_{US}(1,1)z_{5}$\\
\hline
$p_{USX}(2,0,0) =p_{US}(2,0)z_{2}$&$p_{USX}(2,1,0) =p_{US}(2,1)z_{6}$\\
\hline
$p_{USX}(3,0,0) =p_{US}(3,0)z_{3}$&$p_{USX}(3,1,0) =p_{US}(3,1)z_{7}$\\
\hline
\end{tabular} \end{center}\caption{$p_{USX}$} \label{Table:p_USX} 
\end{table}  
Moreover, condition 4) and Markov chain $X-(U,S)-Y$ implies $p_{USXY}$ is completely determined in terms of entries of table \ref{Table:p_USY} and $z_{i}:i=0,1,\cdots,7$. For example $p_{USXY}(3,0,1,1)=p_{USY}(3,0,1)(1-z_{3})$. This enables us to compute the marginal $p_{SXY}$ in terms of entries of table \ref{Table:p_USY} and $z_{i}:i=0,1,\cdots,7$. This marginal must satisfy conditions 1) and 2) which implies that the last two columns of table \ref{Table:P_SXY} are equal.
\begin{eqnarray}
\label{Eqn:P_SYXOf000}
p_{SYX}(0,0,0)&=&(1-\epsilon)(1-\theta)\left[\beta_0z_0+\beta_1 z_1+\beta_2  z_2+\beta_{3}z_{3} \right] =(1-\tau)(1-\epsilon)(1-\delta)\\
\label{Eqn:P_SYXOf001}
p_{SYX}(0,0,1)&=& (1-\epsilon)(1-\theta)\left[1-\beta_0z_0-\beta_1 z_1-\beta_2  z_2-\beta_{3}z_{3} \right] = \tau(1-\epsilon)\delta\nonumber\\
\label{Eqn:P_SYXOf010}
p_{SYX}(0,1,0)&=& (1-\epsilon)\theta \left[\gamma_0z_0+\gamma_1 z_1+\gamma_2  z_2 +\gamma_{3} z_{3} \right]= (1-\tau)(1-\epsilon)\delta\\
\label{Eqn:P_SYXOf011}
p_{SYX}(0,1,1)&=& (1-\epsilon)\theta\left[1-\gamma_0z_0-\gamma_1 z_1-\gamma_2  z_2-\gamma_{3} z_{3} \right]
= \tau(1-\epsilon)(1-\delta)\nonumber\\
\label{Eqn:P_SYXOf100}
p_{SYX}(1,0,0)&=& \epsilon \theta\left[\beta_0z_4+\beta_1 z_5+\beta_2  z_6+\beta_{3}z_{7} \right] = (1-\tau)\epsilon \delta\\
\label{Eqn:P_SYXOf101}
p_{SYX}(1,0,1)&=& \epsilon \theta\left[1-\beta_0z_4-\beta_1 z_5-\beta_2 z_6-\beta_{3}z_{7} \right] = \tau \epsilon(1-\delta)\nonumber\\
\label{Eqn:P_SYXOf110}
p_{SYX}(1,1,0)&=& \epsilon(1-\theta) \left[\gamma_0z_4+\gamma_1 z_5+\gamma_2  z_6+\gamma_{3}z_{7} \right] = (1-\tau) \epsilon (1-\delta)\\
\label{Eqn:P_SYXOf111}
p_{SYX}(1,1,1)&=& \epsilon(1-\theta) \left[1-\gamma_0z_4-\gamma_1 z_5-\gamma_2  z_6-\gamma_{3}z_{7} \right]= \tau \epsilon \delta\nonumber
\end{eqnarray}
Since $\epsilon \notin \left\{ 0,1\right\}$, (\ref{Eqn:P_SYXOf000}),(\ref{Eqn:P_SYXOf110}) imply 
\begin{eqnarray}
 \label{Eqn:ConsequenceEquatingP_SYX000AndP_SYX110}
\beta_0z_0+\beta_1 z_1+\beta_2
  z_2+\beta_{3}z_{3} =\gamma_0z_4+\gamma_1 z_5+\gamma_2
  z_6+\gamma_{3}z_{7}=:\psi_{1}\nonumber
\end{eqnarray}
Similarly (\ref{Eqn:P_SYXOf010}),(\ref{Eqn:P_SYXOf100}) imply
\begin{eqnarray}
 \label{Eqn:ConsequenceEquatingP_SYX010AndP_SYX100}
\gamma_0z_0+\gamma_1 z_1+\gamma_2
  z_2 +\gamma_{3} z_{3} =\beta_0z_4+\beta_1 z_5+\beta_2
  z_6+\beta_{3}z_{7}=:\psi_{2}\nonumber
\end{eqnarray}
\begin{table} \begin{center}
\begin{tabular}{|c|c|c|} \hline
SYX & $p_{SYX}$ &\\
\hline
000 & $(1-\epsilon)(1-\theta)\left[\beta_0z_0+\beta_1 z_1+\beta_2
  z_2+\beta_{3}z_{3} \right]$ & $(1-\tau)(1-\epsilon)(1-\delta)$ \\
\hline
001 & $(1-\epsilon)(1-\theta)\left[1-\beta_0z_0-\beta_1 z_1-\beta_2
  z_2-\beta_{3}z_{3} \right] $ & $\tau(1-\epsilon)\delta$\\
\hline
010 & $(1-\epsilon)\theta \left[\gamma_0z_0+\gamma_1 z_1+\gamma_2
  z_2 +\gamma_{3} z_{3} \right]$ & $(1-\tau)(1-\epsilon)\delta$\\
\hline
011 & $(1-\epsilon)\theta\left[1-\gamma_0z_0-\gamma_1 z_1-\gamma_2
  z_2-\gamma_{3} z_{3} \right]$ & $\tau(1-\epsilon)(1-\delta)$\\
\hline
100 & $ \epsilon \theta\left[\beta_0z_4+\beta_1 z_5+\beta_2
  z_6+\beta_{3}z_{7} \right] $ & $(1-\tau)\epsilon \delta$\\
\hline
101 & $\epsilon \theta\left[1-\beta_0z_4-\beta_1 z_5-\beta_2
  z_6-\beta_{3}z_{7} \right]$ & $\tau \epsilon(1-\delta)$\\
\hline
110 & $\epsilon(1-\theta) \left[\gamma_0z_4+\gamma_1 z_5+\gamma_2
  z_6+\gamma_{3}z_{7} \right]$ & $(1-\tau) \epsilon (1-\delta)$\\
\hline
111 & $\epsilon(1-\theta) \left[1-\gamma_0z_4-\gamma_1 z_5-\gamma_2
  z_6-\gamma_{3}z_{7} \right]$ & $\tau \epsilon \delta$\\
\hline
\end{tabular} \end{center}
\caption{Enforcing conditions 1) and 2) for $p_{SXY}$} \label{Table:P_SXY}
\end{table}
We now argue there exists no choice of values for $z_{i}:i=0,1\cdots,7$. Towards that end, we make a couple of observations. Firstly, we argue $\psi_{1} \neq \psi_{2}$. Since $\epsilon \neq 1$ and $\theta \in (0,1)$, we have $\psi_{1}=\frac{(1-\tau)(1-\delta)}{(1-\theta)}$ and $\psi_{2} = \frac{(1-\tau)\delta}{\theta}$ from (\ref{Eqn:P_SYXOf000}) and (\ref{Eqn:P_SYXOf010}) respectively. Equating $\psi_{1}$ and $\psi_{2}$, we obtain either $\tau=1$ or $\tau=0$ or $\delta=\frac{1}{2}$. Since none of the conditions hold, we conclude $\psi_{1}\neq \psi_{2}$. Secondly, one can verify that $\psi_{1}+\psi_{2}-1=\frac{\delta (1-\delta)(1-2\tau)}{\theta(1-\theta)}$. Since $\delta \in (0,\frac{1}{2}),\theta \in (0,1)$ and $\tau \in (0,\frac{1}{2})$, $\psi_{1}+\psi_{2} > 1$.
We now eliminate the possible choices for $z_{i}:i=0,1\cdots,7$ through the following cases. let $m \define \left| \left\{ i \in \left\{ 0,1,2,3 \right\} :z_{i}=1\right\} \right|$ and $l \define \left| \left\{ i \in \left\{ 4,5,6,7 \right\} :z_{i}=1\right\} \right|$.

\noindent \textit{Case 1:} All of $z_{0},z_{1},z_{2},z_{3}$ or all of $z_{4},z_{5},z_{6},z_{7}$ are equal to 0, i.e., $m=0$ or $l=0$. This implies $\psi_{1}=\psi_{2}=0$ contradicting $\psi_{1} \neq \psi_{2}$.

\noindent \textit{Case 2:} All of $z_{0},z_{1},z_{2},z_{3}$ or all of $z_{4},z_{5},z_{6},z_{7}$ are equal to 1, i.e., $m=4$ or $l=4$. This implies $\psi_{1}=\psi_{2}=1$ contradicting $\psi_{1} \neq \psi_{2}$.

Cases 1 and 2 imply $m,l \in \left\{1,2,3\right\}$.

\noindent \textit{Case 3: $m=l=3$}. If $i_{1},i_{2},i_{3}$ are distinct indices in $\left\{ 0,1,2,3 \right\}$ such that $z_{i_{1}}=z_{i_{2}}=z_{i_{3}}=1$, then one among $z_{i_{1}+4},z_{i_{2}+4},z_{i_{3}+4}$ has to be $0$. Else $\psi_{1}=\beta_{i_{1}}+\beta_{i_{2}}+\beta_{i_{3}}$ and $\psi_{2}=\beta_{i_{1}}z_{i_{1}+4}+\beta_{i_{2}}z_{i_{2}+4}+\beta_{i_{3}}z_{i_{3}+4}=\beta_{i_{1}}+\beta_{i_{2}}+\beta_{i_{3}}=\psi_{1}$ contradicting $\psi_{1} \neq \psi_{2}$. Let us consider the case $z_{0}=z_{1}=z_{2}=1$, $z_{3}=z_{4}=0$ and $z_{5}=z_{6}=z_{7}=1$. Table \ref{Table:p_UXSY} tabulates $p_{USXY}$ for this case.
\begin{table} \begin{center}
\begin{tabular}{|c|c|c|c|} \hline
UXSY & $p_{UXSY}$&UXSY & $p_{UXSY}$\\
\hline
0000 & $(1-\epsilon)(1-\theta)\beta_{0}$&2000 & $(1-\epsilon)(1-\theta)\beta_{2}$ \\
\hline
0001 & $(1-\epsilon)\theta\beta_{3}$&2001 & $(1-\epsilon)\theta \gamma_{2}$ \\
\hline
0110 & $\epsilon\theta\beta_{0}$&2010 & $\epsilon\theta\beta_{2}$ \\
\hline
0111 & $\epsilon(1-\theta) \beta_{3}$ &2011 & $\epsilon(1-\theta)\gamma_{2}$ \\
\hline
1000 & $(1-\epsilon)(1-\theta)\beta_{1}$ &3100 & $(1-\epsilon)(1-\theta)\beta_{3}$ \\
\hline
1001 & $(1-\epsilon)\theta \gamma_{1}$ &3101 & $(1-\epsilon)\theta \beta_{0}$ \\
\hline
1010 & $\epsilon\theta\beta_{1}$ &3010 & $\epsilon\theta\beta_{3}$ \\
\hline
1011 & $\epsilon(1-\theta) \gamma_{1}$ &3011 & $\epsilon(1-\theta)\beta_{0}$ \\
\hline
\end{tabular} \end{center}\caption{$p_{UXSY}$} \label{Table:p_UXSY} 
\end{table}
We have $\psi_{1}=\beta_{0}+\beta_{1}+\beta_{2}=\gamma_{1}+\gamma_{2}+\gamma_{3}$ or equivalently $\psi_{1}=1-\beta_{3}=1-\gamma_{0}$ and $\psi_{2}=\gamma_{0}+\gamma_{1}+\gamma_{2}=\beta_{1}+\beta_{2}+\beta_{3}$ or equivalently $\psi_{2}=1-\gamma_{3}=1-\beta_{0}$. These imply $\gamma_{3}=\beta_{0}$, $\gamma_{0}=\beta_{3}$ which further imply $\gamma_{1}+\gamma_{2}=\beta_{1}+\beta_{2}$ (since $1=\gamma_{0}+\gamma_{1}+\gamma_{2}+\gamma_{3}=\beta_{0}+\beta_{1}+\beta_{2+}\beta_{3}$). From table \ref{Table:p_UXSY}, one can verify 
\begin{eqnarray}
\label{Eqn:P_UGivenXSYOf0Given001}
&p_{U|XSY}(0|0,0,1)=\frac{\beta_{3}(1-\epsilon)\theta}{(1-\epsilon)\theta(\beta_{3}+\gamma_{1}+\gamma_{2})} = \frac{\beta_{3}}{\beta_{1}+\beta_{2}+\beta_{3}},
\nonumber\\
\lefteqn{p_{U|XS}(0|0,0)=\frac{(1-\theta)\beta_{0}+\theta\beta_{3}}{(1-\theta)(\beta_{0}+\beta_{1}+\beta_{2})+\theta(\beta_{3}+\gamma_{1}+\gamma_{2})}}.\nonumber
\end{eqnarray}
The Markov chain $U-(X,S)-Y$ implies $p_{U|XSY}(0|0,0,1)=p_{U|XS}(0|0,0)$. Equating the right hand sides of the above equations, we obtain $(1-\theta)(\beta_{0}-\beta_{3})(\beta_{1}+\beta_{2})=0$. Since $\theta \neq 0$, $\beta_{1}+\beta_{2}=0$ or $\beta_{0}=\beta_{3}$. If $\beta_{0}=\beta_{3}$, then $1-\beta_{3}=\psi_{1}=\psi_{2}=1-\beta_{0}$ thus contradicting $\psi_{1} \neq \psi_{2}$. If $\beta_{1}+\beta_{2}=0$, then $\beta_{0}+\beta_{3}=1$ implying $\psi_{1}+\psi_{2}=1$ contradicting $\psi_{1}+\psi_{2}>1$.

\noindent \textit{Case 4: $m=3, l=2$}. Let us assume $z_{0}=z_{1}=z_{2}=z_{6}=z_{7}=1,z_{3}=z_{4}=z_{5}=0$. We then have $\psi_{1}=\beta_{0}+\beta_{1}+\beta_{2}=\gamma_{2}+\gamma_{3}$ and $\psi_{2}=\gamma_{0}+\gamma_{1}+\gamma_{2}=\beta_{2}+\beta_{3}$. Since $\beta_{0}+\beta_{1}+\beta_{2}=1-\beta_{3}$ and $\gamma_{0}+\gamma_{1}+\gamma_{2}=1-\gamma_{3}$, we have $\gamma_{2}+\gamma_{3}=1-\beta_{3}$ and $\beta_{2}+\beta_{3}=1-\gamma_{3}$ and therefore $\gamma_{2}=\beta_{2}$.Table \ref{Table:p_UXSYCasemIs3lIs2} tabulates $p_{USXY}$ for this case.
\begin{table} \begin{center}
\begin{tabular}{|c|c|c|c|} \hline
UXSY & $p_{UXSY}$&UXSY & $p_{UXSY}$\\
\hline
0000 & $(1-\epsilon)(1-\theta)\beta_{0}$&2000 & $(1-\epsilon)(1-\theta)\beta_{2}$ \\
\hline
0001 & $(1-\epsilon)\theta\gamma_{0}$&2001 & $(1-\epsilon)\theta \gamma_{2}$ \\
\hline
0110 & $\epsilon\theta\beta_{0}$&2010 & $\epsilon\theta\beta_{2}$ \\
\hline
0111 & $\epsilon(1-\theta) \gamma_{0}$ &2011 & $\epsilon(1-\theta)\gamma_{2}$ \\
\hline
1000 & $(1-\epsilon)(1-\theta)\beta_{1}$ &3100 & $(1-\epsilon)(1-\theta)\beta_{3}$ \\
\hline
1001 & $(1-\epsilon)\theta \gamma_{1}$ &3101 & $(1-\epsilon)\theta \gamma_{3}$ \\
\hline
1110 & $\epsilon\theta\beta_{1}$ &3010 & $\epsilon\theta\beta_{3}$ \\
\hline
1111 & $\epsilon(1-\theta) \gamma_{1}$ &3011 & $\epsilon(1-\theta)\gamma_{3}$ \\
\hline
\end{tabular} \end{center}\caption{$p_{UXSY}$} \label{Table:p_UXSYCasemIs3lIs2} 
\end{table}
From table \ref{Table:p_UXSYCasemIs3lIs2}, one can verify 
\begin{eqnarray}
\label{Eqn:P_UGivenXSYOf0Given001}
&p_{U|XSY}(2|0,0,1)=\frac{\beta_{2}(1-\epsilon)\theta}{(1-\epsilon)\theta(\beta_{2}+\gamma_{0}+\gamma_{1})} = \frac{\beta_{2}}{\beta_{2}+\gamma_{0}+\gamma_{1}},
\nonumber\\
\lefteqn{p_{U|XS}(2|0,0)=\frac{\beta_{2}}{(1-\theta)(\beta_{0}+\beta_{1})+\theta(\gamma_{0}+\gamma_{1})+\beta_{2}}}.\nonumber
\end{eqnarray}
The Markov chain $U-(X,S)-Y$ implies $p_{U|XSY}(2|0,0,1)=p_{U|XS}(2|0,0)$. Equating the RHS of the above equations, we obtain $\beta_{0}+\beta_{1}=\gamma_{0}+\gamma_{1}$. This implies $\beta_{2}+\beta_{3}=\gamma_{2}+\gamma_{3}$. However $\psi_{2}=\beta_{2}+\beta_{3}$ and $\psi_{1}=\gamma_{2}+\gamma_{3}$, this contradicting $\psi \neq \psi_{2}$.

Let us assume $z_{0}=z_{1}=z_{2}=z_{5}=z_{6}=1$ and $z_{3}=z_{4}=z_{7}=0$. It can be verified that $\psi_{1}=\beta_{0}+\beta_{1}+\beta_{2}=\gamma_{1}+\gamma_{2}$ and $\psi_{2}=\gamma_{0}+\gamma_{1}+\gamma_{2}=\beta_{1}+\beta_{2}$. This implies $\psi_{1}-\psi_{2}=\beta_{0}=-\gamma_{0}$. Since $\beta_{0}$ and $\gamma_{0}$ are non-negative, $\beta_{0}=\gamma_{0}=0$ implying $\psi_{1}-\psi_{2}=0$, contradicting $\psi_{1} \neq \psi_{2}$.

\noindent \textit{Case 5: $m=3, l=1$}. Assume $z_{0}=z_{1}=z_{2}=z_{4}=1$, $z_{3}=z_{5}=z_{6}=z_{7}=0$. It can be verified that $\psi_{1}=\beta_{0}+\beta_{1}+\beta_{2}=\gamma_{0}$ and $\psi_{2}=\gamma_{0}+\gamma_{1}+\gamma_{2}=\beta_{0}$. Therefore $\psi_{1}-\psi_{2}=\beta_{1}+\beta_{2}$ and $\psi_{2}-\psi_{1}=\gamma_{1}+\gamma_{2}$. Since $\beta_{i},\gamma_{i}:i \in \left\{ 0,1,2,3 \right\}$ are non-negative, $\psi_{1}-\psi_{2} \geq 0$ and $\psi_{2}-\psi_{1} \geq 0$ contradicting $\psi_{1} \neq \psi_{2}$.

Assume $z_{0}=z_{1}=z_{2}=z_{7}=1$ and $z_{3}=z_{4}=z_{5}=z_{6}=0$. In this case, $\psi_{1}=\beta_{0}+\beta_{1}+\beta_{2}=\gamma_{3}$, $\psi_{2}=\gamma_{0}+\gamma_{1}+\gamma_{2}=1-\gamma_{3}$. We have $\psi_{1}+\psi_{2}=1$ contradicting $\psi_{1}+\psi_{2}>1$.

\noindent \textit{Case 6: $m=2, l=2$}. Assume $z_{0}=z_{1}=z_{4}=z_{5}=1$, $z_{2}=z_{3}=z_{6}=z_{7}=0$. Note that $\psi_{1}=\beta_{0}+\beta_{1}=\gamma_{0}+\gamma_{1}$, $\psi_{2}=\gamma_{0}+\gamma_{1}=\beta_{0}+\beta_{1}$ contradicting $\psi_{1} \neq \psi_{2}$.

Assume $z_{0}=z_{1}=z_{6}=z_{7}=1$, $z_{2}=z_{3}=z_{4}=z_{5}=0$. Note that $\psi_{1}=\beta_{0}+\beta_{1}=\gamma_{2}+\gamma_{3}$, $\psi_{2}=\gamma_{0}+\gamma_{1}=\beta_{2}+\beta_{3}$ contradicting $\psi_{1} + \psi_{2}>1$.

Assume $z_{0}=z_{1}=z_{5}=z_{6}=1$, $z_{2}=z_{3}=z_{4}=z_{7}=0$. Note that $\psi_{1}=\beta_{0}+\beta_{1}=\gamma_{1}+\gamma_{2}$, $\psi_{2}=\gamma_{0}+\gamma_{1}=\beta_{1}+\beta_{2}$ and therefore $\beta_{2}+\beta_{3}=\gamma_{0}+\gamma_{3}$ and $\beta_{0}+\beta_{3}=\gamma_{2}+\gamma_{3}$. We observe
\begin{equation}
 \label{Eqn:CasemIs2lIs2Psi1-Psi2}
\psi_{1}-\psi_{2}= \beta_{0}-\beta_{2}=\gamma_{2}-\gamma_{0}.
\end{equation}

PMF $p_{UXSY}$ is tabulated in \ref{Table:p_UXSYCasemIs2lIs2} for this case. Table \ref{Table:p_UXSYCasemIs2lIs2} enables us to compute conditional PMF $p_{U|XSY}$ which is tabulated in table \ref{Table:p_UGivenXSYCasemIs2lIs2}.
\begin{table} \begin{center}
\begin{tabular}{|c|c|c|c|} \hline
UXSY & $p_{UXSY}$&UXSY & $p_{UXSY}$\\
\hline
0000 & $(1-\epsilon)(1-\theta)\beta_{0}$&2100 & $(1-\epsilon)(1-\theta)\beta_{2}$ \\
\hline
0001 & $(1-\epsilon)\theta\gamma_{0}$&2101 & $(1-\epsilon)\theta \gamma_{2}$ \\
\hline
0110 & $\epsilon\theta\beta_{0}$&2010 & $\epsilon\theta\beta_{2}$ \\
\hline
0111 & $\epsilon(1-\theta) \gamma_{0}$ &2011 & $\epsilon(1-\theta)\gamma_{2}$ \\
\hline
1000 & $(1-\epsilon)(1-\theta)\beta_{1}$ &3100 & $(1-\epsilon)(1-\theta)\beta_{3}$ \\
\hline
1001 & $(1-\epsilon)\theta \gamma_{1}$ &3101 & $(1-\epsilon)\theta \gamma_{3}$ \\
\hline
1010 & $\epsilon\theta\beta_{1}$ &3110 & $\epsilon\theta\beta_{3}$ \\
\hline
1011 & $\epsilon(1-\theta) \gamma_{1}$ &3111 & $\epsilon(1-\theta)\gamma_{3}$ \\
\hline
\end{tabular} \end{center}\caption{$p_{UXSY}$} \label{Table:p_UXSYCasemIs2lIs2} 
\end{table}
\begin{table} \begin{center}
\begin{tabular}{|c|c|c|c|} \hline
UXSY & $p_{U|XSY}$&UXSY & $p_{U|XSY}$\\
\hline
0000 & $\frac{\beta_{0}}{\beta_{0}+\beta_{1}}$&0001 & $\frac{\gamma_{0}}{\gamma_{0}+\gamma_{1}}$ \\
\hline
0110 & $\frac{\beta_{0}}{\beta_{0}+\beta_{3}}$&0111 & $\frac{\gamma_{0}}{\gamma_{0}+\gamma_{3}}$ \\
\hline
1000 & $\frac{\beta_{1}}{\beta_{0}+\beta_{1}}$ &1001 & $\frac{\gamma_{1}}{\gamma_{0}+\gamma_{1}}$ \\
\hline
1010 & $\frac{\beta_{1}}{\beta_{1}+\beta_{2}}$ &1011 & $\frac{\gamma_{1}}{\gamma_{1}+\gamma_{2}}$ \\
\hline
 2100& $\frac{\beta_{2}}{\beta_{2}+\beta_{3}}$&2101 & $\frac{\gamma_{2}}{\gamma_{2}+\gamma_{3}}$ \\
\hline
 2010& $\frac{\beta_{2}}{\beta_{1}+\beta_{2}}$ &2011 & $\frac{\gamma_{2}}{\gamma_{1}+\gamma_{2}}$ \\
\hline
 3100& $\frac{\beta_{3}}{\beta_{2}+\beta_{3}}$ &3101 & $\frac{\gamma_{3}}{\gamma_{2}+\gamma_{3}}$ \\
\hline
3110 & $\frac{\beta_{3}}{\beta_{0}+\beta_{3}}$ &3111 & $\frac{\gamma_{3}}{\gamma_{0}+\gamma_{3}}$ \\
\hline
\end{tabular} \end{center}\caption{$p_{U|XSY}$} \label{Table:p_UGivenXSYCasemIs2lIs2} 
 \end{table}
Markov chain $U-(X,S)-Y$ implies columns 2 and 4 of table \ref{Table:p_UGivenXSYCasemIs2lIs2} are identical. This implies
\begin{eqnarray}
 \label{Eqn:Relation1ComingOutOfU-XS-YMarkovChain}
\frac{\beta_{0}}{\gamma_{0}}\overset{(a)}{=}\frac{\beta_{0}+\beta_{1}}{\gamma_{0}+\gamma_{1}}\overset{(b)}{=}\frac{\beta_{1}}{\gamma_{1}}, \frac{\beta_{2}}{\gamma_{2}}\overset{(c)}{=}\frac{\beta_{2}+\beta_{3}}{\gamma_{2}+\gamma_{3}}\overset{(d)}{=}\frac{\beta_{3}}{\gamma_{3}}, ~~\mbox{ and }~~ \frac{\beta_{0}}{\gamma_{0}}\overset{(e)}{=}\frac{\beta_{0}+\beta_{3}}{\gamma_{0}+\gamma_{3}}\overset{(f)}{=}\frac{\beta_{3}}{\gamma_{3}},
\end{eqnarray}
where (a),(b),(c),(d) in (\ref{Eqn:Relation1ComingOutOfU-XS-YMarkovChain}) is obtained by equating rows 1, 3, 5, 7 of columns 2 and 4 respectively and (e) and (f) in (\ref{Eqn:Relation1ComingOutOfU-XS-YMarkovChain}) are obtained by equating rows 2 and 8 of columns 2 and 4 respectively. (\ref{Eqn:Relation1ComingOutOfU-XS-YMarkovChain}), enables us to conclude
\begin{equation}
 \label{Eqn:BetaIsEqualToGamma}
\frac{\beta_{0}}{\gamma_{0}}=\frac{\beta_{1}}{\gamma_{1}}=\frac{\beta_{2}}{\gamma_{2}}=\frac{\beta_{3}}{\gamma_{3}}. \nonumber
\end{equation}
Since $\beta_{0}+\beta_{1}+\beta_{2}+\beta_{3}=\gamma_{0}+\gamma_{1}+\gamma_{2}+\gamma_{3}=1$, we have $\beta_{i}=\gamma_{i}$ for each $i \in \left\{ 0,1,2,3\right\}$ which yields $\psi_{1}=\psi_{2}$ in (\ref{Eqn:CasemIs2lIs2Psi1-Psi2}) contradicting $\psi_{1} \neq \psi_{2}$.

\noindent \textit{Case 7: $m=2, l=1$}. Assume $z_{0}=z_{1}=z_{4}=1, z_{2}=z_{3}=z_{5}=z_{6}=z_{7}=0$. Note that $\psi_{1}=\beta_{0}+\beta_{1}=\gamma_{0}, \psi_{2}= \gamma_{0}+\gamma_{1}=\beta_{0}$ and hence $\psi_{1}-\psi_{2}=\beta_{1}$ and $\psi_{2}-\psi_{1}=\gamma_{1}$. Since $\gamma_{1}$ and $\beta_{1}$ are non-negative, we have $\psi_{1}=\psi_{2}$ contradicting $\psi_{1} \neq \psi_{2}$.

Assume $z_{0}=z_{1}=z_{7}=1, z_{2}=z_{3}=z_{4}=z_{5}=z_{6}=0$. Note that $\psi_{1}=\beta_{0}+\beta_{1}=\gamma_{3}, \psi_{2}=\gamma_{0}+\gamma_{1}=\beta_{3}$ and hence $\psi_{1}+\psi_{2} = \beta_{0}+\beta_{1}+\beta_{3} \leq 1$ contradicting $\psi_{1}+\psi_{2} >1$.

\noindent \textit{Case 6: $m=1, l=1$}. Assume $z_{0}=z_{4}=1, z_{1}=z_{2}=z_{3}=z_{5}=z_{6}=z_{7}=0$. Note that $\psi_{1}=\beta_{0}=\gamma_{0},\psi_{2}=\gamma_{0}=\beta_{0}$, thus contradicting $\psi_{1} \neq \psi_{2}$.

Assume $z_{0}=z_{5}=1, z_{1}=z_{2}=z_{3}=z_{4}=z_{6}=z_{7}=0$. Note that $\psi_{1}=\beta_{0}=\gamma_{1},\psi_{2}=\gamma_{0}=\beta_{1}$, and hence $\psi_{1}+\psi_{2} = \beta_{0}+\beta_{1} \leq 1$, thus contradicting $\psi_{1} + \psi_{2} >1$.
\end{proof}

\section{Proof of lemma \ref{Lem:ForBSCMarokovChainWithOutputImpliesMarkovChainWithInput}}
\label{AppSec:ForBSCMarokovChainWithOutputImpliesMarkovChainWithInput}
Since $A-B-Y$ and $AB-X-Y$ are Markov chains, to prove $A-B-XY$ is a Markov chain, it suffices to prove $A-B-X$ is a Markov chain. We therefore need to prove $p_{XA|B}(x_{k},a_{i}|b_{j})=p_{X|B}(x_{k}|b_{j})p_{A|B}(a_{i}|b_{j})$ for every $(x_{k},a_{i},b_{j}) \in \left\{ 0,1\right\}\times \mathcal{A} \times \mathcal{B}$ such that $p_{B}(b_{j})>0$. It suffices to prove $p_{XA|B}(0,a_{i}|b_{j})=p_{X|B}(0|b_{j})p_{A|B}(a_{i}|b_{j})$ for every $(a_{i},b_{j}) \in \mathcal{A} \times \mathcal{B}$ such that $p_{B}(b_{j})>0$.\footnote{Indeed, $p_{XA|B}(1,a_{i}|b_{j})=p_{A|B}(a_{i}|b_{j})-p_{XA|B}(0,a_{i}|b_{j})=p_{A|B}(a_{i}|b_{j})(1-p_{X|B}(0|b_{j}))=p_{A|B}(a_{i}|b_{j})p_{X|B}(1|b_{j})$.}

Fix a $b_{j}$ for which $p_{B}(b_{j})>0$. Let $p_{A|B}(a_{i}|b_{j})=\alpha_{i}$ for each $i \in \naturals$ and $p_{XA|B}(0,a_{i}|b_{j})=\chi_{i}$ for each $(i,j) \in \naturals\times \naturals$. It can be verified $p_{XYA|B}(\cdot,\cdot,\cdot|b_{j})$ is as in table \ref{Table:p_XYAGivenB}. From table \ref{Table:p_XYAGivenB}, we infer $p_{AY|B}(a_{i}0|b_{j})=\chi_{i}(1-\eta)+(\alpha_{i}-\chi_{i})\eta=\alpha_{i}\eta+\chi_{i}(1-2\eta)$. From the Markov chain $A-B-Y$, we have $p_{AY|B}(a_{i}0|b_{j})=p_{A|B}(a_{i}|b_{j})p_{Y|B}(0|b_{j})= \alpha_{i}p_{Y|B}(0|b_{j})$. Therefore, $ \alpha_{i}p_{Y|B}(0|b_{j})=\alpha_{i}\eta+\chi_{i}(1-2\eta)$. Since $1-2\eta \neq 0$, we substitute for $\chi_{i}$ and $\alpha_{i}$ in terms of their definitions to conclude
\begin{equation}
 \label{Eqn:ProbOfXAGivenB}
p_{XA|B}(0,a_{i}|b_{j})=\chi_{i}= \alpha_{i}\cdot \frac{p_{Y|B}(0|b_{j})-\eta}{1-2\eta}=p_{A|B}(a_{i}|b_{j})\frac{p_{Y|B}(0|b_{j})-\eta}{1-2\eta}.\nonumber
\end{equation}
Since $\frac{p_{Y|B}(0|b_{j})-\eta}{1-2\eta}$ is independent of $i$ and $b_{j}$ was an arbitrary element in $\mathcal{B}$ that satisfies $p_{B}(b_{j})>0$, we have established Markov chain $A-B-X$.
\begin{table} \begin{center}
\begin{tabular}{|c|c|c|c|c|c|c|c|} \hline
AXY & $p_{AXY|B}(\cdot,\cdot,\cdot|b_{j})$&AXY & $p_{AXY|B}(\cdot,\cdot,\cdot|b_{j})$&AXY & $p_{AXY|B}(\cdot,\cdot,\cdot|b_{j})$&AXY & $p_{AXY|B}(\cdot,\cdot,\cdot|b_{j})$\\
\hline
$a_{i}00$& $\chi_{i}(1-\eta)$&$a_{i}01$ & $\chi_{i}\eta$&$a_{i}10$ & $(\alpha_{i}-\chi_{i})\eta$&$a_{i}11$ & $(\alpha_{i}-\chi_{i})(1-\eta)$ \\
\hline
\end{tabular} \end{center}\caption{$p_{AXY|B}(\cdot,\cdot,\cdot|b_{j})$} \label{Table:p_XYAGivenB} 
\end{table}
\section{Upper bound on Marton's coding technique for example \ref{Ex:A3To1-OR-BC}}
\label{AppSec:UpperBoundMartonCodingTechniuque}
We begin with a characterization of a test channel
$p_{\TimeSharingRV\PublicRV\underlineSemiPrivateRV\underlinePrivateRV
  X\underlineY}$ for which $(R_{1},h_{b}(\tau_{2} *
\delta_{2})-h_{b}(\delta_{2}),h_{b}(\tau_{3} *
\delta_{3})-h_{b}(\delta_{3})) \in
\alpha_{\mathscr{U}}(p_{\TimeSharingRV\PublicRV\underlineSemiPrivateRV\underlinePrivateRV
  X\underlineY})$. Since independent information needs to be
communicated to users $2$ and $3$ at their respective PTP capacities,
it is expected that their codebooks are not precoded for each other's
signal, and moreover none of users $2$ and $3$ decode a part of the
other users' signal. The following lemma establishes this. 
We remind the reader that $X_{1}X_{2}X_{3}=X$ denote the three binary
digits at the input.

\begin{lemma}
 \label{Lem:CharTstChnlThatContainsTheRtTrpl}
If there exists a test channel $p_{\TimeSharingRV\PublicRV\underlineSemiPrivateRV\underlinePrivateRV X\underlineY} \in \SetOfDistributions_{\mathscr{U}}(\tau)$ and nonnegative numbers $K_{i},S_{ij},K_{ij},L_{ij},S_{i},T_{i}$ that satisfy (\ref{Eqn:3BCSourceCodingBoundNonnegativity})-(\ref{Eqn:3BCChannelCodingSextupleBound}) for each triple $(i,j,k) \in \left\{ (1,2,3),(2,3,1),(3,1,2) \right\}$ such that $R_{2}=K_{2}+K_{23}+L_{12}+T_{2}=h_{b}(\tau_{2}*\delta_{2})-h_{b}(\delta_{2}),R_{3}=K_{3}+K_{31}+L_{23}+T_{3}=h_{b}(\tau_{3}*\delta_{3})-h_{b}(\delta_{3})$, then
\begin{enumerate}
 \item \label{Item:RatesOfPublicAndSemiPrivateCodebooks}$K_1=K_2=K_3=K_{23}=L_{23}=K_{12}=L_{31}=S_2=S_3=0$ and $I(\SemiPrivateRV_{31} V_{1}V_{3};Y_{2}|\TimeSharingRV WU_{23}U_{12}V_{2})=0$,
\item \label{Item:BinningRatesOfSemiPrivateCodebooks}$S_{31}=I(U_{31};U_{23}|\TimeSharingRV W),S_{12}=I(U_{12};U_{23}|\TimeSharingRV W)$, $S_{23}=I(U_{12};U_{31}|\TimeSharingRV WU_{23})=0$,
\item \label{Item:WU23IsConditionallyIndependentOfY2Y3GivenQ}$I(V_{2}U_{12};V_{3}U_{31}|\TimeSharingRV WU_{23})=0$, $I(WU_{23};Y_{j}|\TimeSharingRV )=0:j=2,3,I(V_{2}U_{12};Y_{2}|\TimeSharingRV W U_{23})=h_{b}(\tau_{2}*\delta_{2})-h_{b}(\delta_{2})$ and $I(V_{3}U_{31};Y_{3}|\TimeSharingRV W U_{23})=h_{b}(\tau_{3}*\delta_{3})-h_{b}(\delta_{3})$, $p_{X_{j}|QWU_{23}}(1|q,w,u_{23})=\tau_{j}$ for $j=2,3$.
\item \label{ItemNumber:MarkovChains} $(V_3,X_3,V_1,U_{31}) - (\TimeSharingRV WU_{23}U_{12}V_2) - (X_2,Y_2) $, $(V_2,X_2,V_1,U_{12}) - (\TimeSharingRV WU_{23}U_{31}V_3) - (X_3,Y_3)$ and $V_{1}-QW\underline{U}V_{2}V_{3}-X_{2}X_{3}$ are Markov chains,
\item \label{Item:X2AndX3AreConditionallyIndependentGivenQWU12U23U31}$X_{2} - \TimeSharingRV \PublicRV \SemiPrivateRV_{12} \SemiPrivateRV_{23} \SemiPrivateRV_{31} - X_{3}$ is a Markov chain,
\item \label{Item:U31AndX2AreConditionallyIndependentGivenQWU12U23}$\SemiPrivateRV_{12}-\TimeSharingRV \PublicRV \SemiPrivateRV_{23} \SemiPrivateRV_{31}-X_{3}$ and $\SemiPrivateRV_{31}-\TimeSharingRV \PublicRV \SemiPrivateRV_{23} \SemiPrivateRV_{12}-X_{2}$ are Markov chains.
\end{enumerate}
\end{lemma}
The proof of this lemma is similar to that of lemma \ref{Lem:CharacterizationOfTestChannelThatContainsTheRateTriple} and is therefore omitted. Lemma \ref{Lem:CharTstChnlThatContainsTheRtTrpl} enables us to simplify the bounds (\ref{Eqn:3BCSourceCodingBoundNonnegativity})-(\ref{Eqn:3BCChannelCodingSextupleBound}) for the particular test channel under consideration. The following bounds may be verified. If $(R_{1},h_{b}(\tau_{2}*\delta_{2})-h_{b}(\delta_{2}),h_{b}(\tau_{3}*\delta_{3})-h_{b}(\delta_{3})) \in \alpha_{\mathscr{U}}(p_{\TimeSharingRV\PublicRV\underlineSemiPrivateRV\underlinePrivateRV X\underlineY})$, then there exists nonnegative numbers $S_{1},T_{1},L_{12},K_{31}$ that satisfy $R_{1}=T_{1},R_{2}=L_{12}+T_{2}=h_{b}(\tau_{2}*\delta_{2})-h_{b}(\delta_{2}),R_{3}=K_{31}+T_{3}=h_{b}(\tau_{3}*\delta_{3})-h_{b}(\delta_{3})$,
%
\begin{eqnarray}
 \label{Eqn:ORBCBndsOnUser1CdbkAsAResultOfUser1Dec}
&S_1  \geq I(V_1;U_{23}V_2V_3|\TimeSharingRV WU_{12}U_{31}),~~~T_1+S_1\leq I(V_1;Y_1|\TimeSharingRV WU_{12}U_{31})\\
\label{Eqn:ORBCBndsOnRatesOfUser123CodebooksAsAResultOfUser1Dec}
&L_{12}+K_{31}+T_1+S_1 \leq I(U_{12};U_{31}|\TimeSharingRV W)-I(U_{23};U_{12}|\TimeSharingRV W)+I(V_1U_{12}U_{31};Y_1|\TimeSharingRV W)-I(U_{23};U_{31}|\TimeSharingRV W)\\
\label{Eqn:ORBCBndsOnRatesOfUser2CodebookAsAResultOfUser2Dec}
&0 \leq T_{2} \leq I(V_2;Y_2|\TimeSharingRV WU_{12}U_{23}),~~~h_{b}(\delta_{2}*\tau_{2})-h_{b}(\delta_{2})=T_{2}+L_{12} = I(U_{12}V_2;Y_2|\TimeSharingRV WU_{23})\\
\label{Eqn:ORBCBndsOnRatesOfUser3CodebookAsAResultOfUser3Dec}
&0 \leq T_{3}\leq I(V_3;Y_3|\TimeSharingRV WU_{31}U_{23}),~~~h_{b}(\delta_{3}*\tau_{3})-h_{b}(\delta_{3})=T_{3}+K_{31}= I(U_{31}V_3;Y_3|\TimeSharingRV WU_{23}).
\end{eqnarray}
Following arguments similar to section \ref{Sec:StrictSubOptimalityOfMartonCodingTechnique}, we obtain
\begin{eqnarray}
\label{Eqn:ORBCUpBndOnRateOfUser1ThatContainsRateLoss}
R_{1} \leq  I(V_1;Y_1U_{23}|\TimeSharingRV W\SemiPrivateRV_{12}\SemiPrivateRV_{31})
 -I(V_1;U_{23}V_2V_3|\TimeSharingRV WU_{12}U_{31})=I(V_1;Y_1|\TimeSharingRV W\underlineSemiPrivateRV)
 -I(V_1;V_2V_3|\TimeSharingRV W\underlineSemiPrivateRV),\\
\label{Eqn:ORBCBndOnUser1RateAsAConsqDecoding3Codebooks}
R_1 \leq\!I(V_1;Y_1|\TimeSharingRV W\underlineSemiPrivateRV) \!-\!I(V_1;V_2V_3|\TimeSharingRV W\underlineSemiPrivateRV)\! +\!I(U_{12}U_{31};Y_1|\TimeSharingRV W U_{23}) \!-\!I(U_{12};Y_2|\TimeSharingRV WU_{23})\!-\!I(U_{31};Y_3|\TimeSharingRV WU_{23}),\\
R_{1}+R_{2}+R_{3} \leq I(V_{2};Y_{2}|QWU_{12}U_{23})+I(V_{3};Y_{3}|QWU_{31}U_{23}) +I(U_{12};U_{31}|\TimeSharingRV W)-I(U_{23};U_{12}|\TimeSharingRV W)\nonumber\\\label{Eqn:BoundOnR1PlusR2PlusR3}+I(V_1U_{12}U_{31};Y_1|\TimeSharingRV W)-I(U_{23};U_{31}|\TimeSharingRV W)
\end{eqnarray}
The bound (\ref{Eqn:BoundOnR1PlusR2PlusR3}) is obtained by (i) adding bounds (\ref{Eqn:ORBCBndsOnRatesOfUser123CodebooksAsAResultOfUser1Dec}) and the bounds on $T_{2}$ and $T_{3}$ present in (\ref{Eqn:ORBCBndsOnRatesOfUser2CodebookAsAResultOfUser2Dec}) and (\ref{Eqn:ORBCBndsOnRatesOfUser3CodebookAsAResultOfUser3Dec}) respectively, and (ii) identifying $T_{2}+L_{12} = R_{2}, T_{3}+K_{31}=R_{3}, T_{1} = R_{1}$ and (iii) employing the lower bound on $S_{1}$ found in (\ref{Eqn:ORBCBndsOnUser1CdbkAsAResultOfUser1Dec}). Combining (\ref{Eqn:ORBCUpBndOnRateOfUser1ThatContainsRateLoss}) and (\ref{Eqn:ORBCBndOnUser1RateAsAConsqDecoding3Codebooks}), we have
\begin{equation}
 \label{Eqn:3To1DBCAnalyticalConvTwoUpperBoundsCombined}
R_{1} \leq I(V_1;Y_1|\TimeSharingRV W\underlineSemiPrivateRV)  -I(V_1;V_2V_3|\TimeSharingRV W\underlineSemiPrivateRV) + \min \left\{ \begin{array}{c} 0,I(U_{12}U_{31};Y_1|\TimeSharingRV W U_{23}) -I(U_{12};Y_2|\TimeSharingRV WU_{23})\\-I(U_{31};Y_3|\TimeSharingRV WU_{23})                                                                                                                                     \end{array}
  \right\}.
\end{equation}
From (\ref{Eqn:ORBCUpBndOnRateOfUser1ThatContainsRateLoss}) and the Markov chain $V_{1}-QW\underline{U}V_{2}V_{3}-X_{2}X_{3}$ proved in lemma \ref{Lem:CharTstChnlThatContainsTheRtTrpl}, it can be verified that
\begin{eqnarray}
\label{Eqn:ExtractingBoundOnR1}
R_1 &\leq& I(V_1;Y_1|\TimeSharingRV W\underlineSemiPrivateRV) -I(V_1;V_2V_3|\TimeSharingRV W\underlineSemiPrivateRV) =I(V_1;Y_1|\TimeSharingRV W\underlineSemiPrivateRV) -I(V_1;V_2V_3X_2X_3|\TimeSharingRV W\underlineSemiPrivateRV) \\
\label{Eqn:RateLossBound}
&\leq& I(V_1;Y_1|\TimeSharingRV W\underlineSemiPrivateRV) -
I(V_1;X_2,X_3|\TimeSharingRV W\underlineSemiPrivateRV) \leq I(V_1;Y_1
|\TimeSharingRV W\underlineSemiPrivateRV) -I(V_1;X_2\vee
X_3|\TimeSharingRV W\underlineSemiPrivateRV)\\ 
&\leq& I(V_1;Y_1,X_{2}\vee X_{3} |\TimeSharingRV W\underlineSemiPrivateRV) -I(V_1;X_2\vee X_3|\TimeSharingRV W\underlineSemiPrivateRV) = I(V_1;Y_1 |\TimeSharingRV W\underlineSemiPrivateRV,X_{2}\vee X_{3}) \nonumber\\
\label{Eqn:RateLossBoundLastIneq}
&\leq&H(X_{1}\oplus N_{1}|Q,W,\underline{U},X_{2}\vee X_{3})-h_{b}(\delta_{1}) \leq h_{b}(\tau_{1}*\delta_{1})-h_{b}(\delta_{1})
\end{eqnarray}
with equality above if and only if $p_{X_{1}|Q,W,\underline{U},X_{2}\vee X_{3}}(1|q,w,\underline{u},x)=\tau_{1}$ and $p_{X_{2}\vee X_{3}|Q,W, \underline{U}}(x|q,w,\underline{u}) \in \{ 0,1 \}$ for all $(q,w,\underline{u},x)$ with positive probability. Note that this follows from lemma \ref{lem:GP}. Using (\ref{Eqn:BoundOnR1PlusR2PlusR3}), we now show that $H(V_{2}|QWU_{23}U_{12}) > 0$ or $H(V_{3}|QWU_{23}U_{31}) > 0$. We prove this by contradiction. Suppose $H(V_{2}|QWU_{23}U_{12}) =H(V_{3}|QWU_{23}U_{31}) = 0$, then one can substitute this in the right hand side of (\ref{Eqn:BoundOnR1PlusR2PlusR3}) to obtain the same to be $h_{b}(\tau_{1}*\beta)-h_{b}(\delta_{1})$. The left hand side of (\ref{Eqn:BoundOnR1PlusR2PlusR3}) being $R_{1}+R_{2}+R_{3}$, this condition violates the hypothesis (\ref{Eqn:3To1-OR-ICConditionForStrictSubOptimalityOfMarton}) if $R_{j}=h_{b}(\delta_{j}*\tau_{j})-h_{b}(\delta_{j})$. We therefore have $H(V_{2}|QWU_{23}U_{12}) > 0$ or $H(V_{3}|QWU_{23}U_{31}) > 0$.

Using the Markov chains $U_{31}-QWU_{23}U_{12}-V_{2}$, $U_{12}-QWU_{23}U_{31}-V_{3}$, $QWU_{23}-U_{12}V_{2}-Y_{2}$, $QWU_{23}-U_{31}V_{3}-Y_{3}$, $QW\underline{U}V_{2}-X_{2}-Y_{2}$ and $QW\underline{U}V_{3}-X_{3}-Y_{3}$ proved in lemma \ref{Lem:CharTstChnlThatContainsTheRtTrpl} and standard information theoretic arguments\footnote{These arguments are illustrated in \cite[Proof of fourth claim, Appendix B]{201403arXiv_PadPra} for an analogous setting therein.}, it can be verified that $H(X_{2} \vee X_{3}|Q,W,\underline{U}) > 0$. Referring back to the condition for equality in the inequalities (\ref{Eqn:RateLossBound}) -(\ref{Eqn:RateLossBoundLastIneq}), we conclude $R_{1} < h_{b}(\tau_{1}*\delta_{1})-h_{b}(\delta_{1})$.

We now appeal to the bound (\ref{Eqn:RateLossBound}) containing the rate loss. Clearly lemma \ref{lem:GP} proves that the above condition implies $R_{1} < h_{b}(\tau_{1}*\delta_{1})-h_{b}(\delta_{1})$. This concludes the proof.

\bibliographystyle{IEEEtran}
{
\bibliography{wisl}
}

\end{document}